\documentclass[10pt,journal,compsoc]{IEEEtran}
\IEEEoverridecommandlockouts
\usepackage{amsmath,amssymb,amsfonts,amsthm,mathtools}
\usepackage{graphicx}
\usepackage[nocompress]{cite}
\usepackage{url}
\usepackage{booktabs,multirow,array,makecell}
\usepackage{algorithm}
\usepackage{algpseudocode}
\usepackage{xcolor}
\usepackage{tikz}
\usetikzlibrary{arrows.meta,positioning,shapes,fit,calc,backgrounds}
\usepackage[hidelinks]{hyperref}
\usepackage[capitalize,nameinlink]{cleveref}
\usepackage{balance}
\usepackage[final]{microtype}

\newtheoremstyle{tplain}{3pt}{3pt}{\itshape}{}{\bfseries}{.}{.5em}{}
\newtheoremstyle{tdef}{3pt}{3pt}{\normalfont}{}{\bfseries}{.}{.5em}{}
\theoremstyle{tplain}
\newtheorem{theorem}{Theorem}
\newtheorem{lemma}[theorem]{Lemma}
\newtheorem{proposition}[theorem]{Proposition}
\newtheorem{corollary}[theorem]{Corollary}
\theoremstyle{tdef}
\newtheorem{definition}{Definition}
\newtheorem{assumption}{Assumption}

\theoremstyle{remark}

\crefname{assumption}{Assumption}{Assumptions}

\newcommand{\supp}[1]{Appendix~#1}
\newcommand{\negl}{\mathsf{negl}}
\newcommand{\Adv}{\mathsf{Adv}}
\newcommand{\Ws}{W_{\mathsf{s}}}
\newcommand{\Wp}{W_{\mathsf{p}}}
\newcommand{\Rsig}{R_{\mathsf{sig}}}
\newcommand{\Ratt}{R_{\mathsf{att}}}
\newcommand{\kr}{k_{\mathsf{rev}}}
\newcommand{\cert}{\mathsf{cert}}
\newcommand{\Vault}{\mathsf{Vault}}
\newcommand{\mldsa}{\textsf{ML-DSA}}
\newcommand{\Hash}{\mathsf{H}}
\newcommand{\RO}{\rho}
\newcommand{\cS}{\mathfrak{S}}

\makeatletter\renewenvironment{proof}[1][\proofname]{\par\pushQED{\qed}\normalfont\topsep2pt\relax\trivlist\item[\hskip\labelsep\itshape#1\@addpunct{.}]\ignorespaces}{\popQED\endtrivlist\@endpefalse}\makeatother
\def\IEEEbibitemsep{0pt plus .5pt}
\newenvironment{restate}[1]{\par\smallskip\noindent\textbf{#1.}\ \itshape\ignorespaces}{\par\smallskip}
\begin{document}
\title{PQ-Vault: Provably Quantum-Safe Asset Custody for Unmodified Ethereum}

\author{Hassan~Khodaiemehr,~\IEEEmembership{Member,~IEEE,}
        Khadijeh~Bagheri,
        and~Chen~Feng,~\IEEEmembership{Member,~IEEE}%
\IEEEcompsocitemizethanks{\IEEEcompsocthanksitem H. Khodaiemehr, K. Bagheri, and C. Feng are with the School of Engineering, University of British Columbia (Okanagan Campus), Kelowna, BC, Canada.\protect\\
E-mail: \{hassan.khodaiemehr, khadijeh.bagheri, chen.feng\}@ubc.ca
\IEEEcompsocthanksitem This work has been submitted to the IEEE for possible publication. Copyright may be transferred without notice, after which this version may no longer be accessible.}}

\markboth{IEEE Transactions on Dependable and Secure Computing}%
{Khodaiemehr \MakeLowercase{\textit{et al.}}: PQ-Vault: Provably Quantum-Safe Asset Custody for Unmodified Ethereum}

\IEEEtitleabstractindextext{%
\begin{abstract}
A quantum computer could recover the key of any Ethereum account whose public key appears on chain. We prove that no sidechain or bridge that leaves Ethereum's rules unchanged can protect such accounts: their funds must move into a contract whose release condition relies only on quantum-safe primitives. We design \emph{PQ-Vault}, a custody system for unmodified Ethereum combining this vault, a quantum-safe proof-of-stake sidechain whose checkpoints carry compact certificates, and an optimistic dispute protocol. We give the first security proof of compact certificates against quantum attackers and a matching attack: with classically chosen parameters, the guarantee falls from 128 to about 60 bits, and a simple rule restores it. We show that a common seed-update rule in hash-chain proof of stake lets an attacker steer chosen slots and profit from splitting its stake, and give an election whose steering gain and inclusion delay are provably bounded against quantum attackers. Since compressing certificates with a pairing-based SNARK breaks quantum security, we use fraud proofs that need a single signature check. For 65,536 validators, a certificate is 0.8 MB, 260 times smaller than all signatures, and verifies in about 20 ms; the vault accepts a checkpoint for about 420,000 gas.
\end{abstract}

}

\maketitle
\IEEEdisplaynontitleabstractindextext
\IEEEpeerreviewmaketitle

\section{Introduction}\label{sec:intro}

\IEEEPARstart{E}{thereum} authenticates every externally owned account (EOA) with ECDSA over secp256k1, and its consensus layer relies on BLS signatures over pairing-friendly curves. Both are broken in polynomial time by Shor's algorithm~\cite{shor1997}. Once a cryptographically relevant quantum computer (CRQC) exists, which we call \emph{Q-day}, any account whose public key has appeared on chain can be emptied by an adversary that recovers the secret key and signs a valid transaction~\cite{aggarwal2018quantum,stewart2018committing}. Users and applications therefore need a way to hold value \emph{today} that stays safe after Q-day, using post-quantum (PQ) cryptography, without waiting for changes to Layer~1 (L1).

Four lines of work address this problem, and each contributes an essential piece. \emph{Protocol migration}, through an emergency hard fork that freezes ECDSA and recovers accounts~\cite{buterin2024emergency}, PQ precompiles~\cite{eip8051} and hash-based multi-signatures for consensus~\cite{drake2025hashsig}, is the definitive long-term solution, but it depends on community coordination and an uncertain timeline. \emph{PQ ledgers}, such as QRL~\cite{QRL}, Abelian~\cite{abelian_website} and permissioned designs~\cite{Refblock14,cryptoeprint:2020/205}, show that PQ consensus and PQ accounts work in production; they secure their own assets rather than those on Ethereum. \emph{PQ smart accounts}, built on Solidity verifiers for ML-DSA and Falcon~\cite{zknox_dilithium,zknox_falcon}, already run on unmodified Ethereum; each operation pays for an on-chain PQ verification, millions of gas with today's verifiers. \emph{Sidechains, bridges and light-client certificates}~\cite{gazi2019pos,micali2020compact,algorand_stateproofs,chaidos2021mithril,xie2022zkbridge} make cross-chain verification succinct and amortize it over many users, and PQ sidechain designs for Ethereum~\cite{myeprint} combine them with PQ signatures and hash-based proofs~\cite{Virgo_conf}. Their security arguments, however, were developed for classical adversaries. To our knowledge, no existing approach yet offers a custody architecture that runs on unmodified Ethereum today, keeps every check on the release path PQ, is proven secure against quantum adversaries end to end, and amortizes its cost across users.

Building such an architecture raises four questions that the classical setting hides. \emph{What can an external layer protect at all?} A PQ proof does not change which transactions L1 accepts, so a quantum adversary spends any asset that an ECDSA signature can spend (\cref{prop:impossible}). \emph{Are succinct certificates still sound?} Compact certificates~\cite{micali2020compact} are proven in the classical random-oracle model, where grinding with $q$ queries gains a factor $q$; Grover's algorithm~\cite{grover1996} gains a factor of order $q^2$, and classically chosen parameters keep only about $60$ of $128$ bits (\cref{cor:loss}). \emph{Is PQ leader election fair?} In hash-chain proof of stake (PoS)~\cite{QRL}, the next seed depends on who proposes, so a leader can steer it by skipping, and splitting stake pays off (\cref{prop:steer,cor:qrl}); a fix must also hold against superposition queries. \emph{Can verification on L1 be cheap without losing PQ security?} A pairing-based SNARK saves gas but puts a quantum-breakable assumption on the path that L1 executes (\cref{prop:snark}).

We present \emph{PQ-Vault}, a custody architecture for unmodified Ethereum that answers these questions, each with a proof and an evaluation (\cref{tab:contrib}, \cref{fig:arch}): a vault released only by ML-DSA and hash checks, which escrow under a PQ predicate makes necessary (C1); a quantum random-oracle model (QROM) analysis and reveal rule for compact certificates (C2); a steering-resistant, Sybil-proof epoch seed for stake-weighted election, which leaves only an explicitly bounded choice among seeds (C3); and single-reveal fraud proofs in place of SNARK compression (C4). PQ-Vault needs no change to L1; one checkpoint per epoch, accepted for $4.2\times10^5$ gas, serves all users; and the evaluation (C5), with real certificates for up to $65{,}536$ validators and a Solidity vault, is released as an open artifact~\cite{pqvault_artifact}. The appendices contain the full proofs, the complete dispute protocol, and further details of the design and evaluation.

\begin{table}[t]
\centering
\caption{Contributions of this paper.}
\label{tab:contrib}
\scriptsize
\setlength{\tabcolsep}{2.5pt}
\renewcommand{\arraystretch}{1.25}
\begin{tabular}{@{}>{\raggedright\arraybackslash}p{0.2\columnwidth}>{\raggedright\arraybackslash}p{0.77\columnwidth}@{}}
\toprule
Contribution & Key results\\
\midrule
\textbf{C1} Custody (\cref{sec:model,sec:design}) & No mechanism that leaves L1 unchanged protects ECDSA-controlled assets (\cref{prop:impossible}); escrow under a PQ predicate is necessary (\cref{cor:vault}); vault with a PQ-vetoable identity binding (\cref{prop:binding}).\\
\hline
\textbf{C2} Certificates (\cref{sec:cert}) & QROM bound $12t^2(\Wp/\Ws)^{\kr}+O(t^3/2^h)+2\Adv^{\mathsf{uf}}_K$ with explicit constants (\cref{thm:qrom}); Grover attack: the $q^2$ factor is necessary (\cref{thm:attack}); classical parameters: ${\approx}60$ proven bits, attack at $2^{-81}$ (\cref{cor:loss}); reveal rule with $\lambda+2\log_2 q+4$ and a security budget.\\
\hline
\textbf{C3} Leader election (\cref{sec:consensus}) & Exponential race (\cref{lem:race}); steering gain $m\alpha/(1+(m-1)\alpha)$ (\cref{prop:steer,cor:qrl}); commit--reveal epoch seeds with random reveal order ($\Pr[c\ge x]\le\mu\alpha^x$) and fallback proposers; share, capture and liveness bounds (\cref{thm:schedule,prop:late,lem:incl}) against mildly adaptive quantum adversaries, error below $2^{-178}$ plus a per-epoch failure below $2^{-40}$ at $\alpha=0.3$ (\cref{thm:fresh,cor:real}).\\
\hline
\textbf{C4} L1 acceptance (\cref{sec:opf}) & Groth16 wrapper is not PQ (\cref{prop:snark}); dispute state machine with single-reveal fraud proofs (\cref{thm:opf}); end-to-end custody by induction over checkpoints (\cref{thm:e2e}).\\
\hline
\textbf{C5} Evaluation (\cref{sec:eval}) & Dilithium (r3) certificates for $2^8$--$2^{16}$ validators ($0.8$\,MB, $260\times$ smaller), FIPS~204 cross-check; Falcon and hash-based~\cite{drake2025hashsig} baselines; TCP deployment; simulations of \cref{sec:consensus}; Solidity vault ($4.2\times10^5$ gas per checkpoint).\\
\bottomrule
\end{tabular}
\end{table}

\section{Related Work}\label{sec:related}

\noindent\textbf{Quantum threats and PQ ledgers.} Shor's algorithm breaks ECDSA, EdDSA and BLS, and Grover's algorithm speeds up preimage search~\cite{aggarwal2018quantum}; we adopt the distinction between slow and fast quantum attackers of~\cite{stewart2018committing} (\cref{asm:slow}). Our binding (\cref{sec:binding}) is compatible with the proposed quantum-emergency hard fork~\cite{buterin2024emergency} but does not depend on it, and \cref{sec:model} compares PQ-Vault with PQ smart accounts~\cite{zknox_dilithium,zknox_falcon}. QRL~\cite{QRL}, Abelian and its PQZK bridge~\cite{abelian_website,abelian_bridge}, PQFabric~\cite{Refblock14} and SodsBC~\cite{cryptoeprint:2020/205} secure their own ledgers; to our knowledge, which Ethereum assets an external PQ layer can protect has not been formalized before.

\noindent\textbf{Sidechains, light clients and aggregation.} PoS sidechains were formalized in~\cite{gazi2019pos}; light-client proofs include compact certificates and Algorand state proofs~\cite{micali2020compact,algorand_stateproofs}, Mithril~\cite{chaidos2021mithril}, and SNARK-based bridges such as zkBridge~\cite{xie2022zkbridge}, which rely on pairings~\cite{groth16} or on transparent hash-based proofs~\cite{Virgo_conf,bcs2016}. Lattice-based synchronized multi-signatures~\cite{fleischhacker2023chipmunk} and hash-based multi-signatures for Ethereum~\cite{drake2025hashsig} need synchronized or few-time keys, whereas compact certificates work with an unmodified standard scheme and arbitrary weights (\cref{sec:baselines}). To our knowledge, \cref{thm:qrom,thm:attack} give the first quantum analysis of compact certificates. Like optimistic rollups~\cite{kalodner2018arbitrum}, our dispute game is optimistic; the conjunctive structure of certificate verification makes its fraud proofs one-step.

\noindent\textbf{Leader election and the QROM.} Leader election based on verifiable random functions (VRFs)~\cite{gilad2017algorand,david2018praos} is not PQ unless the VRF is, and lattice VRFs remain costly~\cite{esgin2021lbvrf}; hash chains~\cite{QRL} and commit--reveal beacons are PQ alternatives. Grinding and long-range attacks on PoS are well known~\cite{kiayias2017ouroboros,gazi2018stakebleeding}. In the QROM~\cite{boneh2011qrom}, the soundness error of Fiat--Shamir and of BCS-compiled interactive oracle proofs scales with $t^2\varepsilon$~\cite{don2019fs,chiesa2019qrom}; we instantiate this lifting for compact certificates and show that the factor $q^2$ is necessary.

\Cref{tab:related} positions PQ-Vault against representative systems.

\begin{table}[t]
\centering
\caption{Positioning against related systems. \checkmark: provided; $\circ$: partial or assumed; --: not provided.}
\label{tab:related}
\scriptsize
\setlength{\tabcolsep}{2.5pt}
\resizebox{\columnwidth}{!}{%
\begin{tabular}{@{}lcccccc@{}}
\toprule
 & \makecell{Unmod.\\Ethereum} & \makecell{PQ L1\\accept path} & \makecell{QROM\\proof} & \makecell{Grinding\\analysis} & \makecell{Sublinear\\cert} & \makecell{Impl. +\\gas model}\\
\midrule
QRL / Abelian~\cite{QRL,abelian_website} & -- & \checkmark & -- & $\circ$ & -- & --\\
PQZK bridge~\cite{abelian_bridge} & $\circ$ & -- & -- & -- & -- & --\\
Compact certs~\cite{micali2020compact} & -- & -- & -- & $\circ$ & \checkmark & $\circ$\\
zkBridge~\cite{xie2022zkbridge} & \checkmark & -- & -- & n/a & \checkmark & \checkmark\\
Hash-sig multisig~\cite{drake2025hashsig} & -- & \checkmark & $\circ$ & n/a & $\circ$ & $\circ$\\
Our prior work~\cite{myeprint} & \checkmark & -- & -- & -- & \checkmark & $\circ$\\
\textbf{PQ-Vault (this work)} & \checkmark & \checkmark & \checkmark & \checkmark & \checkmark & \checkmark\\
\bottomrule
\end{tabular}}
\end{table}

\section{Preliminaries}\label{sec:prelim}

\noindent\textbf{Notation.} $\lambda$ is the security parameter. $\Hash$ is a hash function with $h$-bit output (Keccak-256 or SHA3-256, with $h=256$), modeled as a random oracle $\RO$ to which the adversary may make \emph{quantum} (superposition) queries~\cite{boneh2011qrom}. We write $q$ for the adversary's queries and $t\ge q$ for all queries in an experiment, including those of honest parties and verifiers. ML-DSA hashes internally with SHAKE256, whose padding differs from that of Keccak-256 and SHA3-256; we model it as an oracle independent of $\RO$. cSHAKE256 with a dedicated customization string, used for the commitments and seeds of \cref{sec:consensus}, is a third independent oracle. For weights $(w_i)$, $w(S)=\sum_{i\in S}w_i$. Stakes and weights are integers.

\noindent\textbf{Signatures.} We use ML-DSA (FIPS~204~\cite{FIPS204}, formerly CRYSTALS-Dilithium~\cite{dilithium2018}), designed to be strongly unforgeable~\cite[Sec.~3.1]{FIPS204} under MLWE, MSIS and SelfTargetMSIS~\cite[Sec.~4]{dilithium2018}. Our proofs need only existential unforgeability (EUF-CMA), so they hold for hedged and for deterministic signing alike~\cite[Sec.~3.4]{FIPS204}. In the multi-user game with $K$ honest keys, all generated by the challenger, and adaptive corruptions, $\Adv^{\mathsf{uf}}_{K}$ denotes the advantage of forging for an uncorrupted key; $K$ counts every honest key generated in the experiment, and guessing the target key loses at most a factor $K$. Public-key/signature sizes are $1312/2420$, $1952/3309$ and $2592/4627$ bytes for ML-DSA-44/65/87.

\noindent\textbf{Merkle trees.} Trees are domain-separated (prefix $0x00$ for leaves, $0x01$ for internal nodes), padded to $2^d$ leaves, and opened with \emph{compact multiproofs}, which contain the siblings of the minimal subtree spanning the opened set $P$ that are not in that subtree: at most $|P|(d-\lfloor\log_2|P|\rfloor+1)$ hashes, and far fewer for clustered positions, so we report measured sizes.

\noindent\textbf{QROM tools.} We use the database framework of Chiesa, Manohar and Spooner~\cite{chiesa2019qrom}. The oracle is simulated as a compressed oracle~\cite{zhandry2019record}, and measuring it after the algorithm halts yields a \emph{database} $D$ of input--output pairs. A \emph{database property} is a set of databases. Its \emph{instability} $I(P,t)$ is the largest probability, over databases $D$ with fewer than $t$ entries, inputs $x$ not in $D$, and a uniform output $y$, that adding $(x,y)$ to $D$ changes whether the database lies in $P$. The conditional instability $I(P\,|\,\bar Q,t)$ counts only transitions in which neither database lies in a second property $Q$~\cite[Def.~7]{chiesa2019qrom}.

\begin{lemma}[\cite{chiesa2019qrom}, Lemmas~1 and~9]\label{lem:extract}
Let $\mathcal A$ make $t$ quantum queries to a random oracle with $n$-bit outputs, and let $D$ be the database measured after it halts. (i)~For database properties $P,Q$ that do not contain the empty database, $\Pr[D\in P]\le 6t^2\big(I(P\,|\,\bar Q,t)+I(Q,t)\big)$. (ii)~If the final check of a game reads $k$ oracle values, its winning probability $\omega_O$ against the real oracle and the probability $\omega_D$ that the check passes with those values read from $D$ satisfy $\sqrt{\omega_O}\le\sqrt{\omega_D}+\sqrt{k/2^{n}}$.
\end{lemma}

\noindent When inputs of different kinds have different output lengths, instability is computed with the output length of the queried input, and in (ii) $n$ is the smallest output length among the values read. For the property $\mathsf{Col}$ that $D$ contains two inputs with equal $h$-bit outputs, $I(\mathsf{Col},t)<t/2^h$~\cite[Sec.~2.5]{chiesa2019qrom}. If $D\notin\mathsf{Col}$, each root in $D$ determines at most one tree of entries of $D$ below it; a leaf that does not occur in $D$ is treated as $\bot$.

\noindent\textbf{Compact certificates~\cite{micali2020compact}.} Fix an attestor list $((pk_i,w_i))_{i<N}$ committed in a Merkle tree with root $\Ratt$, a message $M$, and a \emph{proven weight} $\Wp$. Let $S$ be the set of attestors that signed $M$, and let $\Ws=w(S)>\Wp$ be the \emph{signed weight}. The prover assigns to each $i\in S$ the interval $[L_i,L_i+w_i)$, with the $L_i$ accumulated over signers in index order. It commits to the leaves $(\sigma_i,L_i)$ in a tree with root $\Rsig$ and derives $\kr$ coins
\begin{equation}\label{eq:coin}
c_j=\RO(\textsc{dom}\,\|\,j\,\|\,\Rsig\,\|\,\Ratt\,\|\,\Ws\,\|\,\Wp\,\|\,N\,\|\,M)\bmod \Ws ,
\end{equation}
for $j\in[\kr]$, using the first 128 output bits. For each coin it reveals the leaf whose interval covers the coin, together with that attestor's leaf in $\Ratt$. The verifier checks every coin, interval, signature, and Merkle opening (\cref{alg:cert}). In the classical ROM, a prover that knows signatures of total weight at most $\Wp$ succeeds with probability at most $q(\Wp/\Ws)^{\kr}$, and~\cite{micali2020compact} accordingly sets $\kr=\lceil(\lambda+\log_2 q)/\log_2(\Ws/\Wp)\rceil$.

\begin{algorithm}[t]
\caption{Certificate verification $\mathsf{Vf}(\Ratt,M,\Wp,\cert)$. The verifier state holds $W$, the number $N$ of positions of the attestor list (zero-weight positions included), and $d=\lceil\log_2 N\rceil$; $\delta$, $\kappa_Q$ and the design query bound $q=2^{64}$ are fixed in \cref{sec:cert}. All quantities are integers, and leaves have fixed-length encodings. Comments mark checks that were missing from, or different in, our prior prototype~\cite{NIROPoK2025}.}\label{alg:cert}
\footnotesize
\begin{algorithmic}[1]
\State parse $\cert=(\Rsig,\Ws,\{(i,pk_i,w_i,\sigma_i,L_i)\}_{i\in P},\pi_{\mathsf{sig}},\pi_{\mathsf{att}})$
\State \textbf{require} $(1+\delta)\Wp\le\Ws\le W$ \Comment{$W,N$ from state}
\State $\kr\gets\lceil(\lambda+2\log_2 q+\kappa_Q)/\log_2(\Ws/\Wp)\rceil$ \Comment{QROM rule}
\State \textbf{require} $P\subseteq[N]$, $1\le|P|\le\kr$ \Comment{query bound}
\For{$i\in P$} \textbf{require} $w_i>0 \wedge L_i+w_i\le \Ws$ \Comment{interval}
\EndFor
\For{$j=0,\dots,\kr-1$}
  \State $c_j\gets$ \cref{eq:coin} with a 128-bit reduction
  \State \textbf{require} $\exists i\in P: L_i\le c_j<L_i+w_i$ \Comment{coverage}
\EndFor
\State \textbf{require} each $i\in P$ covers some $c_j$ \Comment{no unused reveal}
\For{$i\in P$} \textbf{require} $\mldsa.\mathsf{Vf}(pk_i,M,\sigma_i)=1$
\EndFor
\State \textbf{require} $\pi_{\mathsf{sig}}$ opens $\{(\sigma_i,L_i)\}_{i\in P}$ at $P$ under $\Rsig$
\State \textbf{require} $\pi_{\mathsf{att}}$ opens $\{(pk_i,w_i)\}_{i\in P}$ at $P$ under $\Ratt$
\State \Return $1$
\end{algorithmic}
\end{algorithm}

\section{System and Threat Model}\label{sec:model}

\subsection{Entities}
\noindent On unchanged \emph{Ethereum (L1)}, \emph{users} hold an address $\mathit{Addr}_M$ with ECDSA keys $(sk_M,pk_M)$ and ML-DSA keys $(sk_S,pk_S)$; ML-DSA-keyed \emph{validators} stake on the sidechain $\mathcal{S}$; \emph{relayers} post checkpoints to L1; \emph{watchers} (anyone) challenge invalid ones; the \emph{vault} $\Vault$ escrows assets (\cref{sec:design}).

\subsection{Adversary}
The QPT adversary $\mathcal{A}$ (i)~recovers the secret key of any ECDSA or BLS public key it sees after Q-day $\tau_Q$; (ii)~makes $q$ quantum queries to $\RO$ ($q\le 2^{64}$ in concrete estimates); (iii)~corrupts validators holding at most a fraction $f<1/3$ of every epoch's stake, \emph{mildly adaptively}: a corruption requested in epoch $x$ takes effect from epoch $x+3$, the validator staying honest and keeping its secrets until then. It may also suppress (DoS) at most $B$ honest proposers per slot, chosen after seeing the schedule. Offline honest validators count as corrupt in the leader-election analysis, whose bounds need corrupt plus offline stake below one third. Past-epoch validators may be corrupted after unbonding, which releases stake only $T_{\mathsf{fin}}$ after $\Vault$ finalized a checkpoint whose next attestor set excludes the validator (\cref{sec:opf}), so no retired key is in an attestor set the vault still accepts, even if checkpointing stalls; (iv)~schedules sidechain messages within the synchrony bound $\Delta_{\mathcal S}$; and (v)~submits arbitrary L1 transactions and orders the mempool, subject to \cref{asm:l1}.

\subsection{Assumptions}

\begin{assumption}[L1 integrity]\label{asm:l1}
L1 executes contract code as specified, and any transaction submitted by an honest party is included within $\Delta_{\mathsf{inc}}$.
\end{assumption}

\begin{assumption}[PQ primitives]\label{asm:pq}
ML-DSA (hedged or deterministic) is EUF-CMA secure against QPT adversaries, and $\Hash$ behaves as a quantum random oracle.
\end{assumption}

\begin{assumption}[Slow key recovery]\label{asm:slow}
Recovering an ECDSA secret key from a public key takes longer than $\Delta_{\mathsf{inc}}$ (no \emph{fast} quantum attacker~\cite{stewart2018committing}). It is needed only for transactions that \emph{expose} an ECDSA key (deposits, challenges from fresh EOAs); private transaction channels relax it further.
\end{assumption}

\begin{assumption}[One honest watcher]\label{asm:watcher}
At least one honest watcher is online whenever a proposal is pending. Within $\Delta_{\mathsf{chk}}$ of a proposal's inclusion, or of the completion of an open-all posting, it obtains the certificate records or determines that they are unavailable, and evaluates $\mathsf{Local}$ on all of them (\cref{sec:opf}). It can pay, per pending proposal, one open-all fee and one challenge bond $b_C$, and submits $\mathsf{timeout}$ and $\mathsf{timeoutAll}$ once eligible.
\end{assumption}

\subsection{Security Goals}
\noindent\textbf{G1 (custody safety).} The vault releases an asset only to a recipient authorized by its legitimate PQ owner, as recorded by the honest sidechain.
\noindent\textbf{G2 (certificate soundness).} No accepted checkpoint certificate for $(\Ratt,M,\Wp)$ exists unless attestors of total weight greater than $\Wp-w(\text{corrupt})$ honestly signed $M$.
\noindent\textbf{G3 (fair, steering-resistant leader election).} The adversary's leadership share and its probability of capturing a target slot are bounded explicitly in its stake fraction and $B$.
\noindent\textbf{G4 (liveness).} Given a live finality gadget, an honest withdrawal is paid within bounded time, except with the small probability of \cref{cor:real}.

\subsection{Why External PQ Layers Cannot Protect EOAs}

\begin{proposition}[Custody impossibility]\label{prop:impossible}
Let $\Pi$ be any protocol that does not modify L1 validity rules. Suppose that at some time $t\ge\tau_Q$ an asset $a$ is still spendable under the L1 predicate $\mathsf{ECDSA.Vf}(pk_M,\cdot)=1$ with $pk_M$ public. Then a QPT adversary can produce an L1-valid transaction spending $a$ that L1 cannot distinguish from one signed by the owner. Nothing produced by $\Pi$ can make this transaction invalid.
\end{proposition}
\begin{proof}[Proof sketch]
$\mathcal{A}$ obtains $sk_M$ by Shor's algorithm; L1 validity depends only on the predicate and account state, which $\Pi$ cannot change. The full proof is in \supp{S1}.
\end{proof}

\begin{corollary}\label{cor:vault}
Consider, on unmodified L1, an asset that must remain \emph{live}, i.e., movable by its owner, and \emph{safe} against $\mathcal{A}$ after $\tau_Q$. Every predicate that can release it after $\tau_Q$ must then rest on PQ assumptions and on \cref{asm:l1} alone; in particular, no ECDSA key whose public key is exposed may authorize its release, or change or redirect such a predicate without a PQ-authorized veto. On an account-based L1 this requires holding the asset under a contract predicate $\Phi$.
\end{corollary}
\begin{proof}[Proof sketch]
Such a key falls under \cref{prop:impossible}, and only contract code can replace an EOA's ECDSA predicate. The full proof is in \supp{S1}.
\end{proof}

\noindent PQ-Vault directly instantiates \cref{cor:vault} (scope and consequences: \supp{S1}). So does a contract account verifying an ML-DSA or Falcon signature per operation~\cite{zknox_dilithium,zknox_falcon}, with no committee or watcher but $1.3$--$8.2\times10^{6}$ gas per operation (\cref{tab:gas}), versus $1.0\times10^{5}$ gas per exit plus $4.2\times10^{5}$ per shared checkpoint for PQ-Vault under the assumptions of \cref{sec:opf}. An EOA delegating to such code (EIP-7702~\cite{eip7702}) stays ECDSA-controlled and falls under \cref{prop:impossible}.

\section{PQ-Vault Design}\label{sec:design}

In \cref{fig:arch}, value enters the vault on L1, moves on the PQ sidechain, and leaves only via withdrawals in a \emph{certified checkpoint}, accepted directly or optimistically (\cref{sec:opf}). After $\tau_Q$ the vault evaluates only Keccak and, on the dispute path, ML-DSA predicates; the ECDSA-authorized binding registry is off the release path (\cref{sec:binding}), so \cref{cor:vault} is met.

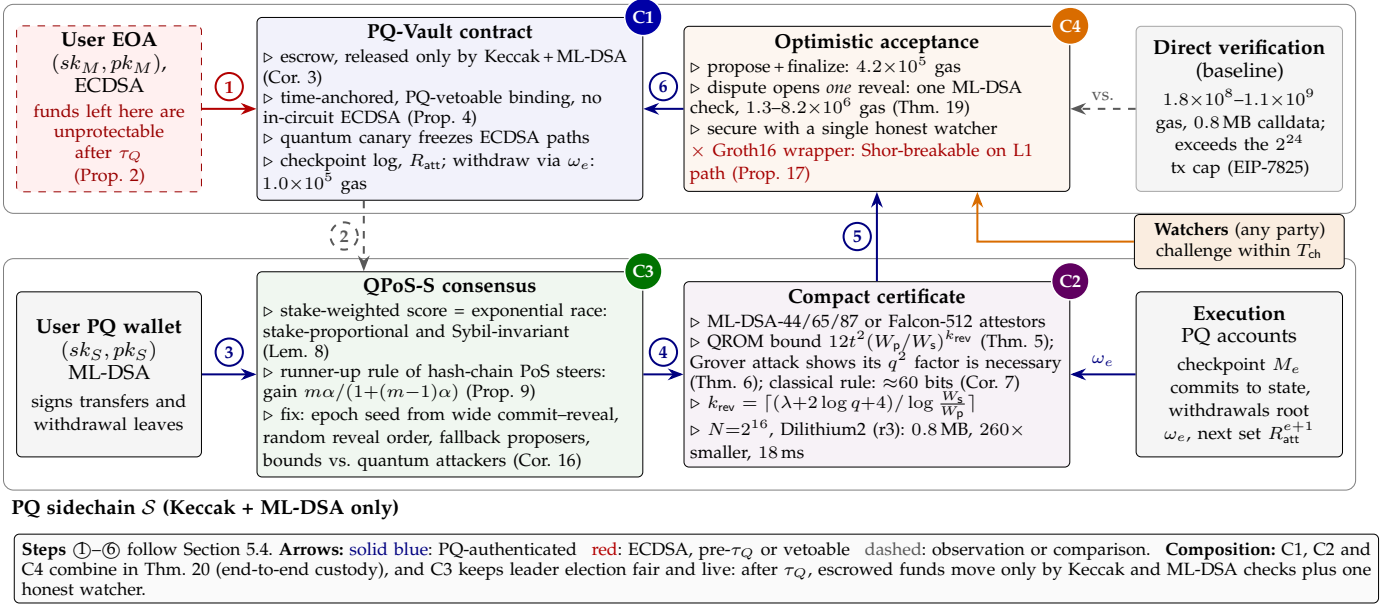
\begin{figure*}[t]
\centering
\resizebox{\textwidth}{!}{%
\begin{tikzpicture}[
  font=\footnotesize,
  hd/.style={font=\footnotesize\bfseries},
  box/.style={draw, rounded corners=2pt, align=left, text width=52mm, inner sep=2.5pt, minimum height=23mm, execute at begin node=\raggedright, fill=white},
  small/.style={draw, rounded corners=2pt, align=center, text width=24mm, inner sep=2.5pt, minimum height=23mm},
  lane/.style={draw, rounded corners=4pt, inner sep=5pt, gray},
  badge/.style={circle, draw=white, line width=0.6pt, text=white, font=\scriptsize\bfseries, inner sep=0.6pt, minimum size=5.2mm},
  num/.style={circle, draw, fill=white, font=\scriptsize\bfseries, inner sep=0.5pt, minimum size=3.8mm},
  lab/.style={font=\scriptsize, fill=white, inner sep=1pt, align=center},
  arr/.style={-{Stealth[length=2.2mm]}, thick, blue!50!black},
  rarr/.style={-{Stealth[length=2.2mm]}, thick, red!70!black},
  darr/.style={-{Stealth[length=2.2mm]}, thick, dashed, gray!70!black}]
\colorlet{cA}{blue!65!black}\colorlet{cB}{violet!75!black}\colorlet{cC}{green!45!black}\colorlet{cD}{orange!85!black}
\node[small, fill=red!6, draw=red!60!black, dashed] (user1) at (1.3,0) {\textbf{User EOA}\\ $(sk_M,pk_M)$, ECDSA\\[2pt] {\scriptsize\color{red!60!black}funds left here are\\ unprotectable after $\tau_Q$\\ (Prop.~\ref{prop:impossible})}};
\node[box, fill=cA!5] (vault) at (6.05,0) {{\centering\textbf{PQ-Vault contract}\par}\vspace{1.5pt}\scriptsize
$\triangleright$ escrow, released only by Keccak\,+\,ML-DSA (Cor.~\ref{cor:vault})\\
$\triangleright$ time-anchored, PQ-vetoable binding, no in-circuit ECDSA (Prop.~\ref{prop:binding})\\
$\triangleright$ quantum canary freezes ECDSA paths\\
$\triangleright$ checkpoint log, $\Ratt$; withdraw via $\omega_e$: $1.0{\times}10^5$ gas};
\node[box, fill=cD!6] (opf) at (12.0,0) {{\centering\textbf{Optimistic acceptance}\par}\vspace{1.5pt}\scriptsize
$\triangleright$ propose\,+\,finalize: $4.2{\times}10^5$ gas\\
$\triangleright$ dispute opens \emph{one} reveal: one ML-DSA check, $1.3$--$8.2{\times}10^6$ gas (Thm.~\ref{thm:opf})\\
$\triangleright$ secure with a single honest watcher\\
{\color{red!65!black}$\times$ Groth16 wrapper: Shor-breakable on L1 path (Prop.~\ref{prop:snark})}};
\node[small, fill=gray!8, draw=gray!70, text width=27mm] (direct) at (17.05,0) {\textbf{Direct verification}\\ (baseline)\\[2pt] {\scriptsize $1.8{\times}10^8$--$1.1{\times}10^9$ gas, $0.8$\,MB calldata; exceeds the $2^{24}$ tx cap (EIP-7825)}};
\node[small, fill=gray!8] (user2) at (1.3,-3.7) {\textbf{User PQ wallet}\\ $(sk_S,pk_S)$\\ ML-DSA\\[2pt] {\scriptsize signs transfers and withdrawal leaves}};
\node[box, fill=cC!6] (qpos) at (6.05,-3.7) {{\centering\textbf{QPoS-S consensus}\par}\vspace{1.5pt}\scriptsize
$\triangleright$ stake-weighted score = exponential race: stake-proportional and Sybil-invariant (Lem.~\ref{lem:race})\\
$\triangleright$ runner-up rule of hash-chain PoS steers: gain $m\alpha/(1{+}(m{-}1)\alpha)$ (Prop.~\ref{prop:steer})\\
$\triangleright$ fix: epoch seed from wide commit--reveal, random reveal order, fallback proposers, bounds vs.\ quantum attackers (Cor.~\ref{cor:real})};
\node[box, fill=cB!5] (ccok) at (12.0,-3.7) {{\centering\textbf{Compact certificate}\par}\vspace{1.5pt}\scriptsize
$\triangleright$ ML-DSA-44/65/87 or Falcon-512 attestors\\
$\triangleright$ QROM bound $12t^2(\Wp/\Ws)^{\kr}$ (Thm.~\ref{thm:qrom}); Grover attack shows its $q^2$ factor is necessary (Thm.~\ref{thm:attack}); classical rule: ${\approx}60$ bits (Cor.~\ref{cor:loss})\\
$\triangleright$ $\kr=\lceil(\lambda{+}2\log q{+}4)/\log\frac{\Ws}{\Wp}\rceil$\\
$\triangleright$ $N{=}2^{16}$, Dilithium2 (r3): $0.8$\,MB, $260\times$ smaller, $18$\,ms};
\node[small, fill=gray!8, text width=27mm] (roll) at (17.05,-3.7) {\textbf{Execution}\\ PQ accounts\\[2pt] {\scriptsize checkpoint $M_e$ commits to state, withdrawals root $\omega_e$, next set $\Ratt^{e+1}$}};
\node[draw, rounded corners=2pt, fill=cD!12, align=center, font=\scriptsize, text width=27mm] (watch) at (17.05,-1.85) {\textbf{Watchers} (any party)\\ challenge within $T_{\mathsf{ch}}$};
\node[badge, fill=cA] at (vault.north east) {C1};
\node[badge, fill=cD] at (opf.north east) {C4};
\node[badge, fill=cC] at (qpos.north east) {C3};
\node[badge, fill=cB] at (ccok.north east) {C2};
\begin{scope}[on background layer]
\node[lane, fit=(user1)(vault)(opf)(direct), label={[anchor=south west, hd]north west:Ethereum L1 (unmodified: no fork, no precompile)}] {};
\node[lane, name=sclane, fit=(user2)(qpos)(ccok)(roll), label={[anchor=north west, hd]south west:PQ sidechain $\mathcal{S}$ (Keccak + ML-DSA only)}] {};
\end{scope}
\draw[rarr] (user1) -- node[num, above=3pt]{1} (vault);
\draw[darr] ([xshift=-12mm]vault.south) -- node[num, left=2pt]{2} ([xshift=-12mm]qpos.north);
\draw[arr] (user2) -- node[num, above=3pt]{3} (qpos);
\draw[arr] (qpos) -- node[num, above=3pt]{4} (ccok);
\draw[arr] (roll) -- node[lab, above=2pt]{$\omega_e$} (ccok);
\draw[arr] (ccok.north) -- node[num, left=2pt]{5} (opf.south);
\draw[arr] (opf) -- node[num, above=3pt]{6} (vault);
\draw[arr, cD] (watch.west) -| ([xshift=14mm]opf.south);
\draw[darr] (direct) -- node[lab, above=2pt]{vs.} (opf);
\node[draw, rounded corners=2pt, fill=gray!4, align=left, font=\scriptsize, text width=188mm, inner sep=3pt, anchor=north west] (leg) at ([yshift=-6mm,xshift=1.5mm]sclane.south west) {%
\textbf{Steps} \textcircled{1}--\textcircled{6} follow \cref{sec:flow}. \textbf{Arrows:} {\color{blue!50!black}solid blue}: PQ-authenticated \enspace {\color{red!70!black}red}: ECDSA, pre-$\tau_Q$ or vetoable \enspace {\color{gray!70!black}dashed}: observation or comparison. \enspace
\textbf{Composition:} C1, C2 and C4 combine in Thm.~\ref{thm:e2e} (end-to-end custody), and C3 keeps leader election fair and live: after $\tau_Q$, escrowed funds move only by Keccak and ML-DSA checks plus one honest watcher.};
\end{tikzpicture}}
\caption{PQ-Vault architecture (QPoS-S: quantum-safe proof of stake). Badges C1--C4: contributions of \cref{tab:contrib}; gas measured with Foundry (\cref{sec:eval}); grey: baseline avoided by C4.}
\label{fig:arch}
\end{figure*}

\subsection{Vault Contract}\label{sec:vault}
$\Vault$ stores (i)~escrowed balances; (ii)~a binding registry $\mathsf{Bind}[\mathit{Addr}_M]=\Hash(pk_S)$; (iii)~the latest finalized epoch $e^\ast$, its checkpoint message $M_{e^\ast}$, the attestor commitment $\Ratt$ with total stake $W$, $\Wp=\lfloor W/3\rfloor$ and attestor count $N$, and the finalized withdrawal roots; (iv)~pending proposals with their dispute state; and (v)~spent withdrawal nullifiers. $\mathsf{deposit}(\beta)$ escrows \texttt{msg.value} for the PQ identity $\beta=\Hash(pk_S)$; $\mathsf{propose}$, with a bond $b_R$, takes a header $(e,h_e,\rho_e,\omega_e,\Ratt^{e+1},\beta_e,W^{e+1},N^{e+1})$ (block hash, state root, \emph{withdrawals root}, next attestor commitment, binding-operations root, next stake and attestor count) and $(\Rsig,\Ws,\kr,C_{\cert})$. It requires $e=e^\ast+1$ and $\kr$ equal to \cref{eq:rule}, recomputes
\begin{multline*}
M_e=\Hash(\textsc{cp}\,\|\,\mathit{chainid}\,\|\,\Vault\,\|\,e\,\|\,M_{e^\ast}\,\|\,h_e\,\|\,\rho_e\\\|\,\omega_e\,\|\,\Ratt^{e+1}\,\|\,\beta_e\,\|\,W^{e+1}\,\|\,N^{e+1}),
\end{multline*}
on which signatures and coins are checked, snapshots $(\Ratt,W,\Wp,N)$, and opens the proposal for challenge. As the vault computes $M_e$ itself, a finalized $\omega_e$ and $\Ratt^{e+1}$ are exactly those signed; $C_{\cert}$ commits to the certificate's reveal records (blob data). Further proposal mechanics are in \supp{S2} and \cref{sec:opf}. After the challenge window, with no dispute open, $\mathsf{finalize}$ sets $e^\ast\gets e$, $\Ratt\gets\Ratt^{e+1}$, $W\gets W^{e+1}$, $\Wp\gets\lfloor W^{e+1}/3\rfloor$, $N\gets N^{e+1}$, and other proposals for epoch $e$ become stale. $\mathsf{withdraw}$ pays $\mathit{amount}$ to $\mathit{recipient}$ if a Merkle path opens the leaf $\xi=\Hash(0x00\,\|\,\Vault\,\|\,\mathit{chainid}\,\|\,\mathit{recipient}\,\|\,\mathit{amount}\,\|\,\mathit{nonce})$ at its index under a finalized $\omega_e$ and $\xi$ is unspent; it marks $\xi$ spent before transferring. Globally unique nonces make $\xi$ a nullifier, paid at most once.

A withdrawal leaf is created only by a sidechain transaction signed with $sk_S$, so the chain of trust is $\Ratt^{0}$ (fixed at deployment) $\to$ certified checkpoints $\to$ $\omega_e$ $\to$ a withdrawal leaf $\to$ an ML-DSA signature; ECDSA is never consulted. Recipients should be fresh addresses or PQ smart accounts (\cref{prop:impossible}).

\subsection{Time-Anchored, PQ-Vetoable Identity Binding}\label{sec:binding}
Replacing the ZK binding of~\cite{myeprint} (\supp{S2}), we authenticate the binding classically \emph{while classical authentication is still meaningful}, and freeze it afterwards. \emph{Bind:} if no binding exists, $\mathit{Addr}_M$'s call $\mathsf{bind}(\Hash(pk_S))$ takes effect immediately. \emph{Rebind:} a later call creates the (single) pending rebinding of the address, with a fresh counter $n$ and activation time $\tau_{\mathsf{act}}=\tau_{\mathsf{req}}+D$; $\mathsf{applyRebind}$ succeeds only at a timestamp strictly greater than $\tau_{\mathsf{act}}$ and only if $n$ is not cancelled. \emph{Veto:} the current PQ identity ML-DSA-signs $\textsf{cancel}(\Vault,\mathit{Addr}_M,n)$ on the sidechain; it reaches L1 in the root $\beta_e$ of a certified checkpoint, and $\mathsf{applyCancel}(\mathit{Addr}_M,n,e,\pi)$, callable by anyone, checks that $\pi$ opens it under the finalized $\beta_e$ and marks $n$ cancelled, effective if included by $\tau_{\mathsf{act}}$. \emph{Canary:} a bounty on an address with unknown discrete logarithm, claimable by an ECDSA signature on a fixed digest, permanently disables all ECDSA-authorized registry operations once quantum attacks are demonstrated (\supp{S2}).

\begin{proposition}[Binding integrity]\label{prop:binding}
Suppose $\mathsf{Bind}[\mathit{Addr}_M]=\Hash(pk_S)$ was set by $\mathit{Addr}_M$'s owner, and let $\Delta_{\mathsf{v}}=\Delta_{\mathsf{obs}}+\Delta_{\mathsf{sc}}+\Delta_{\mathsf{cp}}+\Delta_{\mathsf{apply}}$, where $\Delta_{\mathsf{obs}}$ bounds the owner's time to observe a pending rebinding and sign, $\Delta_{\mathsf{sc}}$ bounds sidechain inclusion of its transaction, $\Delta_{\mathsf{cp}}$ bounds the time from inclusion to finalization of the carrying checkpoint, including disputes and the $\mathsf{finalize}$ call (\cref{thm:opf}), and $\Delta_{\mathsf{apply}}$ bounds preparing and including $\mathsf{applyCancel}$. Suppose $D>\Delta_{\mathsf{v}}$, the owner is online at least once in every window of length $D-\Delta_{\mathsf{v}}$, these bounds hold, and \cref{asm:l1,asm:pq,asm:slow,asm:watcher} and the timing conditions of \cref{thm:opf} hold. By \cref{cor:real} and \cref{thm:opf}, $\Delta_{\mathsf{sc}}$ and $\Delta_{\mathsf{cp}}$ fail only with the probabilities given there. Then a QPT adversary can change the binding only by forging an ML-DSA signature, by breaking certificate soundness, or before the owner's first binding transaction is included.
\end{proposition}
\begin{proof}[Proof sketch]
Any change after the first binding needs a rebinding pending for $D$, during which the owner is online with at least $\Delta_{\mathsf{v}}$ remaining and gets a cancellation of $n$ applied by $\tau_{\mathsf{act}}$, before $\mathsf{applyRebind}$ can succeed. Preventing this requires an ML-DSA-signed operation the owner never signed (a forgery), or a finalized checkpoint omitting final binding operations, which no honest validator signs (\cref{thm:opf,thm:qrom}). After the canary fires, the binding cannot change at all. The full proof is in \supp{S2}.
\end{proof}

\noindent No later theorem relies on the registry or the canary (\supp{S2}).

\subsection{Sidechain Protocol}
The sidechain runs QPoS-S (quantum-safe proof of stake), the steering-resistant stake-weighted election of \cref{sec:consensus}. Each epoch has $k$ slots; validator $i$ has a unique $64$-bit identifier $\mathit{ID}_i$, and epochs are $64$-bit numbers. With its ML-DSA-signed registration for epoch $e$, which fixes its stake $w_i$, validator $i$ commits to a uniform $r_i\in\{0,1\}^{512}$ by $a_i=\Hash_{1536}(\textsc{commit}\,\|\,\mathit{ID}_i\,\|\,e\,\|\,r_i)$, with $\Hash_n$ the $n$-bit domain-separated cSHAKE256 (\supp{S2}). Registrations for epoch $e$ must be final before the reveal window of epoch $e-2$ opens. Reveals for epoch $e$ are ordinary transactions in the last $\ell=\ell_0+L_t\ell_s$ slots of epoch $e-1$. The reveal \emph{order} is a stake-weighted random permutation: with $\kappa_i=-\log_2 u(\Hash(\textsc{order}\,\|\,e\,\|\,\mathit{ID}_i\,\|\,\cS^{(e-1)}))/w_i$, rank~$1$ goes to the smallest $\kappa_i$. The $L_t$ best-ranked identities form the \emph{tail} and reveal last, in reverse rank order, each in its own sub-window of $\ell_s$ slots; all others reveal in a common first phase of $\ell_0$ slots. A reveal is valid only if it opens $a_i$ and is included in its owner's phase or sub-window, and honest validators broadcast $r_i$ only when that starts. Single-slot finality (as in Tendermint-style BFT gadgets) finalizes each reveal-window block before the next slot, so earlier reveals are final when a sub-window starts. A validator missing its epoch-$e$ reveal is suspended (not slashed) from epoch $e+3$ until it re-registers. QPoS-S uses $L_t=25$, $\ell_s=\ell_0=52$, so $\ell=1{,}352$, and $k\ge2{,}000$. The reveal list $X$, encoding the revealed $(\mathit{ID}_i,r_i)$ as fixed-length fields in increasing identifier order, determines the seed
\begin{equation}\label{eq:schedule}
\cS^{(e)}=\Hash_{512}(\textsc{seed}\,\|\,e\,\|\,X).
\end{equation}
In slot $j\in\{1,\dots,k\}$ of epoch $e$, validator $i$'s score is
\begin{equation}\label{eq:score}
s_{i,j}=\frac{-\log_2 u\big(\Hash(\textsc{score}\,\|\,e\,\|\,j\,\|\,\mathit{ID}_i\,\|\,\cS^{(e)})\big)}{w_i},
\end{equation}
with $u(x)=(\mathrm{int}(x)+1)/2^{h}$. The lowest score leads (ties: smaller identifier); the epoch's schedule is computable once the seed is final. Ranks $2,\dots,L+1$ are \emph{fallback proposers}: the proposer of slot $j$ is the best-ranked identity among ranks $1,\dots,L+1$ whose ML-DSA-signed block arrives within its deadline $r\Delta_{\mathcal S}$ for rank $r$, and an honest proposer includes every pending transaction, including reveals, received $\Delta_{\mathcal S}$ earlier. The fallback changes who proposes, never the seed, so it opens no steering channel. QPoS-S uses $L=3$. These sizes keep commitments statistically binding and hiding errors small in the QROM (\cref{thm:fresh}; overhead in \supp{S2}). Unlike~\cite{QRL,myeprint,NIROPoK2025}, the seed \emph{does not depend on who proposed earlier blocks}; \cref{sec:consensus} shows why this matters.

Blocks are signed with ML-DSA. Every $N_{\mathsf{cp}}$ blocks the leader assembles a checkpoint message $M_e$. \emph{Honest-signing rule:} an honest validator signs $M_e$ only if (i)~the prefix it commits to is final (any standard BFT finality gadget with $f<1/3$ fits) and extends $M_{e-1}$; (ii)~every state transition is valid under the execution rules, which credit each final deposit event once, debit a withdrawal's amount from its owner's balance, and accept a transfer or withdrawal only with a valid ML-DSA signature of the owner's $sk_S$ on the full operation; (iii)~$\omega_e$ contains exactly the leaves of the withdrawals executed in epoch $e$, each with a fresh $\mathit{nonce}$; and (iv)~$\Ratt^{e+1}$ lists the bonded validators of epoch $e+1$ with their stakes, each at one position and with a registered key whose expanded form the validator has checked. Any party can aggregate signatures into a certificate (\cref{alg:cert}). With $\Wp=\lfloor W/3\rfloor$ and $f<1/3$, a sound certificate implies that at least one honest validator signed $M_e$, hence that $M_e$ satisfies the rule.

\noindent\textbf{Trust in the attestor set.} The vault checks certification, not validity: over a third of the attestor stake could certify an arbitrary $M_e$, whereas $\Wp=2W/3$ would require more than two thirds at about $4\times$ the reveals and $2.8\times$ the certificate size (\cref{tab:kr,sec:eval}). Custody safety therefore rests on corrupt stake staying below one third of every accepted attestor set (\cref{sec:model}) (economics: \supp{S2}).

\subsection{End-to-End Flow}\label{sec:flow}
The six steps are: (1)~deposit into $\Vault$; (2)~after L1 finality, validators credit $\Hash(pk_S)$; (3)~PQ-signed transfers on $\mathcal S$; (4)~an $sk_S$-signed withdrawal creates a leaf in $\omega_e$; (5)~a relayer proposes $(M_e,\cert_e)$, finalized after $T_{\mathsf{ch}}$; (6)~anyone calls $\mathsf{withdraw}$ with the leaf's Merkle proof. Latency is $\Delta_{\mathsf{cp}}=$ checkpoint interval $+$ certification time $+$ finalization time, the last being $T_{\mathsf{ch}}+\Delta_{\mathsf{inc}}$ without disputes and at most $T_{\mathsf{fin}}+\Delta_{\mathsf{inc}}$ with them (\cref{thm:opf}), if $\mathsf{finalize}$ is submitted once eligible. As a checkpoint may wait for its predecessor ($e=e^\ast+1$), in the worst case $\Delta_{\mathsf{cp}}\le$ checkpoint interval $+$ certification time $+\,2(T_{\mathsf{fin}}+\Delta_{\mathsf{inc}})$, provided the checkpoint interval is at least $T_{\mathsf{fin}}+\Delta_{\mathsf{inc}}$ (used in \cref{prop:binding}; certification time: \cref{sec:eval}).

\section{Compact Certificates Against Quantum Provers}\label{sec:cert}

\subsection{Security Game}
\begin{definition}[$\mathsf{CERT\text{-}FORGE}$]\label{def:game}
The challenger generates ML-DSA key pairs for honest attestors on request, at most $K$ in total. The adversary $\mathcal A$ fixes an attestor list $((pk_i,w_i))_{i<N}$ with integer weights, adaptively and possibly after signing queries. Each position carries either the key of one honest attestor (at most one position per honest key) or a key \emph{registered} by $\mathcal A$; registered positions are corrupt. The challenger commits to the list with root $\Ratt$, and $W=\sum_iw_i$. $\mathcal A$ may corrupt honest attestors and then receives their secret keys; $C$ is the set of positions that are registered or corrupted when $\mathcal A$ halts. At any time, $\mathcal A$ may ask an uncorrupted honest attestor $i$ to sign a message of its choice; $Q_i$ is the set of messages $i$ signed. This models honest validators signing the checkpoints they consider final, whenever they do so. $\mathcal A$ makes $q$ quantum queries to $\RO$ and outputs $(M^\ast,\cert^\ast)$. It wins if $\mathsf{Vf}(\Ratt,M^\ast,\Wp,\cert^\ast)=1$ and
\[w\big(C\cup S_{M^\ast}\big)\le \Wp,\qquad S_{M^\ast}=\{i\notin C: M^\ast\in Q_i\}.\]
\end{definition}

Winning means that the certificate claims more support than exists; the corrupt/honest split exists only in the analysis. If $w(C)\le\Wp$, an accepted certificate that does not win implies that at least one honest attestor signed $M^\ast$. $\Wp=\lfloor W/3\rfloor$ meets this for every $f<1/3$.

\subsection{QROM Soundness}
For a claimed signed weight $\Ws>\Wp$ and $\kr$ reveals, define the per-commitment soundness error
\begin{equation}\label{eq:eps}
\varepsilon(\Ws)=\Big(\frac{\Wp}{\Ws}+\frac{\Ws}{2^{128}}\Big)^{\kr(\Ws)}.
\end{equation}
The additive term accounts for the 128-bit modular reduction in \cref{eq:coin}. The verifier accepts only $(1+\delta)\Wp\le\Ws\le W$ (\cref{alg:cert}), with $W<2^{64}$ from its state; let $\varepsilon^\ast=\max\varepsilon(\Ws)$ over this range.

\begin{theorem}[QROM soundness of compact certificates]\label{thm:qrom}
Let $d=\lceil\log_2N\rceil$, let $q_V=\kr(2d+3)$ with $\kr=\kr((1+\delta)\Wp)$, let $t=q+q_V+2^{d+1}$, and let $h$ be the smallest output width of the oracles and trees involved. For every adversary making $q$ quantum queries to $\RO$ and any number of signing queries,
\[
\Adv^{\mathsf{CERT\text{-}FORGE}}_{\mathcal A}\le 12t^2\varepsilon^\ast+\frac{48t^3+2q_V}{2^{h}}+2\Adv^{\mathsf{uf}}_{K}.
\]
\end{theorem}

\begin{proof}[Proof sketch]
Following~\cite{bcs2016,chiesa2019qrom}, we merge the $\kr$ coins of \cref{eq:coin} into one wide oracle (one query simulates each query) and apply \cref{lem:extract}(i) to a fixed database property $P$: some coin entry is \emph{light} (valid positions below its roots weigh at most $\Wp$) and \emph{hit} (all coins covered), with instability at most $\varepsilon^\ast+2t/2^h$. $P$ depends neither on the signing history nor on $C$. If $\mathcal A$ wins, either $D\in P$ or $D$ holds an EUF-CMA forgery, extracted by measuring $D$ after $\mathcal A$ halts. With $I(\mathsf{Col},t)<t/2^h$, this gives $\omega_D\le6t^2\varepsilon^\ast+24t^3/2^h+\Adv^{\mathsf{uf}}_K$ for the database game, and \cref{lem:extract}(ii) with the $q_V$ values read by the verifier gives $(\sqrt{\omega_D}+\sqrt{q_V/2^h})^2\le2\omega_D+2q_V/2^h$, which is the bound. The full proof is in \supp{S3}.
\end{proof}

\noindent The proof needs no entropy of honest signatures, and for $\Wp\ge W/4$, $(1+\delta)\Wp\le\Ws\le W<2^{64}$ and $\delta=0.1$, the reduction term changes $\varepsilon(\Ws)$ by a factor of at most $1+2^{-51}$ (\supp{S3}).

\subsection{The Quadratic Dependence on \texorpdfstring{$q$}{q} Is Necessary}

\begin{theorem}[Quantum grinding attack]\label{thm:attack}
Let $d=\lceil\log_2 N\rceil$ and $c_{\mathsf{ev}}=d+1+\kr$. Suppose the adversary corrupts $C$ with $a=w(C)\le\Wp$, receives no honest signatures, and some position outside $C$ or some padding position exists. Fix any $\Ws$ with $(1+\delta)\Wp\le\Ws\le W$, let $\beta=|\{y<2^{128}: y\bmod\Ws<a\}|/2^{128}$ and $\varepsilon_C=\beta^{\kr}$; then $\varepsilon_C\ge(a/\Ws)^{\kr}$. If $q\ge 2^{d+1}+c_{\mathsf{ev}}$, there is a quantum adversary making at most $q$ queries that wins $\mathsf{CERT\text{-}FORGE}$ with probability at least
\[
\frac{4}{\pi^2}(2T_G+1)^2\varepsilon_C\Big[1-(2T_G+1)^2\big(\varepsilon_C+(d+2)2^{-h}\big)\Big],
\]
where $T_G=\lfloor(q-2^{d+1}-c_{\mathsf{ev}})/(2c_{\mathsf{ev}})\rfloor$.
\end{theorem}
\begin{proof}[Proof sketch]
The adversary packs the intervals of $C$ into $[0,a)$ (only revealed intervals are checked) and salts one other position; a salt is marked if all coins of the resulting root fall below $a$. Grover search finds one; the marked fraction $\nu$ has $\mathbb E[\nu]=\varepsilon_C$ and $\mathbb E[\nu^2]\le\varepsilon_C^2+\varepsilon_C(d+2)2^{-h}$, which bounds the loss when $\nu$ is large. The full proof is in \supp{S3}.
\end{proof}

\begin{corollary}[Classical parameters are insufficient]\label{cor:loss}
Consider the classical rule $\kr^{\mathsf{cl}}=\lceil(\lambda+\log_2 q)/\log_2(\Ws/\Wp)\rceil$ and $w(C)=\Wp$. The grinding term of \cref{thm:qrom} then guarantees only about $\lambda-\log_2 q-\log_2 12$ bits, and the attack of \cref{thm:attack} succeeds with probability about $(4/\pi^2)\,q^2(\Wp/\Ws)^{\kr}/c_{\mathsf{ev}}^2$. For $\lambda=128$, $q=2^{64}$, $d=16$, $\Ws/\Wp=2.55$ and $\kr=143$, the guarantee is $61.5$ bits ($60.4$ before rounding $\kr$ up) and the attack succeeds with probability $2^{-81.1}$. Both are far from the intended $2^{-128}$.
\end{corollary}

\noindent The upper bound and the attack differ by a factor of about $3\pi^2c_{\mathsf{ev}}^2$ ($2^{19.5}$ at these parameters), which is independent of $q$, so the quadratic dependence on $q$ is unavoidable. The attack needs corrupt weight close to $\Wp$: at $\Ws=0.85W$ with $w(C)=0.30W$ it reaches only $2^{-102.8}$. Depth limits~\cite{nistcfp2016} reduce Grover's gain in practice, so our rule is conservative. Algorand's target $k+2q$~\cite{algorand_stateproofs} matches our rule up to $\kappa_Q$ and, modelling its hashes as random oracles, certifies about $124$ bits; the classical $k+q$ rule of~\cite{micali2020compact}, used by the prototype of~\cite{NIROPoK2025}, falls under \cref{cor:loss} (\supp{S3}).

\subsection{Parameter Rule and Security Budget}\label{sec:rule}
Write $\kappa_Q=\lceil\log_2 12\rceil=4$. Requiring $12t^2\varepsilon^\ast\le 2^{-\lambda}$ in \cref{thm:qrom}, with $t\le q(1+2^{-20})$ for $q=2^{64}$, gives the reveal rule
\begin{equation}\label{eq:rule}
\kr^{\mathsf{Q}}=\Big\lceil\frac{\lambda+2\log_2 q+\kappa_Q}{\log_2(\Ws/\Wp)}\Big\rceil ,
\end{equation}
since $12(1+2^{-20})^2(1+2^{-51})<16$ for $\Wp\ge W/4$. It controls only the \emph{grinding} term; the exact rule for a total query budget $T$ (all L1-path hash queries) is in \supp{S3}.

The other terms need separate budgets; the commitment term, with $h$ the \emph{smallest} output width, is an upper bound like the generic collision bound~\cite{bht1998}, not an attack (\supp{S3}). \Cref{tab:budget} combines all terms of \cref{thm:e2e} for the deployed Keccak-256 vault and a widened 384-bit variant (not implemented). Keccak-384 certificate trees enlarge Dilithium2 (r3) certificates by $4$--$7\%$ at unchanged verification time ($N=2^{12}$--$2^{16}$, \cref{tab:h384}).

\begin{table}[t]
\centering
\caption{Lifetime security budget of \cref{thm:e2e} in the query model, for $T\le 2^{64}(1+2^{-20})$ total queries (honest and verification queries included) and $E=2^{20}$ epochs. Values are $-\log_2$ of upper bounds; $T$ is as in \cref{sec:rule}.}
\label{tab:budget}
\scriptsize
\setlength{\tabcolsep}{2pt}
\begin{tabular}{@{}>{\raggedright\arraybackslash}p{0.56\columnwidth}cc@{}}
\toprule
Term & \makecell{256-bit\\(deployed)} & \makecell{384-bit}\\
\midrule
Grinding $E\cdot12T^2\varepsilon^\ast$, \cref{eq:rule} with $\lambda=128$ & $108$ & $108$\\
Grinding, rule with $\lambda+\log_2E=148$ ($\kr=208$ at $2.55$) & $128$ & $128$\\
Commitments $E(60T^3+2q_V+2q_C)/2^h$ & $38.1$ & $166.1$\\
Signatures $2E\Adv^{\mathsf{uf}}_{N_{\mathsf{tot}}}+\Adv^{\mathsf{uf}}_U$ & \multicolumn{2}{c}{ML-DSA assumption}\\
Finality $\Pr[\mathsf{Fin}]$ & \multicolumn{2}{c}{gadget assumption}\\
\bottomrule
\end{tabular}
\end{table}

\noindent The deployed bound is dominated by the 256-bit commitment term, shared by every Keccak-256 Merkle structure on Ethereum. A lifetime target of $2^{-128}$ for every term requires both the widened trees and the rule with $\lambda+\log_2 E$; our prototype targets $\lambda=128$ for the grinding term per checkpoint and does not claim a 128-bit lifetime custody guarantee.

With $\lambda=128$ and $q=2^{64}$, the rule increases $\kr$ by a factor of $260/192=1.35$ over the classical rule; because certificate size grows with the number of \emph{distinct} reveals, the measured size overhead is smaller (\cref{sec:eval}). \Cref{tab:kr} lists $\kr$ for typical ratios. With $\Wp=\lfloor W/3\rfloor$ and $85\%$ participation, $\Ws/\Wp=2.55$, whereas certifying a two-thirds threshold ($\Ws/\Wp=1.28$) multiplies the cost by almost $4$. This is why we certify finality, which needs only $\Wp=\lfloor W/3\rfloor$ under $f<1/3$, rather than a supermajority vote.

\begin{table}[t]
\centering
\caption{Reveal count $\kr$ for $\lambda=128$ and $\log_2 q=64$.}
\label{tab:kr}
\scriptsize
\begin{tabular}{@{}lcccccc@{}}
\toprule
$\Ws/\Wp$ & 1.25 & 1.5 & 2 & 2.55 & 3 & 4\\
\midrule
classical~\cite{micali2020compact} & 597 & 329 & 192 & 143 & 122 & 96\\
QROM, \cref{eq:rule} ($\kappa_Q{=}4$) & 808 & 445 & 260 & 193 & 165 & 130\\
\bottomrule
\end{tabular}
\end{table}

\noindent\textbf{Implementation pitfalls.}\label{sec:impl} Three omissions in the prototype of~\cite{NIROPoK2025} each made certificates forgeable: (i)~coin coverage was not checked; (ii)~interval endpoints were accumulated over all attestors rather than signers and were not checked; and (iii)~a heuristic $\kr\propto(1-\Wp/\Ws)$ decreased as $\Ws\to\Wp$, whereas soundness requires $\kr\propto 1/\log(\Ws/\Wp)$. A fourth check, $|P|\le\kr$ with every reveal covering a coin, keeps the verifier's query count bounded as \cref{thm:qrom} assumes. \Cref{alg:cert} includes all four, and our implementation tests every check negatively.

\section{Leader Election: Fairness and Steering}\label{sec:consensus}

Total stake is $1$; the adversary holds stake $\alpha$ (possibly split across identities, including offline honest validators). We analyze an \emph{ideal experiment} with independent uniform scores and prove that the QPoS-S seed of \cref{sec:design} realizes it against quantum adversaries. Full proofs are in \supp{S4}.

\subsection{Stake-Weighted Scores Form an Exponential Race}

\begin{lemma}[Exponential race]\label{lem:race}
Let $U_1,\dots,U_n$ be independent and uniform on $(0,1]$, and let $s_i=-\log_2 U_i/w_i$. Then the $s_i$ are independent, $s_i\sim\mathrm{Exp}(w_i\ln 2)$, so (i)~$\Pr[i\text{ has the lowest score}]=w_i/W$, and (ii)~splitting stake $w$ into identities $w',w''$ with $w'+w''=w$ leaves the coalition's winning probability unchanged (Sybil invariance). With $h$-bit values $u_i=\lceil U_i2^h\rceil/2^h$ and ties broken by the smaller identity, the winner differs from the continuous winner with probability at most $N(N-1)2^{-h}$.
\end{lemma}
\begin{proof}[Proof sketch]
$-\ln U\sim\mathrm{Exp}(1)$, and the minimum of exponentials is attained at $i$ with probability proportional to its rate. The full proof is in \supp{S4}.
\end{proof}

\noindent\Cref{lem:race} applies the weighted-sampling identity of~\cite{efraimidis2006} to the score of~\cite{QRL}; it concerns \emph{unselected} uniforms only.

\subsection{Steering Through Proposer-Dependent Seeds}
In~\cite{QRL,myeprint}, scores derive from private hash-chain values, the seed is updated as $\cS_{j+1}=\Hash(\cS_j\,\|\,\text{reveal of the \emph{actual} proposer of }j)$, and a skipped leader is replaced by the runner-up. A coalition holding the top $L$ positions can thus choose among $L+1$ seeds, knowing its \emph{own} next-slot scores (\supp{S4}).

\begin{proposition}[$m$-way steering]\label{prop:steer}
Suppose that, before slot $T$, the adversary can choose $\cS_T$ among $m$ candidates; that under each candidate its own best score is a fresh $\mathrm{Exp}(\alpha)$ variable (rates in units of $\ln 2$), independent across candidates, while the honest best score is $\mathrm{Exp}(1-\alpha)$ and unobserved. Picking the candidate that minimizes its own best score is optimal and leads slot $T$ with probability $p_m(\alpha)=m\alpha/(1+(m-1)\alpha)$.
\end{proposition}
\begin{proof}[Proof sketch]
As a function of its best score, the adversary's winning probability is $\mathrm{Beta}(\alpha/(1-\alpha),1)$-distributed; the maximum of $m$ copies has mean $p_m(\alpha)$ (\supp{S4}).
\end{proof}

\begin{corollary}[Steering under the runner-up rule]\label{cor:qrl}
Consider proposer-dependent seeds in the race without threshold of~\cite{QRL}, and a single-identity adversary that uses only its decision at slot $T-1$. (a)~If it observes the honest runner-up's reveal before its publish-or-skip decision is irrevocable (e.g., rushing), it captures slot $T$ with probability $\alpha+\alpha^2(1-\alpha)/(1+\alpha)$. (b)~Otherwise, it captures slot $T$ with probability $\alpha\big(\alpha+(1-\alpha)\alpha^{1/(1-\alpha)}\big)+(1-\alpha)\alpha$. Both exceed $\alpha$. Splitting into more identities raises the probability of holding the top $L\ge 2$ ranks, hence the option count, so Sybil invariance fails.
\end{corollary}
\begin{proof}[Proof sketch]
Leading slot $T-1$, the adversary applies \cref{prop:steer} with $m=2$ in case~(a), and gains $\mathbb E[\max(p,\alpha)]>\alpha$ against a fresh seed in case~(b) (\supp{S4}).
\end{proof}
\noindent For $\alpha=0.3$ the capture probability is $0.348$ ($+16\%$) in case~(a) and $0.338$ in case~(b); $p_{16}(0.1)=0.64$. Target slots matter because roles attach to slot indices, for example the $N_{\mathsf{cp}}$-th block that fixes a checkpoint, or the slot whose block includes or censors a dispute-relevant transaction.

\subsection{Epoch Seeds: Ideal Model}

In QPoS-S (\cref{sec:design}), each validator reveals its committed secret $r_i$ in the previous epoch's reveal window, in an order drawn from the previous seed; the epoch seed hashes all reveals. Throughout, $\alpha<1/3$.

In each slot, the scores rank all identities as a stake-weighted sample without replacement (\cref{lem:race},~\cite{efraimidis2006}). Ranks $1,\dots,L+1$ may propose; the best-ranked one whose block arrives in time proposes. The adversary may suppress at most $B\le L$ honest identities per slot. Call a slot \emph{guaranteed} if its first $B+1$ ranks are honest; a guaranteed slot has an honest proposer whatever the adversary suppresses. Let $w_{\max}$ be the largest honest stake, $\alpha_B=\alpha/(1-Bw_{\max})$ and $p_B=1-(1-\alpha_B)^{B+1}$; without denial of service, $p_0=\alpha$.

Let $i^\ast$ be the best-ranked honest identity of the tail (it reveals last among honest tail identities), $\tau^\ast$ the first time $r_{i^\ast}$ leaves validator $i^\ast$, and $c$ the number of (adversarial) identities ranked above $i^\ast$, which reveal after it. The adversary decides which of them reveal, so the reveal set takes at most $m=2^c$ values. Let $\mathsf{Fail}$ be the event that the tail contains no honest identity, or that the common phase or some tail sub-window has no guaranteed slot after its first slot. It depends only on the order and the previous schedule, both fixed before $\tau^\ast$.

\begin{definition}[Ideal seed experiment]\label{def:ideal}
An experiment is \emph{$m$-ideal} if, at time $\tau^\ast$, it samples for each of the at most $m$ possible reveal sets $b$ an independent uniform table of scores and next-epoch order keys, independent of everything before $\tau^\ast$; it makes all tables available to the adversary; and every party uses the table of the reveal set that actually results.
\end{definition}

\begin{theorem}[Leader election in the ideal model]\label{thm:schedule}
In any $m$-ideal experiment with continuous scores: (a)~Given the reveal set, the proposer ranking of each slot is a fixed public function of the table. No adversarial action within the epoch (skipping, withholding proposals, choosing among its identities) changes any ranking, except by removing the adversary's own identities. If $m=1$, slot leaders are independent across slots and distributed as in \cref{lem:race}. (b)~The adversary chooses the table among at most $m\le 2^c$ candidates by withholding a subset of its late reveals. A slot is not guaranteed with probability at most $p_B$ under each candidate, and the adversary can propose only in such slots. It proposes in a fixed target slot with probability at most $1-(1-p_B)^m$, and its expected fraction of the $k$ slots of the epoch is at most $p_B+\sqrt{c\ln 2/(2k)}$. (c)~The probability that none of $n$ specified slots is guaranteed is at most $\min\{1,mp_B^{n}\}$. For a window of $n$ consecutive slots whose start the adversary may choose after seeing the schedule, it is at most $\min\{1,(k-n+1)mp_B^{n}\}$.
\end{theorem}
\begin{proof}[Proof sketch]
Each of the first $B+1$ ranks is adversarial with probability at most $\alpha_B$; part~(b) uses independence across candidates, part~(c) a union bound over candidates and window starts, and the share uses the sub-Gaussian maximal inequality (\supp{S4}).
\end{proof}

\noindent Fallback proposers matter only under denial of service. Without them ($L=0$), suppressing the leader of every slot leaves the sidechain without honest blocks; with $L\ge B$, only non-guaranteed slots can lose their honest proposer. Part~(c) deliberately does not multiply per-slot probabilities under the selected table, which would be invalid after seed selection.

\begin{lemma}[Reveals are included]\label{lem:incl}
Suppose that an honest proposer includes every reveal it received $\Delta_{\mathcal S}$ before its slot, that messages are delivered within $\Delta_{\mathcal S}$ and a slot lasts at least $(L+2)\Delta_{\mathcal S}$, and that the finality gadget finalizes each block of the reveal window before the next slot starts. If a phase or sub-window of $\ell$ slots has a guaranteed slot after its first slot, every honest reveal of that phase or sub-window is included. Under each candidate table of the epoch that contains the window, this fails with probability at most $p_B^{\ell-1}$.
\end{lemma}
\begin{proof}[Proof sketch]
A reveal sent when its window starts reaches every later honest proposer, a guaranteed slot has one, and slots are independent under each candidate (\supp{S4}).
\end{proof}

\begin{proposition}[Late reveals are rare]\label{prop:late}
Consider consecutive $m$-ideal epochs, each with respect to its own $c$, starting from an epoch with $c=0$. Let $\mu=(1-2\alpha)/(1-3\alpha)$, and let $\mathsf{OK}_e$ be the event that $\mathsf{Fail}$ occurs in none of the epochs up to $e$. Then, for every epoch $e$ and every $x\ge1$,
\begin{gather*}
\Pr[c_e\ge x\wedge\mathsf{OK}_e]\le\mu\alpha^x,\qquad\mathbb E[2^{c_e}\mathbf 1_{\mathsf{OK}_e}]\le\mu,\\
\Pr[\mathsf{OK}_{e-1}\wedge\neg\mathsf{OK}_e]\le\delta_{\mathsf{LE}},
\end{gather*}
where $\delta_{\mathsf{LE}}=\mu(\alpha^{L_t}+L_tp_B^{\ell_s-1}+p_B^{\ell_0-1})$. The adversary proposes in a fixed target slot of epoch $e$ with probability at most $\pi_{\mathsf{cap}}+e\delta_{\mathsf{LE}}$, where $\pi_{\mathsf{cap}}=p_B+\sum_{x\ge1}\min\{1,\mu\alpha^x\}\big((1-p_B)^{2^{x-1}}-(1-p_B)^{2^x}\big)$.
\end{proposition}
\begin{proof}[Proof sketch]
By mild adaptivity, registrations and the corrupt set of epoch $e$ are fixed before the candidate tables of epoch $e-1$ are sampled, and under each candidate the first $x$ ranks of the order are adversarial with probability at most $\alpha^x$. Summing over the $m_{e-1}$ candidates gives $\mathbb E[2^{c_e}\mathbf 1_{\mathsf{OK}_e}]\le1+\mathbb E[2^{c_{e-1}}\mathbf 1_{\mathsf{OK}_{e-1}}]\,\alpha/(1-2\alpha)$, whose fixed point is $\mu$ (\supp{S4}).
\end{proof}

\noindent For $B=0$ and $\alpha=0.1$, $0.2$, $0.3$: $\mu=1.14$, $1.5$, $4$; $\pi_{\mathsf{cap}}=0.112$, $0.265$, $0.622$; and, as $\mathbb E[c]\le\mu\alpha/(1-\alpha)$, the expected share exceeds $p_B$ by at most $\sqrt{\mathbb E[c]\ln 2/(2k)}\le 0.005$, $0.008$, $0.017$ for $k=2000$. Without the random order, owning the last $c=4$ positions would give $1-(1-\alpha)^{16}=0.997$ at $\alpha=0.3$; with it, $\pi_{\mathsf{cap}}$ still exceeds the $0.348$ of \cref{cor:qrl}, as one adversarial last revealer yields two candidate schedules. With $L_t=25$ and $\ell_s=\ell_0=52$, $\delta_{\mathsf{LE}}\le2^{-41.4}$ at $\alpha=0.3$ without denial of service and $\delta_{\mathsf{LE}}\le2^{-40.7}$ with $B=1$ ($w_{\max}\le0.01$, $p_1=0.514$).

\subsection{From the Ideal Model to the Protocol, Against Quantum Adversaries}

\begin{lemma}[Reprogramming hidden inputs~\cite{bbbv1997}]\label{lem:bbbv}
Let a quantum process make $q$ queries to an oracle $O$, and let its run with $O$ be determined by randomness independent of a random variable $Z$, so that, with that randomness fixed, its state before each query does not depend on $Z$. Let $S_Z$ be a set of inputs, determined by $Z$, with $\Pr_Z[x\in S_Z]\le p$ for every input $x$, and let $O'$ agree with $O$ outside $S_Z$. Then the final states of the process with $O$ and with $O'$ have expected trace distance at most $2q\sqrt{p}$.
\end{lemma}
\begin{proof}[Proof sketch]
Each query changes the state by at most $2\|\Pi_{S_Z}\lvert\phi_i\rangle\|$, of expectation at most $\sqrt p$ as $\lvert\phi_i\rangle$ is independent of $Z$ (\supp{S4}).
\end{proof}

\begin{theorem}[Ideal to real, QROM]\label{thm:fresh}
Model the commitment and seed functions of QPoS-S, cSHAKE256 with distinct tags, as a quantum random oracle independent of the score oracle $\RO$ and of ML-DSA's internal hashing. Let commitments, secrets and seeds have $n_c$, $n_r$ and $n_s$ bits, and let identities and epochs be fixed-length fields of $n_p$ bits in total. Assume mildly adaptive corruption (\cref{sec:model}), the conditions of \cref{lem:incl}, that $\mathsf{Fail}$ does not occur, and that $c\le\bar c$; on $\neg\mathsf{Fail}$, $\bar c=L_t-1$ suffices. Let $m=2^{\bar c}$, and let $t$ count all oracle queries of all parties before $\tau^\ast$. Then there is an $m$-ideal experiment such that, for every $[0,1]$-valued function $F$ of the epoch's leaders and the adversary's final output,
\begin{multline*}
\big|\mathbb E_{\mathsf{real}}[F]-\mathbb E_{\mathsf{ideal}}[F]\big|\le\epsilon_{\mathsf{lead}}\\=\frac{2^{n_p+2n_r}}{2^{n_c+1}}+\frac{2t}{2^{n_r/2}}+\frac{2t\sqrt m}{2^{n_s/2}}+\frac{m^2}{2^{n_s+1}}.
\end{multline*}
Queries made after $\tau^\ast$ do not appear in the bound.
\end{theorem}
\begin{proof}[Proof sketch]
Four hybrids give the four terms: (1)~statistical binding of the wide commitments leaves late adversaries only reveal-or-withhold; (2)~\cref{lem:bbbv} with $Z=r_{i^\ast}$ hides $r_{i^\ast}$ before $\tau^\ast$; (3)~\cref{lem:bbbv} on the pre-sampled seeds $Z_b$ of the at most $m$ reveal sets makes their scores and order keys fresh; (4)~the $Z_b$ are distinct except with probability $m^2/2^{n_s+1}$, giving an $m$-ideal experiment. Before $\tau^\ast$, validator $i^\ast$ uses $r_{i^\ast}$ only through its commitment, which is uniform for every $r_{i^\ast}$, so the run before $\tau^\ast$ is independent of the reprogrammed inputs; the run after $\tau^\ast$ is the same map in all hybrids, so the trace distance at $\tau^\ast$ bounds the final distance and later queries are free (\supp{S4}).
\end{proof}

\noindent QPoS-S uses $n_r=n_s=512$, $n_c=1536$ and $64$-bit identity and epoch fields ($n_p=128$). For $t=2^{64}$ and $\bar c=L_t-1=24$, the four terms are $2^{-385}$, $2^{-191}$, $2^{-179}$ and $2^{-465}$, so $\epsilon_{\mathsf{lead}}\le 2^{-178.9}$; $\bar c=24$ is a worst case. Scores and order keys use $\Hash$ with $h=256$.

\begin{corollary}[Real protocol]\label{cor:real}
Let $m=2^{L_t-1}$ and $\epsilon_{\mathsf{LE}}=\epsilon_{\mathsf{lead}}+\epsilon_{\mathsf{disc}}$, where $\epsilon_{\mathsf{disc}}=(k+1)mN(N-1)2^{-h}$ applies \cref{lem:race} to the $k$ slot races and the order race of each candidate, and $N$ counts all registered identities. Against quantum adversaries with mildly adaptive corruption, under the conditions of \cref{lem:incl}, every bound of \cref{thm:schedule,prop:late,lem:incl} holds in the real protocol, for the first $E$ epochs, up to an additive $2E(\epsilon_{\mathsf{LE}}+\delta_{\mathsf{LE}})$. In particular, the adversary proposes in a fixed target slot with probability at most $\pi_{\mathsf{cap}}+2E(\epsilon_{\mathsf{LE}}+\delta_{\mathsf{LE}})$, and an epoch has no guaranteed slot among $n$ consecutive slots placed by the adversary with probability at most $(k-n+1)\mu p_B^{n}+2E(\epsilon_{\mathsf{LE}}+\delta_{\mathsf{LE}})$.
\end{corollary}
\begin{proof}[Proof sketch]
Idealizing epochs one at a time costs $\epsilon_{\mathsf{LE}}+\delta_{\mathsf{LE}}$ each, plus $\Pr[\neg\mathsf{OK}_E]\le E\delta_{\mathsf{LE}}$ (\supp{S4}).
\end{proof}

\noindent\textbf{Consequences.} By \cref{cor:real}, an honest transaction is included within $n$ slots of any adversarially chosen start, except with probability $\zeta+2E(\epsilon_{\mathsf{LE}}+\delta_{\mathsf{LE}})$, whenever $n\ge\ln(k\mu/\zeta)/(-\ln p_B)$, as used by \cref{prop:binding}. For $\alpha=0.3$, $k=2000$ and $\zeta=2^{-40}$, $n=31$ suffices without denial of service and $n=56$ with $B=1$; for $E=2^{20}$, $2E(\epsilon_{\mathsf{LE}}+\delta_{\mathsf{LE}})<2^{-19}$ at $\alpha=0.3$ and $<2^{-36}$ at $\alpha=0.2$, with or without $B=1$. Finality additionally needs the finality gadget's liveness, and a window that straddles an epoch boundary adds the bounds of both epochs. No PQ-Vault role is exclusive to a slot: anyone can aggregate certificates, and attestors, not leaders, certify checkpoints.

\noindent\textbf{Long-range attacks.} With $\Vault$ as the weak-subjectivity source, a fork before the latest anchored checkpoint is rejected under \cref{asm:l1}, and a fork after it needs a certificate for the still-bonded attestor set, which \cref{thm:qrom,thm:opf} bound (\supp{S4}).

\section{Accepting Certificates on L1}\label{sec:opf}

Direct L1 verification ($\Theta(\kr)$ ML-DSA verifications on calldata) costs $1.8\times10^{8}$ to $1.1\times10^{9}$ gas, above the per-transaction cap (\cref{sec:gas}), so certificates must be compressed.

\subsection{A Non-PQ SNARK Cannot Be the Acceptance Path}
Our earlier design~\cite{myeprint} used a Groth16~\cite{groth16} proof over certificates; this is unsound once L1 \emph{accepts} checkpoints on its strength.

\begin{proposition}\label{prop:snark}
If $\Vault$ finalizes a checkpoint whenever a Groth16 proof $\Pi$ of ``all certificates verify'' is valid, then a QPT adversary can finalize an arbitrary checkpoint, and with it arbitrary withdrawals, with probability $1-\negl$.
\end{proposition}
\begin{proof}[Proof sketch]
Shor's algorithm recovers the CRS trapdoor, which lets the simulator prove false statements. The full proof is in \supp{S5}.
\end{proof}

\noindent We therefore use an optimistic protocol with a PQ dispute path, whose happy path costs almost nothing; a transparent hash-based proof system is the alternative.

\subsection{Optimistic Acceptance with Single-Reveal Fraud Proofs}
Apart from the range check on $\Ws$ (done at proposal), $\mathsf{Vf}$ (\cref{alg:cert}) is a \emph{conjunction of per-coin checks}. An L1 attestor leaf holds $(\eta_i,w_i)$, where $\eta_i$ is the \texttt{EXTCODEHASH} of the contract storing $i$'s expanded ML-DSA key (\cref{sec:gas}). For each coin index $j$, the \emph{record} $r_j=(j,i_j,\eta_{i_j},w_{i_j},\sigma_{i_j},L_{i_j},\pi^{\mathsf{sig}}_{j},\pi^{\mathsf{att}}_{j})$ holds the covering leaf and its two paths, and $\mathsf{Local}(j,r_j)=1$ iff
(i)~$i_j<N$, $w_{i_j}>0$ and $L_{i_j}\le c_j<L_{i_j}+w_{i_j}\le \Ws$, with $c_j$ computed from the proposal's $M_e$, $\Rsig$, $\Ws$ and its snapshot of $\Ratt$, $\Wp$ and $N$;
(ii)~$\pi^{\mathsf{sig}}_j$ and $\pi^{\mathsf{att}}_j$ have length $d$ and open $(\sigma_{i_j},L_{i_j})$ under $\Rsig$ and $(\eta_{i_j},w_{i_j})$ under $\Ratt$, both at position $i_j$; and
(iii)~$\sigma_{i_j}$ has the scheme's length and verifies on $M_e$ under the key stored in the contract with code hash $\eta_{i_j}$.
The relayer commits to the records in a Merkle tree $C_{\cert}$ of depth $\lceil\log_2\kr\rceil$, with $r_j$ at index $j$ (leaf encoding in \supp{S5}).

\begin{lemma}[Local decomposition]\label{lem:local}
Except for a collision of $\Hash$, $\mathsf{Local}(j,r_j)=1$ for all $j<\kr$ if and only if the distinct positions $P=\{i_j\}$, their leaves, and the multiproofs assembled from the single paths form a certificate that \cref{alg:cert} accepts, instantiated with the vault's leaf encodings.
\end{lemma}
\begin{proof}[Proof sketch]
Without collisions, records naming one position carry equal leaves, and the checks of \cref{alg:cert} are exactly (i)--(iii) over all $j$. The full proof is in \supp{S5}.
\end{proof}

\noindent\textbf{Protocol.} $\mathsf{propose}$ (\cref{sec:vault}) creates a bonded proposal at $\tau_0$, with records in blobs (not relied on for safety). Until the \emph{challenge deadline} $\tau_0+T_{\mathsf{ch}}$ (extended to $\tau_{\mathsf{post}}+\Delta_{\mathsf{post}}$ after an open-all posting completed at $\tau_{\mathsf{post}}$), anyone may $\mathsf{challenge}(j)$ each index once, with bond $b_C$. A $\mathsf{respond}$ within $\Delta_{\mathsf{resp}}$ carries one record: if it is unauthenticated under $C_{\cert}$ or names a wrong key address, the call reverts; otherwise it rejects the proposal if $\mathsf{Local}$ fails and closes the challenge if not. $\mathsf{challengeAll}$ (once, no bond) forces all records onto calldata within $\Delta_{\mathsf{open}}$; $\mathsf{timeout}$ and $\mathsf{timeoutAll}$ reject on an unanswered challenge or incomplete posting. $\mathsf{finalize}$, with no open challenge or pending posting, finalizes if the epoch is $e^\ast+1$, else marks it \emph{stale}. Guards are $<$ the deadline for challenges and open-all, $\le$ for responses and postings, $\ge$ for $\mathsf{finalize}$ and $>$ for timeouts, so no timestamp admits both an action and the transition its deadline unlocks. Bonds and rewards ($b_C$, $b_R$) and the full state machine are in \supp{S5}.

\noindent\textbf{Timing.} With $\Delta_{\mathsf{chk}}$ bounding a watcher's time to check all records and $\Delta_{\mathsf{prep}}$ the relayer's time to answer, we require
\[
T_{\mathsf{ch}}> \Delta_{\mathsf{chk}}+\Delta_{\mathsf{inc}},\qquad \Delta_{\mathsf{post}}> \Delta_{\mathsf{chk}}+\Delta_{\mathsf{inc}},
\]
and $\Delta_{\mathsf{open}},\Delta_{\mathsf{resp}}\ge\Delta_{\mathsf{prep}}+\Delta_{\mathsf{inc}}$. Each pending proposal is then eligible for a final transition by $T_{\mathsf{fin}}=T_{\mathsf{ch}}+\Delta_{\mathsf{open}}+\Delta_{\mathsf{post}}+\Delta_{\mathsf{resp}}$ after its inclusion. With $T_{\mathsf{ch}}=7$ days, $\Delta_{\mathsf{open}}=\Delta_{\mathsf{resp}}=1$ day and $\Delta_{\mathsf{post}}=2$ days (as in optimistic rollups~\cite{kalodner2018arbitrum}), $T_{\mathsf{fin}}\le 11$ days.

\begin{theorem}[Optimistic acceptance]\label{thm:opf}
Assume \cref{asm:l1,asm:pq,asm:slow,asm:watcher} and the timing conditions above. Then the following hold. (Safety) The probability that $\Vault$ finalizes a proposal whose $M_e$ meets the winning condition of \cref{def:game}, with the attestor list committed by the proposal's snapshot of $\Ratt$, is at most the bound of \cref{thm:qrom} plus $\delta_{\mathsf{col}}=(12t^3+2q_C)/2^{h_C}$, where $q_C=2\kr(\lceil\log_2\kr\rceil+2d+7)$ and $h_C$ is the smallest output width of $C_{\cert}$, $\Rsig$ and $\Ratt$ ($h_C=256$ in the deployed vault). Since $\mathsf{Col}$ ranges over the whole oracle, $\delta_{\mathsf{col}}$ also excludes two key contracts with the same code hash $\eta$, and, in \cref{thm:e2e}, two headers with the same $M_e$ or two withdrawal trees with the same root. (Liveness) An honest relayer's proposal is never rejected except with probability $\delta_{\mathsf{col}}$; it is eligible for finalization at $\tau_0+T_{\mathsf{ch}}$ if no dispute is raised, and by $\tau_0+T_{\mathsf{fin}}$ otherwise; if the relayer then submits $\mathsf{finalize}$, it is finalized within a further $\Delta_{\mathsf{inc}}$. (Cost) Each frivolous coin challenge costs its sender $b_C$, and an honest watcher spends at most one open-all transaction fee and one bond $b_C$ per proposal; if its challenge is the first on a failing record, the bond is returned with $b_R$.
\end{theorem}
\begin{proof}[Proof sketch]
Excluding collisions (\cref{lem:extract}), by \cref{asm:watcher} and the strict timing conditions, an honest watcher learns all committed records (from blobs or via open-all) and challenges any failing record before the deadline; the only non-reverting response rejects the proposal, and an unanswered challenge rejects it at $\mathsf{timeout}$. A finalized proposal thus has only passing records, an accepted certificate by \cref{lem:local}, bounded by \cref{thm:qrom}. Liveness: honest records pass and answers meet the inclusive deadlines. The full proof is in \supp{S5}.
\end{proof}

\subsection{End-to-End Custody}

\begin{theorem}[Custody safety]\label{thm:e2e}
Assume \cref{asm:l1,asm:pq,asm:slow,asm:watcher}, the timing conditions above, and, in every epoch $e$, corrupt attestor stake below $W^e/3$ with $\Wp^e=\lfloor W^e/3\rfloor$. Assume also the honest-signing rule of \cref{sec:design}, and let $\mathsf{Fin}$ be the event that two honest validators consider conflicting checkpoints of the same epoch final. Then over $E$ epochs, for every QPT adversary making $t$ oracle queries in total,
\begin{multline*}
\Pr[\text{G1 violated}]\le E\Big(12t^2\varepsilon^\ast+\frac{48t^3+2q_V}{2^h}+\delta_{\mathsf{col}}\Big)\\+2E\Adv^{\mathsf{uf}}_{N_{\mathsf{tot}}}+\Adv^{\mathsf{uf}}_{U}+\Pr[\mathsf{Fin}],
\end{multline*}
where $\varepsilon^\ast\le 2^{-\lambda}/(12t^2)$ under \cref{eq:rule} when $t\le q(1+2^{-20})$ (\cref{sec:rule}), $N_{\mathsf{tot}}$ is the number of honest attestor keys generated over the whole experiment, and $U$ is the number of user keys.
\end{theorem}
\begin{proof}[Proof sketch]
By induction over checkpoints we keep the invariant $I_e$: finalized $M_1,\dots,M_e$ are honest, the vault's $(\Ratt,W,N)$ is the honest one of epoch $e+1$, and each payment pays a leaf of a finalized $\omega_{e'}$ at most once. A proposal for epoch $e$ uses the honest snapshot by $I_{e-1}$; if no honest attestor signed $M_e$, \cref{thm:opf} bounds its finalization, and otherwise the honest-signing rule makes $M_e$ honest outside $\mathsf{Fin}$. Nullifiers, balance conservation and user-signature unforgeability bound payments by the escrow. The full proof is in \supp{S5}.
\end{proof}

\noindent The union bound over epochs is why the reveal rule needs $\lambda+\log_2E$ for a lifetime target (\cref{sec:rule}); $\mathsf{Fin}$ is the safety of the sidechain's finality gadget, which we do not re-analyze (\supp{S5}).

\noindent\textbf{Liveness (G4).} If honest transactions are included on the sidechain within $\Delta_{\mathsf{sc}}$ (except with probability $\zeta+2E(\epsilon_{\mathsf{LE}}+\delta_{\mathsf{LE}})$, \cref{cor:real}) and some honest party relays each checkpoint, an honest withdrawal is paid within $\Delta_{\mathsf{sc}}+\Delta_{\mathsf{cp}}+\Delta_{\mathsf{inc}}$, where $\Delta_{\mathsf{cp}}$ includes at most $2(T_{\mathsf{fin}}+\Delta_{\mathsf{inc}})$ for finalization, including the predecessor's wait when the checkpoint interval is at least $T_{\mathsf{fin}}+\Delta_{\mathsf{inc}}$ (\cref{sec:flow}, \cref{thm:opf,asm:l1}), if an honest relayer calls $\mathsf{finalize}$ once eligible.

\section{Implementation and Evaluation}\label{sec:eval}

\noindent\textbf{Artifacts and setup.} \texttt{pqcert} implements \cref{alg:cert} in Rust with the QROM rule (\cref{eq:rule}) and Keccak-256 Merkle multiproofs, over the round-3 Dilithium implementation used by the prototype of~\cite{NIROPoK2025} (labeled Dilithium; not interoperable with FIPS~204, which a second backend cross-checks). We also built a Python verifier, a network harness, Monte Carlo simulators and the Solidity vault of \cref{sec:gas}, all public~\cite{pqvault_artifact} (\supp{S6}). The verifier rejected all negative tests in every run. Timings are single-threaded medians on a 2-vCPU cloud VM (\emph{M1}). Unless stated otherwise, $\lambda=128$, $\log_2 q=64$, $\kappa_Q=4$, $\delta=0.1$, $\Wp=\lfloor W/3\rfloor$ (equivalently $\lfloor 0.3333\,W\rfloor$ in the size runs, \supp{S6}), stake is uniform, log-normal ($\sigma=1.5$) or Zipf ($s=1$), and validators sign in random order until at least $85\%$ of the stake has signed (under Zipf, up to $99\%$ at $N=2^8$).

\subsection{Certificate Size and Cost}
\Cref{fig:size} and \cref{tab:cert} report certificates with real round-3 Dilithium2 signatures.
\emph{(1)~Compactness.} At fixed $\lambda$, $q$ and $\Ws/\Wp$, $\kr$ does not depend on $N$, and the size $O(\kr(|pk|+|\sigma|)+h\kr\log N)$ grows logarithmically in $N$ while $\Ws/\Wp$ stays bounded away from $1$: from $0.36$\,MB at $N=2^8$ to $0.80$\,MB at $N=2^{16}$ (log-normal, QROM). Shipping all signatures costs $208$\,MB at $N=2^{16}$, $260\times$ more, a ratio growing as $\Theta(N/\log N)$.
\emph{(2)~Cost of quantum security.} The QROM rule raises $\kr$ by $35\%$ ($143$ to $193$), but coins collide on heavy attestors, so distinct reveals grow less: the size overhead is $19$--$34\%$ across all configurations ($23$--$34\%$ for log-normal stake), and verification time grows by $20$--$58\%$ for log-normal stake.
\emph{(3)~Stake skew helps.} Under Zipf stake, coins concentrate on large validators and the certificate shrinks to $0.57$\,MB at $N=2^{16}$.
At $N=2^{16}$, verification takes $18$\,ms on one core and building $0.47$\,s, dominated by hashing the signature tree. As \cref{tab:kr} predicts, the proven threshold dominates cost: at $N=4096$, raising $\Wp$ from $W/3$ to $W/2$ and $2W/3$ grows the certificate from $0.69$ to $1.06$ and $1.95$\,MB. \Cref{tab:h384} reports $384$-bit commitments (\cref{sec:cert}); its totals also reflect slightly different reveal sets and a separate timing run, and with the reveal set fixed the wider hashes alone add $3.1$--$5.7\%$.

\begin{figure}[t]
\centering
\includegraphics[width=0.85\columnwidth]{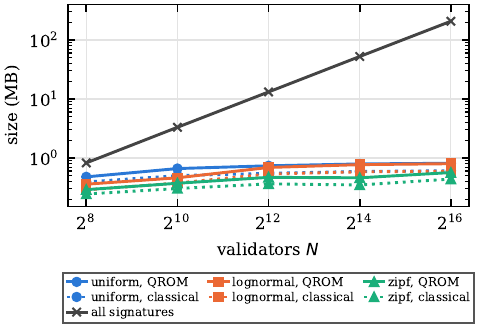}
\caption{Certificate size vs.\ number of validators (round-3 Dilithium2, $\Wp=\lfloor0.3333\,W\rfloor$, $85\%$ signed). Colors: stake distribution; solid: QROM rule; dotted: classical rule; black: all signatures.}
\label{fig:size}
\end{figure}

\begin{table}[t]
\centering
\caption{Certificates on M1 (round-3 Dilithium2, log-normal stake, QROM rule). ``Distinct'': distinct reveals; ``Gain'': size of all signatures divided by certificate size.}
\label{tab:cert}
\scriptsize
\setlength{\tabcolsep}{3.5pt}
\begin{tabular}{@{}rcccccc@{}}
\toprule
$N$ & $\kr$ & Distinct & Size (MB) & Gain & Build (ms) & Verify (ms)\\
\midrule
$2^{8}$  & 192 & 94  & 0.36 & $2.3\times$ & 2.4 & 8.7\\
$2^{10}$ & 192 & 118 & 0.46 & $7.2\times$ & 7.9 & 11.3\\
$2^{12}$ & 193 & 173 & 0.69 & $19\times$ & 29.4 & 16.3\\
$2^{14}$ & 193 & 188 & 0.77 & $68\times$ & 129 & 21.6\\
$2^{16}$ & 193 & 189 & 0.80 & $260\times$ & 474 & 18.4\\
\midrule
\multicolumn{7}{@{}l}{\emph{Classical rule, $2^{16}$:} $\kr=143$, 140 distinct, 0.60\,MB, verify 13.7\,ms}\\
\bottomrule
\end{tabular}
\end{table}

\begin{table}[t]
\centering
\caption{Keccak-384 vs.\ Keccak-256 commitments (round-3 Dilithium2, log-normal, QROM, M1).}
\label{tab:h384}
\scriptsize
\setlength{\tabcolsep}{3.5pt}
\begin{tabular}{@{}rccccc@{}}
\toprule
 & \multicolumn{3}{c}{Certificate size (MB)} & \multicolumn{2}{c}{Verify (ms)}\\
\cmidrule(lr){2-4}\cmidrule(l){5-6}
$N$ & $h=256$ & $h=384$ & Increase & $h=256$ & $h=384$\\
\midrule
$2^{12}$ & 0.692 & 0.734 & $+6.1\%$ & 19.0 & 20.4\\
$2^{14}$ & 0.772 & 0.801 & $+3.8\%$ & 20.6 & 21.7\\
$2^{16}$ & 0.801 & 0.857 & $+7.0\%$ & 21.2 & 20.7\\
\bottomrule
\end{tabular}
\end{table}

\noindent\textbf{FIPS 204 cross-check and light clients.} On a third machine, \emph{M3} (2 vCPUs, Intel Xeon, 2.1\,GHz), we ran the harness with round-3 Dilithium2 and with FIPS~204 ML-DSA-44 for $N=2^8$--$2^{16}$ (log-normal, $\Wp=\lfloor W/3\rfloor$). $\kr$ is identical; sizes differ by $-9.5\%$ to $+2.3\%$ (the latter for the classical rule at $N=2^{16}$) with no consistent sign, because different signatures change which positions the coins select; FIPS~204 verification is $2$--$33\%$ slower. At $N=2^{16}$ (QROM), the certificates are $0.80$ and $0.81$\,MB and verify in $18.4$ and $21.1$\,ms. The Python verifier (FIPS~204) checks a QROM certificate at $N=4096$ in $25$, $44$ and $57$\,ms on M1 for ML-DSA-44/65/87, and in $18$, $27$ and $44$\,ms on a laptop VM (\emph{M2}). L1, not light clients, is the bottleneck (\cref{sec:gas}).

\subsection{Baselines: Falcon and Hash-Based Signatures}\label{sec:baselines}
The signature scheme sets most of the cost. \Cref{tab:base} compares the Dilithium levels with Falcon-512 and with the synchronized hash-based signatures (generalized XMSS with target-sum encoding) that Drake \emph{et al.}~\cite{drake2025hashsig} propose for post-quantum Ethereum, measured with their implementation~\cite{hashsig_impl} at lifetime $2^{32}$. The hash-based rows combine their measured size and verification time with the reveal structure of the Dilithium2 (r3) certificate ($173$ distinct reveals, $1{,}328$ multiproof hashes), as coin selection is scheme-independent; the other rows are complete \texttt{pqcert} certificates.

\begin{table}[t]
\centering
\caption{Signature-scheme baselines on M1 ($N=4096$, log-normal, QROM, $3{,}524$ signers). ``All sigs'': size of shipping every signature and key. Rows marked (r3) use round-3 Dilithium (\cref{sec:eval}). $^\ast$Composed from measured components (see text).}
\label{tab:base}
\scriptsize
\setlength{\tabcolsep}{2.4pt}
\resizebox{\columnwidth}{!}{%
\begin{tabular}{@{}lccccc@{}}
\toprule
Scheme & pk / sig (B) & sign / verify (ms) & Cert (MB) & Cert verify (ms) & All sigs (MB)\\
\midrule
Dilithium2 (r3) & 1312 / 2420 & 0.42 / 0.076 & 0.69 & 16.3 & 13.2\\
Dilithium3 (r3) & 1952 / 3293 & 0.75 / 0.12 & 0.94 & 26.4 & 18.5\\
Dilithium5 (r3) & 2592 / 4595 & 0.74 / 0.22 & 1.27 & 41.4 & 25.3\\
Falcon-512 & 897 / 666 & 0.20 / 0.036 & 0.31 & 8.7 & 5.5\\
XMSS-TS, hashing-opt.~\cite{drake2025hashsig} & 52 / 3116 & 0.77 / 0.22 & 0.59$^\ast$ & 40$^\ast$ & 11.2\\
XMSS-TS, size-opt.~\cite{drake2025hashsig} & 52 / 2092 & 1.55 / 0.38 & 0.42$^\ast$ & 69$^\ast$ & 7.6\\
\bottomrule
\end{tabular}}
\end{table}

Compact certificates beat shipping all signatures by $18$--$20\times$ at $N=4096$. Falcon-512 is the best fit: its certificate is $2.2\times$ smaller than with Dilithium2 (r3) ($0.40$\,MB at $N=2^{16}$, verified in $11$\,ms), at the price of floating-point signing. The hash-based signatures of~\cite{drake2025hashsig} give smaller certificates than ML-DSA thanks to $52$-byte keys, but verify $2.5$--$4.2\times$ slower here and are \emph{synchronized}: a key signs at most once per epoch. Their hash-based succinct aggregation, which we did not measure, gives constant-size proofs but needs a SNARK prover. Compact certificates lose for very small sets (below about $2^6$ signers, extrapolating \cref{tab:cert}) and at high thresholds such as $\Wp=2W/3$ (\cref{tab:kr}). On L1 the verifier contract dominates every per-reveal cost: Falcon verification is reported at $1.5$--$3.9$M gas~\cite{zknox_falcon}, versus $1.3$--$8.2$M for an ML-DSA-44 dispute response here (\cref{tab:gas}), which affects only the dispute path.

\subsection{Validating the Grinding Term}
The attack of \cref{thm:attack} relies on each fresh root independently covering all coins with probability $\beta^{\kr}\ge(w(C)/\Ws)^{\kr}$, where $w(C)$ is the validly covered weight. With a SHA3-256 coin function of the same form, an adversary covering $w(C)/\Ws\in\{1/2,1/3\}$ re-randomizes its root until all $\kr$ coins are covered. Over $120$--$300$ runs per point, the mean number of trials matches $(\Ws/w(C))^{\kr}$ within two standard errors for $\kr=2$--$12$ and $\kr=2$--$8$, respectively (e.g., $4{,}252$ trials at $\kr=12$ against $4{,}096$ predicted). \Cref{fig:kr}(b) applies the attack to classically parameterized certificates: it caps their security at about $81$ bits for $q=2^{64}$ and $65$ bits for $q=2^{80}$, against the intended $128$, whereas under the QROM rule \cref{thm:qrom} guarantees $\lambda$ bits for the grinding term.

\begin{figure*}[t]
\centering
\includegraphics[width=\textwidth]{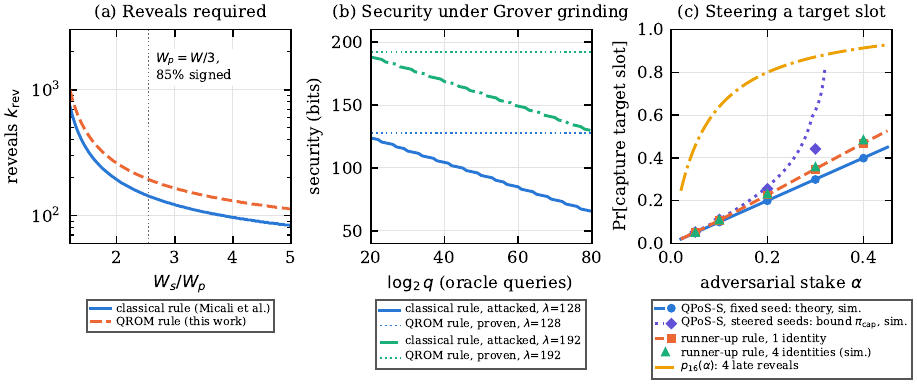}
\caption{(a)~Reveals required by the classical and QROM rules. (b)~Security of \emph{classically} parameterized certificates under the attack of \cref{thm:attack} ($\Ws/\Wp=2.55$, $N=2^{16}$); dotted: level $\lambda$ guaranteed by \cref{thm:qrom} under the QROM rule. (c)~Target-slot capture under QPoS-S, with a fixed seed and with seeds steered across epochs (dotted: bound $\pi_{\mathsf{cap}}$ of \cref{prop:late} for $B=0$; diamonds: the simulated adversary of \cref{sec:eval-leader}, which maximizes $c$ rather than capture), and under the runner-up rule (lines: theory; other markers: $4\times10^{5}$ simulated trials each).}
\label{fig:kr}
\end{figure*}

\subsection{Networked Certification}
An aggregator collects real FIPS~204 ML-DSA-44 signatures over TCP from validators in two worker processes, builds the certificate and sends it to $8$ light clients, under emulated WAN conditions (five regions at $1$--$95$\,ms one-way latency, $10\%$ of validators offline, collection until $80\%$ of the stake; full setup in \supp{S6}). Over five rounds, the median end-to-end time, from header broadcast until all light clients have verified, is $321$, $288$, $408$ and $647$\,ms for $N=128$, $512$, $1024$ and $2048$, and no round exceeded $716$\,ms. Beyond the round trip to the farthest needed region, collection time grows with sequential signature verification at the aggregator ($0.13$\,ms per signature, $221$\,ms for the ${\approx}1{,}720$ signatures at $N=2048$), which parallelizes trivially, and construction takes at most $22$\,ms. Certification is far from limiting withdrawal latency, which $T_{\mathsf{ch}}$ dominates.

\subsection{Cost on Ethereum: Measured Contracts}\label{sec:gas}
We implemented the vault and dispute state machine of \cref{sec:design,sec:opf} in Solidity (\texttt{PQVault.sol}, about $300$ lines) on the ZKNOX ML-DSA verifiers~\cite{zknox_dilithium}, and measured transaction-level gas (EIP-7623~\cite{eip7623}) with Foundry for the NIST-conformant and the packed ZKNOX verifier (\cref{tab:gas}). The test vectors are a real certificate over $N=2^{16}$ positions with $872$ FIPS~204 signatures, which both verifiers accept; $17$ tests per verifier confirm the state machine (\supp{S6}).

\begin{table}[t]
\centering
\caption{Measured L1 gas of \texttt{PQVault} (Foundry; $\kr=193$, $130$ distinct reveals, $d=16$). Storage written earlier in a test is warm. All operations fit in one transaction except direct verification, which exceeds the $2^{24}$-gas cap~\cite{eip7825} (parentheses: minimum number of transactions), and the open-all posting, split for the $128$\,KB mempool limit. Details in \supp{S6}.}
\label{tab:gas}
\scriptsize
\setlength{\tabcolsep}{3pt}
\resizebox{\columnwidth}{!}{%
\begin{tabular}{@{}lrr@{}}
\toprule
Operation & \makecell[r]{Gas, NIST\\verifier} & \makecell[r]{Gas, packed\\verifier}\\
\midrule
Register key (once, $22.4$\,KB ZKNOX format) & $6.8\times10^{6}$ & $6.8\times10^{6}$\\
Deposit & $3.5\times10^{4}$ & $3.5\times10^{4}$\\
Propose checkpoint (computes $M_e$, checks $\kr$) & $3.3\times10^{5}$ & $3.3\times10^{5}$\\
Finalize checkpoint & $9.0\times10^{4}$ & $8.9\times10^{4}$\\
Withdraw (depth-20 proof) & $1.0\times10^{5}$ & $1.0\times10^{5}$\\
\midrule
Coin challenge & $5.3\times10^{4}$ & $5.4\times10^{4}$\\
Respond (one ML-DSA verification) & $8.2\times10^{6}$ & $1.3\times10^{6}$\\
Open-all request & $2.4\times10^{4}$ & $2.4\times10^{4}$\\
Post all $193$ records ($7$ txs, $0.80$\,MB) & $3.0\times10^{7}$ & $3.0\times10^{7}$\\
\midrule
Direct verification, all $193$ coins & \makecell[r]{$1.1\times10^{9}$\\(${\ge}\,65$ txs)} & \makecell[r]{$1.8\times10^{8}$\\(${\ge}\,11$ txs)}\\
\bottomrule
\end{tabular}}
\end{table}

The optimistic path (propose plus finalize) costs about $4.2\times10^{5}$ gas per checkpoint ($4.1\times10^{5}$ in the happy-path test). Answering a coin challenge costs $8.2\times10^{6}$ (NIST) or $1.3\times10^{6}$ gas (packed; $1.4\times10^{6}$ with the challenge), paid by the party in the wrong. Direct verification costs $2{,}600\times$ and $430\times$ the optimistic path, and its $0.81$\,MB of calldata exceeds the common $128$\,KB mempool limit; extrapolated to the $189$ distinct reveals at $N=2^{16}$ (\cref{tab:cert}), it would cost $1.6\times10^{9}$ and $2.6\times10^{8}$ gas (\supp{S6}). Even the draft EIP-8051 precompile~\cite{eip8051} would leave direct verification bound by calldata ($20{,}512$-byte expanded keys). An open-all posting ($803{,}580$ bytes in $7$ transactions) costs $3.0\times10^{7}$ gas, set almost entirely by the calldata floor; execution adds only $1.6$--$2.0\times10^{6}$.

\subsection{Leader Election}\label{sec:eval-leader}
\noindent\textbf{Fairness and Sybil invariance.} We evaluated the score of \cref{eq:score} with SHA3-256 over a per-validator value (a hash of its index), the slot and a fixed seed, for $200$ validators with log-normal stake and $40{,}000$ slots. Win frequencies match stake fractions ($\chi^2=190.8$ with $199$ degrees of freedom, $p=0.65$). Sampling the race of \cref{lem:race}, an adversary with $\alpha=0.2$ split into $m\in\{1,2,5,10,50\}$ identities over $2\times10^5$ slots keeps a share between $0.198$ and $0.203$ (standard error $0.0009$), as \cref{lem:race} predicts.

\noindent\textbf{Steering.} The remaining experiments sample the exponential-race model of \cref{lem:race} directly; they check our derivations rather than an implementation. \Cref{fig:kr}(c) compares target-slot capture. Under the runner-up rule, the simulation (adversary sees both candidate seeds) matches case~(a) of \cref{cor:qrl}, $0.347$ vs.\ $0.348$ at $\alpha=0.3$, and $4$ identities raise it to $0.357$, so Sybil invariance fails. With a fixed QPoS-S seed the probability equals $\alpha$ within sampling error; with seeds steered across epochs it stays below $\pi_{\mathsf{cap}}$ (\cref{prop:late}), within sampling error at $\alpha=0.05$, where the bound is nearly tight. Choosing among $2^c$ epoch seeds ($k=2000$) gains $0.004$--$0.038$ in epoch share for $c\le 12$ (an illustrative range; \cref{prop:late} makes $c\ge4$ rare), always below the bound $\sqrt{c\ln 2/(2k)}$ of \cref{thm:schedule} ($0.013$--$0.046$).

\noindent\textbf{Reveal order and fallback proposers.} A further experiment (\texttt{reveal\_order\_sim.py}) checks \cref{prop:late} and \cref{thm:schedule}, with $300$ honest validators of log-normal stake in its first two parts. More adversarial identities raise $\Pr[c\ge x]$ towards, but not above, $\alpha^x$: at $\alpha=0.3$, $100$ identities give $0.299$, $0.088$, $0.026$ and $0.0077$ for $x=1,\dots,4$ ($10^5$ orders each), against $0.3$, $0.09$, $0.027$ and $0.0081$; a single identity never has $c\ge2$. Over $10^5$ consecutive epochs in which the adversary ($20$ identities) always picks the candidate maximizing the next $c$, $\mathbb E[2^c]$ is $1.14$, $1.42$ and $2.19$ for $\alpha=0.1$, $0.2$ and $0.3$ (bound $\mu$: $1.14$, $1.5$ and $4$), and $c$ never reached $L_t=25$; its target-slot capture rate, $0.112$, $0.255$ and $0.442$, is an example rather than a worst case (bound $\pi_{\mathsf{cap}}$: $0.112$, $0.265$ and $0.622$). The third part ($30$ adversarial and $70$ equal-stake honest identities, $2\times10^5$ slots) finds slots not guaranteed with frequency $0.300$ and $0.511$ at $\alpha=0.3$ for $B=0$ and $B=1$, against $p_B=0.300$ and $0.514$.

\section{Discussion and Conclusion}\label{sec:discussion}

\noindent\textbf{Scope of \cref{asm:l1}.} A CRQC that breaks ECDSA also breaks Ethereum's BLS consensus signatures, so \cref{asm:l1} is not implied by our threat model: it assumes that L1 consensus is not attacked, or is restored by a state-preserving hard fork~\cite{buterin2024emergency}. \Cref{asm:slow} is needed only for transactions that expose an ECDSA key (deposits, initial bindings, bonds of relayers and watchers); PQ smart accounts would remove it (\supp{S7}).

\noindent\textbf{Data availability and watchers.} \Cref{thm:opf} assumes one honest online watcher; withholding blob data cannot cause a false acceptance (\cref{sec:opf}). For a gas-price ceiling $p^{\max}$, $b_C\ge 8.2\times10^{6}p^{\max}$ covers a response and $b_R\ge b_C$ plus a reward. Open-all requests carry no bond, so a watcher spends at most one request fee and one bond $b_C$ per proposal, while relayers must recover the $3.0\times10^{7}$-gas posting through fees, on which G4 relies. Any depositor can watch at about $20$\,ms per checkpoint (\supp{S7}).

\noindent\textbf{Public leader schedule.} The public schedule enables \cref{thm:fresh} but exposes leaders and lets late revealers see candidate schedules. With $L=3$ fallback proposers, an adversary suppressing $B$ honest proposers per slot affects only slots that are not guaranteed, so inclusion within $31$ to $56$ slots holds with overwhelming probability at $\alpha=0.3$ (\cref{thm:schedule,cor:real}). The stake-weighted reveal order makes late reveals rare ($\Pr[c\ge x]\le\mu\alpha^x$, \cref{prop:late}); target-slot capture is at most $0.622$ at $\alpha=0.3$ (runner-up attack: $0.348$; owning the last four reveals: $0.997$) and at most $0.265$ for $\alpha\le0.2$. Instant adaptive corruption, which could seize the last honest revealer, is outside our model, and hidden schedules via a PQ VRF or secret leader election are left open (\supp{S7}).

\noindent\textbf{Limitations.} The sidechain is analyzed by simulation only, the network experiment emulates WAN latency on shared cores, and gas was measured in a local Foundry EVM. If the corrupt third abstains and $85\%$ of honest stake signs, $\Ws/\Wp=1.7$ and $\kr=340$, which enlarges certificates but leaves the optimistic L1 path unchanged. Our leader-election bounds assume a standard BFT finality gadget, which we did not re-analyze, that finalizes each block of the reveal window within its slot. The prototype does not implement the binding registry or the canary. Any EUF-CMA scheme fits; Falcon-512 halves certificate size but is hard to sign in constant time (\supp{S7}).

\noindent\textbf{Conclusion.} Protecting Ethereum assets against quantum adversaries without changing L1 is possible only for escrowed assets, and only if every predicate on the release path is PQ. PQ-Vault realizes this with quantum-secure compact certificates, steering-resistant stake-weighted PoS and single-reveal fraud proofs. These guarantees cost $19$--$34\%$ more certificate size, certification in $0.65$\,s at $2048$ validators, and about $4.2\times10^5$ gas per checkpoint on L1.

\setcounter{section}{0}
\renewcommand{\thesection}{S\arabic{section}}
\noindent\textbf{Appendices.} The appendices contain the full proofs of all results proved in the paper, the complete specification of the dispute protocol, and additional details and discussion. Appendix~S$x$ supplements Section~$x+3$; each result is restated before its proof.
\section{System and Threat Model: Proofs}\label{S-model}

\subsection{On \cref{asm:l1}}
\Cref{asm:l1} is unavoidable: if Ethereum's own consensus is subverted, no contract on Ethereum can protect anything. Our goal is to protect assets against the failure of \emph{account-level} cryptography (\cref{sec:discussion}).

\subsection{Proofs of \cref{prop:impossible} and \cref{cor:vault}}

\begin{restate}{Proposition~\ref{prop:impossible} (Custody impossibility)}
Let $\Pi$ be any protocol that does not modify L1 validity rules. Suppose that at some time $t\ge\tau_Q$ an asset $a$ is still spendable under the L1 predicate $\mathsf{ECDSA.Vf}(pk_M,\cdot)=1$ with $pk_M$ public. Then a QPT adversary can produce an L1-valid transaction spending $a$ that L1 cannot distinguish from one signed by the owner. Nothing produced by $\Pi$ can make this transaction invalid.
\end{restate}
\begin{proof}
$\mathcal{A}$ runs Shor's algorithm on $pk_M$ to obtain $sk_M$ and signs a transaction that transfers $a$. L1 validity depends only on the predicate and on the account state, and $sk_M$ produces exactly the signatures the owner would produce. Since $\Pi$ does not change L1 validity rules, no message, proof, or state of $\Pi$ can invalidate the transaction. Whether the adversary or the owner moves $a$ first is then decided by transaction ordering alone, which the protocol does not control.
\end{proof}

\begin{restate}{Corollary~\ref{cor:vault}}
Consider, on unmodified L1, an asset that must remain \emph{live}, i.e., movable by its owner, and \emph{safe} against $\mathcal{A}$ after $\tau_Q$. Every predicate that can release it after $\tau_Q$ must then rest on PQ assumptions and on \cref{asm:l1} alone; in particular, no ECDSA key whose public key is exposed may authorize its release, or change or redirect such a predicate without a PQ-authorized veto. On an account-based L1 this requires holding the asset under a contract predicate $\Phi$.
\end{restate}
\begin{proof}
An ECDSA key that can release the asset, directly or by changing the predicate without a veto, falls under \cref{prop:impossible} once its public key is exposed. An EOA is controlled by its ECDSA key, so only contract code can implement a different predicate.
\end{proof}

\noindent \Cref{prop:impossible} rules out designs in which validators ``check a PQ proof before including'' an L1 transaction. The sidechain design of~\cite{myeprint} implicitly assumed such a check.

\noindent The corollary does not cover addresses whose public key has never been exposed. Attacking their address commitment is a different problem, and their first spend exposes the key (\cref{asm:slow}). Nor does it count an asset made permanently unspendable as protected: that would be safe but not live.
\section{PQ-Vault Design: Details and Proof}\label{S-design}

\subsection{Why the Earlier Zero-Knowledge Binding Is Replaced}
The sidechain design of~\cite{myeprint} bound $pk_M$ to $pk_S$ with a zero-knowledge proof that $pk_M=\mathsf{KGen}(F(\mathit{Seed}))$. That proof requires secp256k1 scalar multiplication inside a transparent circuit, which is costly. It also does not help after $\tau_Q$: the adversary learns $sk_M$, and the L1 predicate that matters still accepts $sk_M$ (\cref{prop:impossible}). We replace it with a binding that is authenticated \emph{classically, while classical authentication is still meaningful}, and that becomes immutable afterwards.

\subsection{Binding Mechanism in Full}\label{S-binding}
The binding has three parts. In the \emph{bind} step, $\mathit{Addr}_M$ calls $\mathsf{bind}(\Hash(pk_S))$; if no binding exists, it takes effect immediately, and the L1 transaction sender authenticates $\mathit{Addr}_M$. A later \emph{rebind} from $\mathit{Addr}_M$ creates a pending rebinding with a fresh per-address counter $n$ and activation time $\tau_{\mathsf{act}}=\tau_{\mathsf{req}}+D$. At most one rebinding per address may be pending, so the owner never has to cancel more than one at a time; $\mathsf{applyRebind}$ succeeds only in a block with timestamp strictly greater than $\tau_{\mathsf{act}}$, and only if counter $n$ is not cancelled. Until then, the current PQ identity can \emph{veto} it with an ML-DSA signature on $\textsf{cancel}(\Vault,\mathit{Addr}_M,n)$, submitted to the sidechain. The cancellation reaches $\Vault$ through a certified checkpoint, whose binding-operations root $\beta_e$ is part of $M_e$; no ML-DSA verification on L1 is needed. A finalized root does not by itself change the registry: anyone, normally the owner, then calls $\mathsf{applyCancel}(\mathit{Addr}_M,n,e,\pi)$, which checks that $\pi$ opens the operation $(\textsf{cancel},\mathit{Addr}_M,n)$ under the finalized $\beta_e$ and marks counter $n$ cancelled. A cancellation is effective if this call is included at a timestamp at most $\tau_{\mathsf{act}}$; it consumes $n$, and an operation for any other counter has no effect.

\noindent\textbf{Quantum canary.} The canary makes the registry immutable once quantum attacks are demonstrated. $\Vault$ holds a bounty locked to an address $\mathit{Addr}^\star$ derived from a nothing-up-my-sleeve secp256k1 point whose discrete logarithm no one knows. The canary fires on a valid ECDSA signature for $\mathit{Addr}^\star$, checked with \texttt{ecrecover}, on the \emph{fixed} digest $z^\star=\Hash(\textsc{canary}\,\|\,\mathit{chainid}\,\|\,\Vault)$, and then permanently disables all ECDSA-authorized registry operations. The digest must be fixed: for a digest chosen by the submitter, a valid signature exists without the key (pick $u,v\ne0$, let $R=uG+vQ$, $r=x(R)$, $s=rv^{-1}$ and $z=us$). A signature on $z^\star$ is a forgery of ECDSA on a predetermined message, which for a key with unknown discrete logarithm signals a cryptanalytic or quantum break.

\subsection{Further Details of the Vault and Sidechain}
\noindent\textbf{Vault.} The sidechain credits each deposit event once, after it is final on L1. Several proposals for the same epoch may be pending, and a rejected epoch can be proposed again; the dispute functions are specified in \cref{sec:opf}. When $\mathsf{finalize}$ succeeds, other proposals for epoch $e$ become stale. The sidechain assigns each withdrawal a globally unique $\mathit{nonce}$, so $\xi$ is a nullifier: an authorized withdrawal that appears in several checkpoints is paid at most once. Recipients that are not fresh addresses or PQ smart accounts would expose the withdrawn funds again (\cref{prop:impossible}).

\noindent\textbf{Sidechain.} $\Hash_n$ denotes $n$ output bits of cSHAKE256 with a PQ-Vault customization string, so that its inputs are separated from Keccak-256 and from ML-DSA's internal SHAKE256. Honest validators send $r_i$ to no one before their phase or sub-window starts. A validator whose reveal for epoch $e$ is missing is suspended from epoch $e+3$ on, the first epoch whose registrations are still open, until it registers again; it is not slashed, because a censored honest reveal looks the same. Every node computes the schedule of the whole epoch once the seed is final. Proposer eligibility needs no proof beyond the block signature. The $1536$-bit commitments make each commitment open to at most one secret, statistically in the QROM, and the $512$-bit secrets and seed keep the quantum hiding errors small (\cref{thm:fresh}). They add $192$ bytes to a registration that already carries a $2{,}420$-byte signature.

\subsection{Proof of \cref{prop:binding}}

\begin{restate}{Proposition~\ref{prop:binding} (Binding integrity)}
Suppose $\mathsf{Bind}[\mathit{Addr}_M]=\Hash(pk_S)$ was set by $\mathit{Addr}_M$'s owner, and let $\Delta_{\mathsf{v}}=\Delta_{\mathsf{obs}}+\Delta_{\mathsf{sc}}+\Delta_{\mathsf{cp}}+\Delta_{\mathsf{apply}}$, where $\Delta_{\mathsf{obs}}$ bounds the owner's time to observe a pending rebinding and sign, $\Delta_{\mathsf{sc}}$ bounds sidechain inclusion of its transaction, $\Delta_{\mathsf{cp}}$ bounds the time from inclusion to finalization of the carrying checkpoint, including disputes and the $\mathsf{finalize}$ call (\cref{thm:opf}), and $\Delta_{\mathsf{apply}}$ bounds preparing and including $\mathsf{applyCancel}$. Suppose $D>\Delta_{\mathsf{v}}$, the owner is online at least once in every window of length $D-\Delta_{\mathsf{v}}$, these bounds hold, and \cref{asm:l1,asm:pq,asm:slow,asm:watcher} and the timing conditions of \cref{thm:opf} hold. By \cref{cor:real} and \cref{thm:opf}, $\Delta_{\mathsf{sc}}$ and $\Delta_{\mathsf{cp}}$ fail only with the probabilities given there. Then a QPT adversary can change the binding only by forging an ML-DSA signature, by breaking certificate soundness, or before the owner's first binding transaction is included.
\end{restate}
\begin{proof}
Before the owner's first binding transaction is included, no binding exists, and this case is excluded by the statement. Afterwards, any change to $\mathsf{Bind}[\mathit{Addr}_M]$ goes through a pending rebinding that lasts $D$. Because the owner is online in every window of length $D-\Delta_{\mathsf{v}}$, it is online at some time at which at least $\Delta_{\mathsf{v}}$ of the delay remains. It then signs a cancellation of counter $n$ with $sk_S$, which is included on the sidechain, finalized on L1, and applied by $\mathsf{applyCancel}$ no later than $\tau_{\mathsf{act}}$. Since $\mathsf{applyRebind}$ needs a timestamp strictly greater than $\tau_{\mathsf{act}}$, the rebinding can never be applied. An adversary can prevent this outcome in only two ways. It can produce an ML-DSA-signed operation under $pk_S$ that the owner never signed; operations signed for earlier counters have no effect on $n$, so this is a signature forgery. Or it can have the vault finalize a checkpoint that omits or reverses binding operations that are final on the honest sidechain. Such a checkpoint is not signed by any honest validator, and its finalization is bounded by \cref{thm:opf} (with \cref{thm:qrom} on the direct path). After the canary fires, the vault rejects every ECDSA-authorized registry operation, and the binding can no longer change at all.
\end{proof}

\subsection{Legacy EOAs and Recovery Hard Forks}
The binding gives users with funds still in legacy EOAs a pre-registered PQ identity, which is useful if Ethereum later adopts a recovery hard fork~\cite{buterin2024emergency}. We emphasize that such funds are not protected by PQ-Vault itself (\cref{prop:impossible}). The registry is not on the release path of escrowed funds, which deposits credit to $\beta$ directly, and the canary is defense in depth: no theorem of the main paper relies on either.

\subsection{Transparent Proof Layer}
If the transparent proof layer is enabled, a block carries a proof of correct state transition, which lets light validators check blocks without re-execution. Certificate soundness does not depend on that layer. The leader has no exclusive power over certification, since any party can aggregate signatures into a certificate (\cref{alg:cert}).

\subsection{Trust in the Attestor Set: Further Discussion}
With $\Wp=\lfloor W/3\rfloor$, a coalition holding more than a third of the attestor stake could certify an arbitrary $M_e$ without any honest signer; $\Wp$ trades certificate cost against the fraction of stake that must be corrupted. Comparing the cost of corrupting a third of the stake with the escrowed value is only an economic motivation for the assumption that corrupt stake stays below one third of every accepted attestor set. It would need further assumptions on utilities, bribery and enforceability, and an adversary may value disruption above the escrow. Signers of a fraudulent certificate are identifiable, but slashing them would require stake held on L1, which our prototype does not implement.
\section{Compact Certificates: Proofs and Details}\label{S-cert}

\subsection{Proof of Theorem~\ref{thm:qrom}}\label{S-cert-qrom}
The additive term accounts for the 128-bit modular reduction in \cref{eq:coin}: for any set $B\subseteq[\Ws]$, at most $|B|\lceil2^{128}/\Ws\rceil$ of the $2^{128}$ values reduce into $B$. The verifier takes $W$ and $N$ from its own state and accepts only $(1+\delta)\Wp\le\Ws\le W$ (\cref{alg:cert}), which bounds $\kr$ and makes the additive term negligible. The game fixes $W\le W_{\max}=2^{64}-1$ (stakes are 64-bit integers) and $\Wp$, which the verifier takes from its state.

\begin{restate}{Theorem~\ref{thm:qrom} (QROM soundness of compact certificates)}
Let $d=\lceil\log_2N\rceil$, let $q_V=\kr(2d+3)$ with $\kr=\kr((1+\delta)\Wp)$, let $t=q+q_V+2^{d+1}$, and let $h$ be the smallest output width of the oracles and trees involved. For every adversary making $q$ quantum queries to $\RO$ and any number of signing queries,
\[
\Adv^{\mathsf{CERT\text{-}FORGE}}_{\mathcal A}\le 12t^2\varepsilon^\ast+\frac{48t^3+2q_V}{2^{h}}+2\Adv^{\mathsf{uf}}_{K}.
\]
\end{restate}

\noindent The certificate is the BCS transformation~\cite{bcs2016} of a one-round public-coin IOP whose prover commits to the leaf vector and whose verifier draws the $\kr$ coins, and the proof follows the analysis of~\cite[Sec.~2.9]{chiesa2019qrom}. Two features need care. The verifier's predicate involves signatures, which $\mathcal A$ obtains adaptively; and \cref{eq:coin} derives each coin from a separate query.

\begin{proof}
\emph{Step 1 (one oracle value per commitment).} Let $k_{\max}=\kr((1+\delta)\Wp)$. Replace the oracle on coin inputs $(\textsc{dom},j,x)$ with $j<k_{\max}$, where $x=(\Rsig,\Ratt,\Ws,\Wp,N,M)$, by $\RO_c$ with $hk_{\max}$-bit outputs: block $j$ of $\RO_c(x)$ answers $\RO(\textsc{dom}\|j\|x)$, and the coin is its first 128 bits reduced modulo $\Ws$. Distinct blocks are independent and uniform, so the oracle's distribution is unchanged. The domain-separated pair $(\RO,\RO_c)$ is one oracle with inputs of two kinds, and $\mathsf{Col}$ refers to its $h$-bit Merkle outputs. A quantum query, possibly in superposition over both kinds of inputs and over $j$, is simulated by one query to the combined oracle: the output register is swapped into block $j$, the other blocks are set to $\lvert+\rangle$, which XOR leaves unchanged, and the swap is undone. Query counts are therefore preserved. The verifier's queries are bounded by $q_V$, since $|P|\le\kr$ (\cref{alg:cert}) and it computes at most $\kr$ coins and at most $2|P|(d+1)$ leaf and node hashes. Building $\Ratt$ costs at most $2^{d+1}$ queries, and ML-DSA's internal hashing is independent of $\RO$ (\cref{sec:prelim}). Hence $t$ bounds all queries to $\RO$ and $\RO_c$ in the experiment.

\emph{Step 2 (a fixed database property).} Let $D$ be a database without collisions. Consider each coin entry of $D$ with input $x=(R,\Ratt,\Ws,\Wp,N,M)$ whose $\Wp$ field is the game's $\Wp$ and whose $\Ws$ lies in $[(1+\delta)\Wp,W]$; the verifier reads no other coin entries, and read the attestor leaves $(pk_i,w_i)$ below $\Ratt$ and the leaves $(\sigma_i,L_i)$ below $R$ from $D$, with $\bot$ for leaves not in $D$. Let $V$ be the set of positions $i<N$ at which both leaves are present, $\mldsa.\mathsf{Vf}(pk_i,M,\sigma_i)=1$, $w_i>0$ and $L_i+w_i\le\Ws$, and let $B=\bigcup_{i\in V}[L_i,L_i+w_i)$. Call the entry \emph{light} if $w(V)\le\Wp$ and \emph{hit} if all $\kr(\Ws)$ coins lie in $B$. Let $P$ be the property that $D$ contains a light, hit coin entry. $P$ is a fixed set of databases: it depends neither on the signing history nor on $C$, because signature validity is a fixed predicate once ML-DSA's independent oracle is fixed, and we condition on that oracle.

We bound $I(P\,|\,\overline{\mathsf{Col}},t)$. Adding a coin entry for input $x$ leaves every tree in $D$ unchanged and creates one new entry. If that entry is light, then $|B|\le w(V)\le\Wp$, and each coin lies in $B$ with probability at most $|B|\lceil2^{128}/\Ws\rceil/2^{128}\le\Wp/\Ws+\Ws/2^{128}$; the $\kr$ blocks are independent, so the entry is hit with probability at most $\varepsilon(\Ws)\le\varepsilon^\ast$. Adding a hash entry changes a tree below an existing root, and thereby $V$ or $B$ for some coin entry, only if its output equals one of the at most $2t$ child or root values occurring in inputs of $D$, which has probability less than $2t/2^h$. Hence $I(P\,|\,\overline{\mathsf{Col}},t)\le\varepsilon^\ast+2t/2^h$. With $I(\mathsf{Col},t)<t/2^h$, \cref{lem:extract}(i) gives $\Pr[D\in P]\le 6t^2\varepsilon^\ast+18t^3/2^h$ and $\Pr[D\in\mathsf{Col}]\le 6t^3/2^h$.

\emph{Step 3 (the database game).} Suppose that the verifier accepts $(M^\ast,\cert^\ast)$ when every oracle value it reads is taken from $D$, that $D\notin\mathsf{Col}$, and that $w(C\cup S_{M^\ast})\le\Wp$. The coin entry for $\cert^\ast$ is in $D$, and because $D$ has no collisions the revealed leaves equal the leaves that $D$ defines below $\Rsig$ and $\Ratt$. Every revealed position verifies and therefore lies in $V$, and every coin is covered, so the entry is hit. If it is light, $D\in P$. Otherwise $w(V)>\Wp\ge w(C\cup S_{M^\ast})$, so some $i\in V$ is an uncorrupted honest position that never signed $M^\ast$, while $D$ holds a valid signature $\sigma_i$ on $M^\ast$ under its key. A reduction runs $\mathcal A$ with an efficient compressed-oracle simulation, forwards signing and corruption queries to the multi-user EUF-CMA challenger for the $K$ honest keys, measures $D$ at the end, and outputs $(M^\ast,\sigma_i)$. Hence the database game is won with probability at most $\omega_D\le 6t^2\varepsilon^\ast+24t^3/2^h+\Adv^{\mathsf{uf}}_K$.

\emph{Step 4 (back to the real oracle).} The final check reads at most $q_V$ oracle values, each of at least $h$ bits. By \cref{lem:extract}(ii), $\Adv\le(\sqrt{\omega_D}+\sqrt{q_V/2^h})^2\le 2\omega_D+2q_V/2^h$, which is the claimed bound.
\end{proof}

\subsection{Remarks on Theorem~\ref{thm:qrom}}\label{S-cert-remarks}
The proof needs no entropy property of honest signatures and no statement about \emph{when} a leaf was committed: a valid leaf at an honest position is either charged to $S_{M^\ast}$ or is a forgery extracted after $\mathcal A$ halts. Registered keys, including copies of honest keys, lie in $C$. Quantifying over all coin entries covers the adaptive choice of $\Ws$, $M^\ast$ and commitment. With $\Wp\ge W/4$, $(1+\delta)\Wp\le\Ws\le W<2^{64}$ (stakes are 64-bit integers in the contract) and $\delta=0.1$, we have $\kr<2^{11}$ and $\Ws^2/(\Wp2^{128})<2^{-62}$, so the reduction term changes $\varepsilon(\Ws)$ by a factor of at most $1+2^{-51}$.

\subsection{Proof of Theorem~\ref{thm:attack}}\label{S-cert-attack}
\begin{restate}{Theorem~\ref{thm:attack} (Quantum grinding attack)}
Let $d=\lceil\log_2 N\rceil$ and $c_{\mathsf{ev}}=d+1+\kr$. Suppose the adversary corrupts $C$ with $a=w(C)\le\Wp$, receives no honest signatures, and some position outside $C$ or some padding position exists. Fix any $\Ws$ with $(1+\delta)\Wp\le\Ws\le W$, let $\beta=|\{y<2^{128}: y\bmod\Ws<a\}|/2^{128}$ and $\varepsilon_C=\beta^{\kr}$; then $\varepsilon_C\ge(a/\Ws)^{\kr}$. If $q\ge 2^{d+1}+c_{\mathsf{ev}}$, there is a quantum adversary making at most $q$ queries that wins $\mathsf{CERT\text{-}FORGE}$ with probability at least
\[
\frac{4}{\pi^2}(2T_G+1)^2\varepsilon_C\Big[1-(2T_G+1)^2\big(\varepsilon_C+(d+2)2^{-h}\big)\Big],
\]
where $T_G=\lfloor(q-2^{d+1}-c_{\mathsf{ev}})/(2c_{\mathsf{ev}})\rfloor$.
\end{restate}

\begin{proof}
Write $2^{128}=m\Ws+r$ with $0\le r<\Ws$. Exactly $am+\min(a,r)\ge a2^{128}/\Ws$ values reduce into $[0,a)$, which gives $\beta\ge a/\Ws$. The adversary signs the claimed $M$ with every key in $C$ and packs their intervals into $[0,a)$. The verifier checks only revealed intervals, so the claimed $\Ws$ need not equal the weight of the committed leaves. One position outside $C$ carries an $h$-bit salt $z$, and all other leaves are fixed. Building the tree once costs at most $2^{d+1}$ queries. For each $z$, recomputing $\Rsig(z)$ along one root path and evaluating the $\kr$ coins costs $c_{\mathsf{ev}}$ queries. Let $G(z)=1$ iff all coins are below $a$; a marked $z$ yields a certificate that reveals only positions in $C$, all of which verify.

Let $\nu$ be the fraction of marked salts. Fix the oracle values on Merkle inputs, which determine $\Rsig(\cdot)$. Coin inputs are domain-separated from Merkle inputs, so the indicators for distinct roots are independent with mean $\varepsilon_C$, and salts with equal roots share their indicator. Two distinct salts give equal roots with probability at most $(d+1)2^{-h}$, because the first equal hash on their paths is a collision. Hence $\mathbb E[\nu]=\varepsilon_C$ and $\mathbb E[\nu^2]\le\varepsilon_C^2+\varepsilon_C(d+2)2^{-h}$. Grover's algorithm with $T_G$ iterations uses $2T_Gc_{\mathsf{ev}}$ queries (computation and uncomputation) and $c_{\mathsf{ev}}$ more to recompute the certificate classically. For a fixed oracle, it outputs a marked salt with probability $\sin^2((2T_G+1)\theta)$, where $\sin^2\theta=\nu$. If $\nu\le\nu_{\max}=(2T_G+1)^{-2}$, then $(2T_G+1)\theta\le\pi/2$ (since $\arcsin x\le\pi x/2$), and $\sin y\ge 2y/\pi$ and $\theta\ge\sqrt\nu$ give success at least $(4/\pi^2)(2T_G+1)^2\nu$. Averaging over the oracle, the success probability is at least $(4/\pi^2)(2T_G+1)^2\big(\mathbb E[\nu]-\mathbb E[\nu\,\mathbf 1_{\nu>\nu_{\max}}]\big)$, and $\mathbb E[\nu\,\mathbf 1_{\nu>\nu_{\max}}]\le\mathbb E[\nu^2]/\nu_{\max}$ gives the bound.
\end{proof}

\noindent For the grinding term, the best attack's success probability in the example of \cref{cor:loss} lies between $2^{-81.1}$ and $2^{-61.5}$; only the factor $3\pi^2c_{\mathsf{ev}}^2$ remains open. The attack is a query-model statement.

\subsection{Deployed Parameters: Algorand State Proofs}\label{S-cert-algorand}
Algorand state proofs~\cite{algorand_stateproofs}, to our knowledge the only deployed compact certificates, use consensus parameters~\cite{goalgorand_consensus} with a strength target of $256$, documented as ``$k+q$ (for pre-quantum security) or $k+2q$ (for post-quantum security)'', and a proven weight of $30\%$ of the top $1024$ accounts. Coins come from one SHAKE256 stream with rejection sampling, without the reduction term of \cref{eq:eps}. Read as $\lambda+2\log_2 q$ with $\log_2 q=64$, this target matches our rule up to the constant $\kappa_Q$. If its hash functions are modeled as random oracles (an assumption for its SumHash-based commitments), the grinding term of \cref{thm:qrom} certifies about $256-128-\log_2 12\approx 124$ bits against quantum provers making $2^{64}$ queries. We found no vulnerability: our theorem justifies its heuristic $k+2q$ choice for the grinding term, with a gap of $3.6$ bits. Systems that follow the classical $k+q$ rule of~\cite{micali2020compact}, as the prototype of~\cite{NIROPoK2025} did, fall under \cref{cor:loss}.

\subsection{The Exact Parameter Rule}\label{S-cert-rule}
The rule~(\ref{eq:rule}) controls the \emph{grinding} term only. For an arbitrary total query budget $T$ and a per-epoch grinding budget $\eta_g$, the exact sufficient rule is
\[
\kr(\Ws)\ge\Big\lceil\frac{\log_2(12T^2/\eta_g)}{-\log_2(\Wp/\Ws+\Ws/2^{128})}\Big\rceil ,
\]
which \cref{eq:rule} implies for $\eta_g=2^{-\lambda}$ in the regime $T\le q(1+2^{-20})$ and for $\Wp\ge W/4$, as deployed, because $12(1+2^{-20})^2(1+2^{-51})<16$ (\cref{S-cert-remarks}). $T$ counts all queries to the hash oracle on the L1 path ($256$- or $384$-bit): coins, certificate trees, $C_{\cert}$, checkpoint headers, withdrawal trees and code hashes. Scores and order keys use domain-separated inputs, which we model as independent oracles.

\subsection{Output Width and the Commitment Term}\label{S-cert-budget}
The other terms of \cref{thm:qrom} need separate budgets. In \cref{thm:qrom}, $h$ is the \emph{smallest} output width of the oracles and trees involved, because $\mathsf{Col}$ and the attach events range over all of them; widening only $\Rsig$ does not change $h$ while $\Ratt$ stays at 256 bits. The commitment term has the form of the generic collision bound~\cite{bht1998}. It is an upper bound, not an attack. Generic quantum collision search is no cheaper than classical $2^{128}$ search in realistic cost models~\cite{bernstein2009collisions,nistcfp2016}, but that does not make the query-model term $2^{-128}$. \Cref{tab:budget} combines all terms of \cref{thm:e2e} for two instantiations. The deployed vault uses Keccak-256 for every tree and for $C_{\cert}$. A uniformly widened variant uses 384-bit nodes everywhere (Keccak-384 off chain, two domain-separated Keccak-256 outputs per node on L1; not implemented). Our \texttt{pqcert} implementation supports Keccak-384 certificate trees, which enlarge Dilithium2 (r3) certificates by $4$--$7\%$ at unchanged verification time at $N=2^{12}$--$2^{16}$ (\cref{tab:h384}).
\section{Leader Election: Proofs and Details}\label{S-consensus}

This section contains the full proofs of the results of \cref{sec:consensus} and material moved from that section.

\subsection{Exponential Race and Steering}\label{S-cons-race}

\begin{restate}{Lemma~\ref{lem:race} (Exponential race)}
Let $U_1,\dots,U_n$ be independent and uniform on $(0,1]$, and let $s_i=-\log_2 U_i/w_i$. Then the $s_i$ are independent with $s_i\sim\mathrm{Exp}(w_i\ln 2)$. Consequently, (i)~$\Pr[i\text{ has the lowest score}]=w_i/W$, and (ii)~splitting stake $w$ into identities $w',w''$ with $w'+w''=w$ leaves the coalition's winning probability unchanged (Sybil invariance). With $h$-bit values $u_i=\lceil U_i2^h\rceil/2^h$ and ties broken by the smaller identity, the winner differs from the continuous winner with probability at most $N(N-1)2^{-h}$.
\end{restate}
\begin{proof}
If $U$ is uniform on $(0,1]$, then $-\ln U\sim\mathrm{Exp}(1)$ and hence $-\log_2 U/w_i\sim\mathrm{Exp}(w_i\ln 2)$. For independent exponentials with rates $r_i$, the minimum is attained at $i$ with probability $r_i/\sum_l r_l$, which gives (i); the minimum of $\mathrm{Exp}(r')$ and $\mathrm{Exp}(r'')$ is $\mathrm{Exp}(r'+r'')$, which gives (ii). For the discretization, if the continuous winner $i$ loses to $j$ after rounding, then $U_j<U_i^{w_j/w_i}\le u_i^{w_j/w_i}\le u_j<U_j+2^{-h}$. So $U_j$ lies in an interval of length $2^{-h}$ below $U_i^{w_j/w_i}$, which has probability at most $2^{-h}$ for each of the $N(N-1)$ ordered pairs.
\end{proof}

\noindent\Cref{lem:race} is the weighted-sampling identity of~\cite{efraimidis2006} applied to the score of~\cite{QRL}. It is a statement about \emph{unselected} uniforms. If a party picks one of several inputs after seeing the oracle's outputs, the selected output is not uniform: keeping the larger of two independent uniforms gives a value with CDF $u^2$, although both inputs were distinct.

\begin{restate}{Proposition~\ref{prop:steer} ($m$-way steering)}
Suppose that, before slot $T$, the adversary can choose $\cS_T$ among $m$ candidates; that under each candidate its own best score is a fresh $\mathrm{Exp}(\alpha)$ variable (rates in units of $\ln 2$), independent across candidates; and that the honest best score is $\mathrm{Exp}(1-\alpha)$ and unobserved when it chooses. Picking the candidate that minimizes its own best score is optimal and leads slot $T$ with probability $p_m(\alpha)=m\alpha/(1+(m-1)\alpha)$.
\end{restate}
\begin{proof}
For a fixed candidate with adversarial best score $X=x$, the adversary wins with probability $p=e^{-(1-\alpha)x}$, which decreases in $x$; minimizing $x$ is therefore optimal given its information. Writing $b=\alpha/(1-\alpha)$, we get $\Pr[p\le v]=\Pr[X\ge-\ln v/(1-\alpha)]=v^{b}$, that is, $p\sim\mathrm{Beta}(b,1)$. The best of $m$ independent candidates yields $\max_l p_l$, whose CDF is $v^{mb}$ and whose mean is $mb/(mb+1)=p_m(\alpha)$.
\end{proof}

\begin{restate}{Corollary~\ref{cor:qrl} (Steering under the runner-up rule)}
Consider proposer-dependent seeds in the race without threshold of~\cite{QRL}, and a single-identity adversary that uses only its decision at slot $T-1$. (a)~If it can observe the honest runner-up's reveal before its publish-or-skip decision becomes irrevocable, for example as a rushing adversary when honest proposals are broadcast at the start of the slot, it captures slot $T$ with probability $\alpha+\alpha^2(1-\alpha)/(1+\alpha)$. (b)~If that reveal becomes known only afterwards, it captures slot $T$ with probability $\alpha\big(\alpha+(1-\alpha)\alpha^{1/(1-\alpha)}\big)+(1-\alpha)\alpha$. Both exceed $\alpha$. Splitting into more identities raises the probability that the adversary holds the top $L\ge 2$ ranks, and hence its option count; Sybil invariance therefore fails.
\end{restate}
\begin{proof}
The adversary leads slot $T-1$ with probability $\alpha$ (\cref{lem:race}); otherwise the seed of slot $T$ is fresh and its capture probability is $\alpha$. In case~(a), it knows both candidate seeds, its own and the runner-up's, and \cref{prop:steer} with $m=2$ gives $p_2(\alpha)=2\alpha/(1+\alpha)$; the total is $\alpha p_2(\alpha)+(1-\alpha)\alpha$. In case~(b), skipping yields an unknown fresh seed with capture probability $\alpha$, while publishing yields probability $p\sim\mathrm{Beta}(b,1)$ as in \cref{prop:steer}, known to the adversary; the best choice achieves $\mathbb E[\max(p,\alpha)]=\alpha+\int_\alpha^1(1-v^b)\,dv=\alpha+(1-\alpha)\alpha^{1/(1-\alpha)}$. Since $\mathbb E[\max(p,\alpha)]>\mathbb E[p]=\alpha$, both totals exceed $\alpha$. If the adversary holds the top $L$ ranks of slot $T-1$, it can choose among up to $L+1$ seeds, and $p_{L+1}(\alpha)$ increases with $L$. Splitting a fixed stake leaves its probability of holding rank~1 unchanged (\cref{lem:race}) but raises its probability of holding ranks~1 and~2 together.
\end{proof}

\subsection{Ideal Model}\label{S-cons-ideal}

\begin{restate}{Theorem~\ref{thm:schedule} (Leader election in the ideal model)}
In any $m$-ideal experiment with continuous scores, the following hold. (a)~Given the reveal set, the proposer ranking of each slot is a fixed public function of the table. No adversarial action within the epoch (skipping, withholding proposals, choosing among its identities) changes any ranking, except by removing the adversary's own identities. If $m=1$, slot leaders are independent across slots and distributed as in \cref{lem:race}. (b)~The adversary chooses the table among at most $m\le 2^c$ candidates by withholding a subset of its late reveals. A slot is not guaranteed with probability at most $p_B$ under each candidate, and the adversary can propose only in such slots. It proposes in a fixed target slot with probability at most $1-(1-p_B)^m$, and its expected fraction of the $k$ slots of the epoch is at most $p_B+\sqrt{c\ln 2/(2k)}$. (c)~The probability that none of $n$ specified slots is guaranteed is at most $\min\{1,mp_B^{n}\}$. For a window of $n$ consecutive slots whose start the adversary may choose after seeing the schedule, it is at most $\min\{1,(k-n+1)mp_B^{n}\}$.
\end{restate}
\begin{proof}
For part (a), the ranking of slot $j$ sorts row $j$ of the table, which no action within the epoch changes. Removing its own identities cannot help the adversary. For $m=1$, all scores are independent and uniform, and \cref{lem:race} gives the leader distribution and independence across slots.

We next bound one slot under one candidate table. Given the first $r-1$ ranks, rank $r$ is adversarial with probability $(\alpha-a)/(1-a-g)$, where $a$ and $g$ are the adversarial and honest stake already ranked. This is at most $\alpha/(1-g)$, since $(\alpha-a)(1-g)\le\alpha(1-a-g)$ is equivalent to $a(\alpha+g-1)\le0$. For $r\le B+1$, at most $B$ honest identities have been ranked, so $g\le Bw_{\max}$ and the probability is at most $\alpha_B$. Hence the first $B+1$ ranks are all honest with probability at least $(1-\alpha_B)^{B+1}$, and a slot is not guaranteed with probability at most $p_B$. In a guaranteed slot, at most $B$ of the first $B+1\le L+1$ ranks are suppressed, and the best-ranked surviving one is honest and proposes.

For part (b), these per-candidate bounds hold independently across candidates, so the adversary proposes in the target slot under at least one candidate with probability at most $1-(1-p_B)^m$; no selection rule does better. For the share, let $\bar X_l=\frac1k\sum_j X_{l,j}$, where $X_{l,j}$ indicates that slot $j$ is not guaranteed under candidate $l$. Under each candidate, the $X_{l,j}$ are independent with mean at most $p_B$. By Hoeffding's lemma, each $\bar X_l$ minus its mean is sub-Gaussian with variance proxy $1/(4k)$. The maximal inequality for sub-Gaussian variables, which requires no independence between candidates, gives $\mathbb{E}[\max_{l\le m}\bar X_l]\le p_B+\sqrt{\ln m/(2k)}$, and $\ln m\le c\ln 2$.

For part (c), the event that none of the $n$ slots is guaranteed under the selected table is contained in the union of the same events under the $m$ candidates. Under each candidate, slots are independent, and each fails with probability at most $p_B$. For an adversarially placed window, we take a further union bound over the $k-n+1$ possible starts.
\end{proof}

\begin{restate}{Lemma~\ref{lem:incl} (Reveals are included)}
Suppose that an honest proposer includes every reveal it received $\Delta_{\mathcal S}$ before its slot, that messages are delivered within $\Delta_{\mathcal S}$ and a slot lasts at least $(L+2)\Delta_{\mathcal S}$, and that the finality gadget finalizes each block of the reveal window before the next slot starts. If a phase or sub-window of $\ell$ slots has a guaranteed slot after its first slot, every honest reveal of that phase or sub-window is included. Under each candidate table of the epoch that contains the window, this fails with probability at most $p_B^{\ell-1}$.
\end{restate}
\begin{proof}
An honest validator sends its reveal when its phase or sub-window starts, so every honest proposer of the second and later slots of the window has received it in time, and the proposer of a guaranteed slot is honest. The positions of the windows are fixed by the order, and under each candidate table the slots are independent, each not guaranteed with probability at most $p_B$ by \cref{thm:schedule}.
\end{proof}

\begin{restate}{Proposition~\ref{prop:late} (Late reveals are rare)}
Consider consecutive $m$-ideal epochs, each with respect to its own $c$, starting from an epoch with $c=0$. Let $\mu=(1-2\alpha)/(1-3\alpha)$, and let $\mathsf{OK}_e$ be the event that $\mathsf{Fail}$ occurs in none of the epochs up to $e$. Then, for every epoch $e$ and every $x\ge1$,
\begin{gather*}
\Pr[c_e\ge x\wedge\mathsf{OK}_e]\le\mu\alpha^x,\qquad\mathbb E[2^{c_e}\mathbf 1_{\mathsf{OK}_e}]\le\mu,\\
\Pr[\mathsf{OK}_{e-1}\wedge\neg\mathsf{OK}_e]\le\delta_{\mathsf{LE}},
\end{gather*}
where $\delta_{\mathsf{LE}}=\mu(\alpha^{L_t}+L_tp_B^{\ell_s-1}+p_B^{\ell_0-1})$. The adversary proposes in a fixed target slot of epoch $e$ with probability at most $\pi_{\mathsf{cap}}+e\delta_{\mathsf{LE}}$, where $\pi_{\mathsf{cap}}=p_B+\sum_{x\ge1}\min\{1,\mu\alpha^x\}\big((1-p_B)^{2^{x-1}}-(1-p_B)^{2^x}\big)$.
\end{restate}
\begin{proof}
Write $m_e=2^{c_e}$. The registrations for epoch $e$ and, by mild adaptivity, the corrupt set acting during its reveal window, which lies in epoch $e-1$ and was therefore requested by epoch $e-4$, are fixed before the candidate tables of epoch $e-1$ are sampled at that epoch's time $\tau^\ast$, in epoch $e-2$. Under each candidate table of epoch $e-1$, the order of epoch $e$ is a stake-weighted sample without replacement. As in the proof of \cref{thm:schedule} with $g=0$, its first $x$ ranks are all adversarial with probability at most $\alpha^x$. In particular, the tail contains no honest identity with probability at most $\alpha^{L_t}$, and by \cref{lem:incl} some of the $L_t$ sub-windows, which lie in epoch $e-1$ at positions fixed by the order, has no guaranteed slot after its first slot with probability at most $L_tp_B^{\ell_s-1}$; for the common phase, the bound is $p_B^{\ell_0-1}$. On $\mathsf{OK}_{e-1}$, the adversary selects one of $m_{e-1}$ candidates, and the candidate tables are independent of $m_{e-1}$. Summing over candidates gives $\Pr[c_e\ge x\wedge\mathsf{OK}_e]\le\mathbb E[m_{e-1}\mathbf 1_{\mathsf{OK}_{e-1}}]\alpha^x$ and $\Pr[\mathsf{OK}_{e-1}\wedge\neg\mathsf{OK}_e]\le\mathbb E[m_{e-1}\mathbf 1_{\mathsf{OK}_{e-1}}](\alpha^{L_t}+L_tp_B^{\ell_s-1}+p_B^{\ell_0-1})$. On $\mathsf{OK}_e$, no honest reveal is missing, so suspensions remove only adversarial identities and do not raise $\alpha$. Since $2^{c}=1+\sum_{x=1}^{c}2^{x-1}$,
\begin{align*}
\mathbb E[2^{c_e}\mathbf 1_{\mathsf{OK}_e}]&\le1+\sum_{x\ge1}2^{x-1}\Pr[c_e\ge x\wedge\mathsf{OK}_e]\\
&\le 1+\mathbb E[m_{e-1}\mathbf 1_{\mathsf{OK}_{e-1}}]\frac{\alpha}{1-2\alpha},
\end{align*}
and induction from $m=1$ gives the bound $\mu$, the fixed point of $y\mapsto 1+y\alpha/(1-2\alpha)$. The capture bound averages \cref{thm:schedule}(b) over $c_e$ on $\mathsf{OK}_e$, and adds $\Pr[\neg\mathsf{OK}_e]\le e\delta_{\mathsf{LE}}$.
\end{proof}

\subsection{From the Ideal Model to the Protocol}\label{S-cons-real}

\begin{restate}{Lemma~\ref{lem:bbbv} (Reprogramming hidden inputs)}
Let a quantum process make $q$ queries to an oracle $O$, and let its run with $O$ be determined by randomness independent of a random variable $Z$, so that, with that randomness fixed, its state before each query does not depend on $Z$. Let $S_Z$ be a set of inputs, determined by $Z$, with $\Pr_Z[x\in S_Z]\le p$ for every input $x$, and let $O'$ agree with $O$ outside $S_Z$. Then the final states of the process with $O$ and with $O'$ have expected trace distance at most $2q\sqrt{p}$.
\end{restate}
\begin{proof}
Purify the process, including all parties and their classical interaction, and defer all measurements; queries made in parallel are counted separately. Fix all randomness except $Z$. Let $\lvert\phi_i\rangle$ be the state before the $i$-th query when the oracle is $O$, and let $\Pi$ project the query register onto $S_Z$. Replacing the oracle by $O'$ one query at a time, and using that later operations preserve norms, gives $\|\lvert\psi^{O}\rangle-\lvert\psi^{O'}\rangle\|\le\sum_{i=1}^{q}\|(U_O-U_{O'})\lvert\phi_i\rangle\|\le 2\sum_{i}\|\Pi\lvert\phi_i\rangle\|$, because $U_O$ and $U_{O'}$ act identically on inputs outside $S_Z$. Since $\lvert\phi_i\rangle$ does not depend on $Z$, $\mathbb E_Z\|\Pi\lvert\phi_i\rangle\|\le(\mathbb E_Z\|\Pi\lvert\phi_i\rangle\|^2)^{1/2}=(\sum_x\|\Pi_x\lvert\phi_i\rangle\|^2\Pr[x\in S_Z])^{1/2}\le\sqrt p$. The trace distance of pure states is at most their Euclidean distance.
\end{proof}

\begin{restate}{Theorem~\ref{thm:fresh} (Ideal to real, QROM)}
Model the commitment and seed functions of QPoS-S, cSHAKE256 with distinct tags, as a quantum random oracle independent of the score oracle $\RO$ and of ML-DSA's internal hashing. Let commitments, secrets and seeds have $n_c$, $n_r$ and $n_s$ bits, and let identities and epochs be fixed-length fields of $n_p$ bits in total. Assume mildly adaptive corruption (\cref{sec:model}), the conditions of \cref{lem:incl}, that $\mathsf{Fail}$ does not occur, and that $c\le\bar c$; on $\neg\mathsf{Fail}$, $\bar c=L_t-1$ suffices. Let $m=2^{\bar c}$. Let $t$ count all oracle queries made by all parties before $\tau^\ast$. Then there is an $m$-ideal experiment such that, for every $[0,1]$-valued function $F$ of the epoch's leaders and the adversary's final output,
\begin{multline*}
\big|\mathbb E_{\mathsf{real}}[F]-\mathbb E_{\mathsf{ideal}}[F]\big|\le\epsilon_{\mathsf{lead}}\\=\frac{2^{n_p+2n_r}}{2^{n_c+1}}+\frac{2t}{2^{n_r/2}}+\frac{2t\sqrt m}{2^{n_s/2}}+\frac{m^2}{2^{n_s+1}}.
\end{multline*}
Queries made after $\tau^\ast$ do not appear in the bound.
\end{restate}
\begin{proof}
Write $\RO_c$ for the commitment and seed oracle. Index the possible reveal sets by bitmasks $b\in\{0,1\}^{\bar c}$ over the late adversarial identities in reveal order, and let $X_b$ denote the seed input of reveal set $b$.

\emph{Step 0 (binding).} Modify the protocol so that a reveal is accepted only if it is the lexicographically smallest opening of its commitment. The modified and real experiments are identical unless, for some identity-and-epoch prefix $p$, the map $r\mapsto\RO_c(\textsc{commit}\,\|\,p\,\|\,r)$ is not injective on $\{0,1\}^{n_r}$. This event depends only on the uniform oracle and has probability at most $2^{n_p}\cdot2^{2n_r}/2^{n_c+1}$. In the modified protocol, each adversarial identity can only reveal its unique admissible opening or withhold it. The admissibility check for an identity $\mathit{ID}_j\ne\mathit{ID}_{i^\ast}$ reads $\RO_c$ only on commitment inputs with that prefix, which lie outside the set $S$ of Step~1; in Steps~1 and~2 these values are part of the fixed randomness. The reveal list, and hence each $X_b$, is thus a function of $b$, of the commitments and honest secrets, and of $\RO_c$ on commitment inputs. Registrations, which also fix the stakes, close before the reveal window opens, so all commitments and weights are fixed before $\tau^\ast$. Every reveal before the sub-window of $i^\ast$ is final at $\tau^\ast$, so the adversary's only remaining choice is which late identities reveal. The reveal order is a function of the previous seed, which is public before $\tau^\ast$. By mild adaptivity, the corrupt set during the reveal window was requested before that seed was known, so $i^\ast$ is honest until after $\tau^\ast$ and is part of the fixed randomness as well.

\emph{Step 1 (hiding the last honest secret).} Let $r=r_{i^\ast}$. Let $S$ be the set of commitment inputs with fields $(\mathit{ID}_{i^\ast},e,r)$, together with the seed inputs whose entry for $\mathit{ID}_{i^\ast}$ equals $r$. Fixed-length fields and the strictly ordered encoding of the reveal list give $\Pr_r[x\in S]\le 2^{-n_r}$ for every input $x$. Sample a function $H$ on commitment and seed inputs, an independent function $W$ on seed inputs, and a uniform value $a$. Let $G$ map the commitment input of $i^\ast$ to $a$, agree with $W$ on the seed inputs in $S$, and agree with $H$ elsewhere. For every fixed $r$, $G$ is a uniform oracle, so the experiment of Step~0, in which $\RO_c=G$ throughout and $a_{i^\ast}=a$, is distributed as before.

Let $\mathsf{Hyb}_A$ answer commitment and seed queries before $\tau^\ast$ with $H$ and later queries with $G$. Validator $i^\ast$ uses $r$ before $\tau^\ast$ only through $a_{i^\ast}=a$, so we keep $r$ outside the process. Given $(H,W,a)$ and the other randomness, the run before $\tau^\ast$ in $\mathsf{Hyb}_A$ therefore does not depend on $r$. \Cref{lem:bbbv} with $Z=r$ bounds the trace distance of the states at $\tau^\ast$ by $2t\,2^{-n_r/2}$. The execution after $\tau^\ast$ is the same map in both experiments, so the bound carries over to the final states, however many queries follow.

\emph{Step 2 (fresh scores).} Call $b$ \emph{realizable} if its bits name existing late identities that all have an admissible opening. In $\mathsf{Hyb}_A$, the points $X_b$ of realizable $b$ are distinct seed inputs in $S$; for other $b$, the value $Z_b$ below is an unused dummy. They are determined by $r$ and by the run before $\tau^\ast$, which uses $H$ and $a$ but not $W$. Their values $W(X_b)$ are therefore independent uniforms, independent of $(H,a,r)$ and of that run. Hence $\mathsf{Hyb}_A$ is unchanged in distribution if we sample independent uniform $Z_b$ in advance and set $W(X_b)=Z_b$, so that the seed of reveal set $b$ is $\sigma_b=Z_b$. Let $S'$ be the set of score and next-epoch order inputs whose seed field equals some $Z_b$. $S'$ is a function of $Z=(Z_b)_b$ alone, and $\Pr_Z[x\in S']\le m2^{-n_s}$ for every input $x$. Let $\mathsf{Hyb}_B$ be $\mathsf{Hyb}_A$ with fresh uniform values $V'$ on $S'$ for queries after $\tau^\ast$. Because $S'$ depends only on $Z$, which is independent of the score oracle $\RO$, replacing the values of $\RO$ on $S'$ by $V'$ both before and after $\tau^\ast$ leaves $\mathsf{Hyb}_A$ unchanged in distribution. The result differs from $\mathsf{Hyb}_B$ only in the answers on $S'$ before $\tau^\ast$. The run of $\mathsf{Hyb}_B$ before $\tau^\ast$ does not depend on $Z$, so \cref{lem:bbbv} bounds the distance by $2t\sqrt m\,2^{-n_s/2}$.

\emph{Step 3 (the ideal experiment).} In $\mathsf{Hyb}_B$, the seeds $Z_b$ are distinct except with probability $m^2/2^{n_s+1}$. The score table and next-epoch order keys of reveal set $b$ are then $V'$ on the inputs of $S'$ with seed $Z_b$, a family of independent uniform tables, independent of everything before $\tau^\ast$ and accessible to the adversary afterwards. $\mathsf{Hyb}_B$ is therefore $m$-ideal, and the triangle inequality gives the bound, because trace distance bounds the difference of expectations of $[0,1]$-valued functions.
\end{proof}

\begin{restate}{Corollary~\ref{cor:real} (Real protocol)}
Let $m=2^{L_t-1}$ and $\epsilon_{\mathsf{LE}}=\epsilon_{\mathsf{lead}}+\epsilon_{\mathsf{disc}}$, where $\epsilon_{\mathsf{disc}}=(k+1)mN(N-1)2^{-h}$ applies \cref{lem:race} to the $k$ slot races and the order race of each candidate, and $N$ counts all registered identities. Against quantum adversaries with mildly adaptive corruption, under the conditions of \cref{lem:incl}, every bound of \cref{thm:schedule,prop:late,lem:incl} holds in the real protocol, for the first $E$ epochs, up to an additive $2E(\epsilon_{\mathsf{LE}}+\delta_{\mathsf{LE}})$. In particular, the adversary proposes in a fixed target slot with probability at most $\pi_{\mathsf{cap}}+2E(\epsilon_{\mathsf{LE}}+\delta_{\mathsf{LE}})$, and an epoch has no guaranteed slot among $n$ consecutive slots placed by the adversary with probability at most $(k-n+1)\mu p_B^{n}+2E(\epsilon_{\mathsf{LE}}+\delta_{\mathsf{LE}})$.
\end{restate}
\begin{proof}
Replace the epochs by $m$-ideal ones one at a time. Each step applies \cref{thm:fresh}, whose function $F$ may depend on the entire later run, and \cref{lem:race}. Its premises fail only on $\neg\mathsf{OK}_e$, which has probability at most $\delta_{\mathsf{LE}}$ given $\mathsf{OK}_{e-1}$ by \cref{prop:late}, applied to the already idealized earlier epochs. Each step therefore costs at most $\epsilon_{\mathsf{LE}}+\delta_{\mathsf{LE}}$, and $\Pr[\neg\mathsf{OK}_E]\le E\delta_{\mathsf{LE}}$ adds the rest. The bounds then hold for consecutive $m$-ideal epochs, and the window bound averages \cref{thm:schedule}(c) over $m_e$ with $\mathbb E[m_e\mathbf 1_{\mathsf{OK}_e}]\le\mu$.
\end{proof}

\subsection{Additional Remarks}\label{S-cons-remarks}

\paragraph{Proposer-dependent seeds} In the hash-chain protocol of~\cite{QRL} and in the sidechain design of~\cite{myeprint}, each validator's score is computed from its own secret hash-chain value, so scores are private until their owners propose, and the seed is updated as $\cS_{j+1}=\Hash(\cS_j\,\|\,\text{reveal of the \emph{actual} proposer of }j)$, and a skipped leader is replaced by the runner-up. A coalition can then influence $\cS_{j+1}$. If it leads slot $j$, it can publish its own block or skip and let the runner-up publish. If it holds the top $L$ positions, it can make any of $L+1$ proposers publish. Because the adversary knows its own future chain values, it can compute its \emph{own} best score in the next slot for each candidate seed. Honest scores remain unpredictable.

\paragraph{Target slots} Target slots matter because roles are attached to slot indices, for example the $N_{\mathsf{cp}}$-th block that fixes a checkpoint, or the slot whose block includes or censors a dispute-relevant transaction.

\paragraph{Fallback proposers} The fallback proposers matter only under denial of service. Without them ($L=0$), suppressing the leader of every slot leaves the sidechain without honest blocks; with $L\ge B$, a slot loses its honest proposer only if it is not guaranteed. A non-guaranteed slot may still have an honest proposer, but an adversarial proposer of better rank can take the slot and censor it.

\paragraph{Reveal order} The order of reveals keeps $c$ small. Because rank~$1$ of the order is drawn in proportion to stake, Sybil splitting does not help the adversary obtain the last positions, and the adversary can influence the order only through its choice among the previous epoch's candidates.

\paragraph{Parameters of \cref{thm:fresh}} QPoS-S uses $n_r=n_s=512$, $n_c=1536$ and $64$-bit identity and epoch fields ($n_p=128$). For $t=2^{64}$ and $\bar c=L_t-1=24$, the four terms are $2^{-385}$, $2^{-191}$, $2^{-179}$ and $2^{-465}$, so $\epsilon_{\mathsf{lead}}\le 2^{-178.9}$. The value $\bar c=24$ is a worst case; by \cref{prop:late}, $c$ is usually $0$ or $1$. Scores and order keys use $\Hash$ with $h=256$. The wide secrets and seeds keep the two hiding terms small. The wide commitments make binding statistical in the QROM, so the proof needs no argument about computationally binding commitments against quantum adversaries. In practice, SHAKE-based binding is capped at the function's 256-bit security claim.

\subsection{Consequences in Detail}\label{S-cons-consequences}

Part~(c) does not multiply per-slot probabilities, which is invalid after seed selection. Since the schedule is public, the adversary can time an action, such as the rebinding of \cref{prop:binding}, to precede a run of non-guaranteed slots. By \cref{cor:real}, an honest transaction is still included within $n$ slots of any start chosen after the schedule is public, except with probability $\zeta+2E(\epsilon_{\mathsf{LE}}+\delta_{\mathsf{LE}})$, whenever $n\ge\ln(k\mu/\zeta)/(-\ln p_B)$. This high-probability guarantee is the sidechain liveness that \cref{prop:binding} uses. Turning inclusion into finality additionally needs the liveness of the chosen finality gadget, and a window that straddles an epoch boundary adds the bounds of both epochs. For $\alpha=0.3$ and $k=2000$, $\zeta=2^{-40}$ needs $n=31$ slots without denial of service and $n=56$ with $B=1$. For $E=2^{20}$ epochs, $2E(\epsilon_{\mathsf{LE}}+\delta_{\mathsf{LE}})<2^{-19}$ at $\alpha=0.3$ and $<2^{-36}$ at $\alpha=0.2$, with or without $B=1$. Larger $L_t$ or $\ell_s$ lowers $\delta_{\mathsf{LE}}$ further at the cost of a longer reveal window. The window of $1{,}352$ slots occupies about two thirds of a $2{,}000$-slot epoch; reveals are small transactions carried by ordinary blocks, so this costs block space, not slots. A validator whose reveal is missing is suspended, but the adversary withholds only when the honest seed does not already suit it, so suspension is not a deterrent in our analysis. PQ-Vault attaches no exclusive role to any single slot: certificates can be aggregated by anyone, and checkpoints are certified by attestors rather than by a leader.

\subsection{Long-Range Attacks}\label{S-cons-longrange}

A posterior corruption of retired keys allows the adversary to build an alternative history~\cite{gazi2018stakebleeding,deirmentzoglou2019survey}. Signature unforgeability, on which the analysis of~\cite{myeprint} relied, does not prevent this. In PQ-Vault, $\Vault$ accepts only checkpoints for the epoch after its latest finalized checkpoint, verified against the attestor commitment $\Ratt^{e}$ fixed by that checkpoint. Bootstrapping nodes use the vault log as their weak-subjectivity source. Under \cref{asm:l1}, an alternative history that forks before the latest anchored checkpoint is rejected outright. A fork after it requires a certificate for the current attestor set, whose members are still bonded by our unbonding condition; forging one is bounded by \cref{thm:qrom,thm:opf}.
\section{L1 Acceptance: Dispute Protocol and Proofs}\label{S-opf}

This section supplements \cref{sec:opf}. It gives the full specification of the optimistic dispute protocol, the detailed timing discussion, and the full proofs of \cref{prop:snark,lem:local,thm:opf,thm:e2e}.

\subsection{Record Commitment}
The relayer commits to the records in a Merkle tree $C_{\cert}$ of depth $\lceil\log_2\kr\rceil$, with $r_j$ at index $j$. Its leaf hashes a fixed-length encoding of $r_j$ that includes the hashes of $\sigma_{i_j}$ and of both paths.

\subsection{State Machine}\label{S-opf-sm}
A proposal is created by $\mathsf{propose}$ (\cref{sec:vault}) at time $\tau_0$ and is \emph{pending} until it is finalized, rejected or stale. The relayer publishes the records as blob data. Safety never relies on blobs: a watcher that cannot fetch the records or match them to $C_{\cert}$ uses the open-all challenge. Its \emph{challenge deadline} is $\tau_0+T_{\mathsf{ch}}$, extended to $\max\{\tau_0+T_{\mathsf{ch}},\tau_{\mathsf{post}}+\Delta_{\mathsf{post}}\}$ if an open-all posting completed at $\tau_{\mathsf{post}}$. Times are block timestamps at inclusion. A challenge or open-all request is accepted only at timestamps strictly before its deadline, and a response or posting at timestamps up to and including its deadline; $\mathsf{finalize}$ is accepted from the challenge deadline on, and $\mathsf{timeout}$ and $\mathsf{timeoutAll}$ only strictly after their deadlines. Hence no timestamp admits both an action and the transition that its deadline unlocks. The transitions are as follows. (a)~$\mathsf{challenge}(j)$, with bond $b_C$, before the challenge deadline and at most once per index, opens a coin challenge with response deadline $\tau+\Delta_{\mathsf{resp}}$. (b)~$\mathsf{respond}(j,r,\mathit{addr})$ by any party up to and including that deadline first requires that $r$ authenticate under $C_{\cert}$ at index $j$, and otherwise \emph{reverts}. If the committed data fail (i) or (ii), the proposal is \emph{rejected}. Next, the call reverts unless the code hash of $\mathit{addr}$ equals $\eta_{i_j}$. Finally, a failing signature rejects the proposal, and a valid one closes the challenge. (c)~$\mathsf{timeout}(j)$ after an unanswered response deadline rejects the proposal. (d)~$\mathsf{challengeAll}$, at most once per proposal and before the challenge deadline, requires all records to be posted as calldata within $\Delta_{\mathsf{open}}$. $\mathsf{postRecords}$ accepts records that authenticate under $C_{\cert}$ without evaluating $\mathsf{Local}$, and sets $\tau_{\mathsf{post}}$ once all $\kr$ are posted. If posting is incomplete at its deadline, $\mathsf{timeoutAll}$ rejects the proposal. (e)~$\mathsf{finalize}$ requires that the challenge deadline has been reached, that no coin challenge is open, and that no posting is pending. The proposal is then finalized if its epoch is $e^\ast+1$, and stale otherwise.
A proposal is thus rejected only on an authenticated record that fails $\mathsf{Local}$ or on an expired deadline; a malformed record, wrong path or wrong key address from a third party only reverts. On rejection, the challenger is credited $b_R$ and its bond; a correct answer credits $b_C$ to the relayer. Payments are pulled by their recipients, so a recipient that refuses transfers cannot block a transition. An open-all request carries no bond, so a withholding relayer cannot charge honest watchers; in exchange, anyone can force each proposal to post its records once, at a cost to the relayer of $3.0\times10^7$ gas (\cref{sec:gas}). Unavailable data cannot be attributed to either party, so one of them must bear this cost; we protect watchers.

\subsection{Timing}\label{S-opf-timing}
Let $\Delta_{\mathsf{chk}}$ bound the time an honest watcher needs to observe a proposal or a completed posting, fetch the records and evaluate $\mathsf{Local}$ on all of them, and let $\Delta_{\mathsf{prep}}$ bound the relayer's time to prepare a response or posting. All $\Delta$'s bound differences of inclusion timestamps. We require the strict inequalities
\[
T_{\mathsf{ch}}> \Delta_{\mathsf{chk}}+\Delta_{\mathsf{inc}},\qquad \Delta_{\mathsf{post}}> \Delta_{\mathsf{chk}}+\Delta_{\mathsf{inc}},
\]
which match the strict guards on challenges, and $\Delta_{\mathsf{open}},\Delta_{\mathsf{resp}}\ge\Delta_{\mathsf{prep}}+\Delta_{\mathsf{inc}}$, which match the inclusive guards on postings and responses; the relayer submits all posting transactions (seven in our measurements) concurrently. Every pending proposal is then \emph{eligible} for a final transition ($\mathsf{finalize}$, $\mathsf{timeout}$ or $\mathsf{timeoutAll}$) by $T_{\mathsf{fin}}=T_{\mathsf{ch}}+\Delta_{\mathsf{open}}+\Delta_{\mathsf{post}}+\Delta_{\mathsf{resp}}$ after its inclusion. A transition executes only when some party submits it. If an honest party submits it once it is eligible, it is included by $T_{\mathsf{fin}}+\Delta_{\mathsf{inc}}$. Safety never depends on these calls, since a proposal with an open challenge or pending posting cannot be finalized. A stale proposal has no effect on the vault's state from the moment another proposal for its epoch is finalized; only its bond refund needs a call. The worst case is an open-all request just before $\tau_0+T_{\mathsf{ch}}$, followed by a coin challenge just before the extended deadline. Our contract uses $T_{\mathsf{ch}}=7$ days, $\Delta_{\mathsf{open}}=\Delta_{\mathsf{resp}}=1$ day and $\Delta_{\mathsf{post}}=2$ days, in line with optimistic rollups~\cite{kalodner2018arbitrum}, so $T_{\mathsf{fin}}\le 11$ days. Each coin index can be challenged at most once, so at most $\kr$ responses are ever needed per proposal.

\subsection{Proofs}

\begin{restate}{Proposition~\ref{prop:snark}}
If $\Vault$ finalizes a checkpoint whenever a Groth16 proof $\Pi$ of ``all certificates verify'' is valid, then a QPT adversary can finalize an arbitrary checkpoint, and with it arbitrary withdrawals, with probability $1-\negl$.
\end{restate}
\begin{proof}
The common reference string contains $[\tau^i]_1$ and $[\tau]_2$ for a trapdoor $\tau$. Shor's algorithm computes $\tau$ from $[\tau]_1$. With $\tau$ (and the other trapdoor elements, computed the same way), the Groth16 simulator outputs a valid proof for any statement, true or false. The vault accepts it.
\end{proof}

\begin{restate}{Lemma~\ref{lem:local} (Local decomposition)}
Except for a collision of $\Hash$, $\mathsf{Local}(j,r_j)=1$ for all $j<\kr$ if and only if the distinct positions $P=\{i_j\}$, their leaves, and the multiproofs assembled from the single paths form a certificate that \cref{alg:cert} accepts, instantiated with the vault's leaf encodings.
\end{restate}
\begin{proof}
If two records name the same position, both of their leaves open there under $\Rsig$ and $\Ratt$, so they are equal unless the paths contain a collision. Hence $P$ is well defined with $|P|\le\kr$, every position covers its coin, and the single paths contain every sibling that the multiproofs need. The interval, coverage, signature and opening checks of \cref{alg:cert} are exactly (i)--(iii) taken over all $j$, and the range check on $\Ws$ was performed at proposal time. Conversely, an accepted certificate yields, for each coin, a covering leaf whose record passes $\mathsf{Local}$.
\end{proof}

\begin{restate}{Theorem~\ref{thm:opf} (Optimistic acceptance)}
Assume \cref{asm:l1,asm:pq,asm:slow,asm:watcher} and the timing conditions above. Then the following hold. (Safety) The probability that $\Vault$ finalizes a proposal whose $M_e$ meets the winning condition of \cref{def:game}, with the attestor list committed by the proposal's snapshot of $\Ratt$, is at most the bound of \cref{thm:qrom} plus $\delta_{\mathsf{col}}=(12t^3+2q_C)/2^{h_C}$, where $q_C=2\kr(\lceil\log_2\kr\rceil+2d+7)$ and $h_C$ is the smallest output width of $C_{\cert}$, $\Rsig$ and $\Ratt$ ($h_C=256$ in the deployed vault). Since $\mathsf{Col}$ ranges over the whole oracle, $\delta_{\mathsf{col}}$ also excludes two key contracts with the same code hash $\eta$, and, in \cref{thm:e2e}, two headers with the same $M_e$ or two withdrawal trees with the same root. (Liveness) An honest relayer's proposal is never rejected except with probability $\delta_{\mathsf{col}}$; it is eligible for finalization at $\tau_0+T_{\mathsf{ch}}$ if no dispute is raised, and by $\tau_0+T_{\mathsf{fin}}$ otherwise; if the relayer then submits $\mathsf{finalize}$, it is finalized within a further $\Delta_{\mathsf{inc}}$. (Cost) Each frivolous coin challenge costs its sender $b_C$, and an honest watcher spends at most one open-all transaction fee and one bond $b_C$ per proposal; if its challenge is the first on a failing record, the bond is returned with $b_R$.
\end{restate}
\begin{proof}
Two different records that authenticate at the same index of $C_{\cert}$, or two different leaves that open at the same position of $\Rsig$ or $\Ratt$, contain a collision; by \cref{lem:extract} this has probability at most $\delta_{\mathsf{col}}$, and we exclude it. The \emph{committed records} $(r_j)_{j<\kr}$ are then the unique records that authenticate under $C_{\cert}$.

\emph{Safety.} Suppose the vault finalizes a proposal whose $M_e$ meets the winning condition. By \cref{asm:watcher}, an honest watcher observes the proposal and, within $\Delta_{\mathsf{chk}}$, either evaluates $\mathsf{Local}$ on blob records that authenticate under $C_{\cert}$ or finds them missing or inconsistent. In the latter case it calls $\mathsf{challengeAll}$, which is included by $\tau_0+\Delta_{\mathsf{chk}}+\Delta_{\mathsf{inc}}<\tau_0+T_{\mathsf{ch}}$ by \cref{asm:l1}, and hence accepted. Finalization requires that any open-all posting completed, since otherwise $\mathsf{timeoutAll}$ would have rejected the proposal, and then the watcher evaluates $\mathsf{Local}$ on the posted records within $\Delta_{\mathsf{chk}}$ of $\tau_{\mathsf{post}}$. Either way, the watcher learns all committed records before the challenge deadline. If one of them, $r_{j^\ast}$, fails $\mathsf{Local}$, index $j^\ast$ is challenged strictly before the deadline, by the watcher or by another party, since the watcher's challenge is included within $\Delta_{\mathsf{chk}}+\Delta_{\mathsf{inc}}$ of the proposal or of $\tau_{\mathsf{post}}$, which is less than $T_{\mathsf{ch}}$ or $\Delta_{\mathsf{post}}$. The only response that does not revert carries $r_{j^\ast}$ and rejects the proposal, and an unanswered challenge rejects it at $\mathsf{timeout}$; finalization waits for this outcome. Hence every committed record passes $\mathsf{Local}$, and by \cref{lem:local} they form a certificate that \cref{alg:cert} accepts for $M_e$. An adversary that achieves this wins $\mathsf{CERT\text{-}FORGE}$, which \cref{thm:qrom} bounds. The vault does not authenticate challengers, and \cref{asm:slow} protects only the watcher's bond.

\emph{Liveness and cost.} An honest relayer publishes records that authenticate and pass $\mathsf{Local}$. Every response it gives to a coin challenge is accepted, and responses by other parties either revert or close the challenge. It posts all records within $\Delta_{\mathsf{open}}$ if an open-all is requested. Its responses and postings are included by their inclusive deadlines because $\Delta_{\mathsf{resp}},\Delta_{\mathsf{open}}\ge\Delta_{\mathsf{prep}}+\Delta_{\mathsf{inc}}$. No rejection condition can then occur, and the proposal is eligible for finalization, or becomes stale if another proposal for its epoch finalized first, once its deadlines have passed, which happens by $\tau_0+T_{\mathsf{fin}}$. The cost statements follow from the bond rules and the once-per-proposal limit on open-all.
\end{proof}

\begin{restate}{Theorem~\ref{thm:e2e} (Custody safety)}
Assume \cref{asm:l1,asm:pq,asm:slow,asm:watcher}, the timing conditions above, and, in every epoch $e$, corrupt attestor stake below $W^e/3$ with $\Wp^e=\lfloor W^e/3\rfloor$. Assume also the honest-signing rule of \cref{sec:design}, and let $\mathsf{Fin}$ be the event that two honest validators consider conflicting checkpoints of the same epoch final. Then over $E$ epochs, for every QPT adversary making $t$ oracle queries in total,
\begin{multline*}
\Pr[\text{G1 violated}]\le E\Big(12t^2\varepsilon^\ast+\frac{48t^3+2q_V}{2^h}+\delta_{\mathsf{col}}\Big)\\+2E\Adv^{\mathsf{uf}}_{N_{\mathsf{tot}}}+\Adv^{\mathsf{uf}}_{U}+\Pr[\mathsf{Fin}],
\end{multline*}
where $\varepsilon^\ast\le 2^{-\lambda}/(12t^2)$ under \cref{eq:rule} when $t\le q(1+2^{-20})$ (\cref{sec:rule}), $N_{\mathsf{tot}}$ is the number of honest attestor keys generated over the whole experiment, and $U$ is the number of user keys.
\end{restate}
\begin{proof}
We show by induction on $e$ that, outside the bad events below, the following invariant $I_e$ holds after the vault finalizes its epoch-$e$ checkpoint: (a)~the finalized $M_1,\dots,M_e$ are the checkpoints of the honest finalized sidechain; (b)~the vault's $(\Ratt,W,N)$ is the honest $(\Ratt^{e+1},W^{e+1},N^{e+1})$; and (c)~every payment made by $\mathsf{withdraw}$ so far pays a leaf of some finalized $\omega_{e'}$, each leaf at most once. $I_0$ holds by deployment: $\Ratt^{1}$ is fixed in the constructor, and nothing has been paid.

Assume $I_{e-1}$, and consider the first proposal for epoch $e$ that the vault finalizes. Every proposal for epoch $e$ is opened while $e^\ast=e-1$, so its snapshot of $(\Ratt,W,\Wp,N)$ is the honest one of epoch $e$ by $I_{e-1}$(b), and its $M_e$ is computed by the vault over $M_{e-1}$ and the proposed header. The attestors of epoch $e$ are still bonded at that time by our unbonding condition, so their corrupt stake is below $W^e/3$ and hence at most $\Wp^e$. If no honest attestor signed $M_e$, the winning condition of \cref{def:game} holds for this list. Consider the experiment truncated at this finalization, in which $\mathcal A$ makes at most $t$ queries. By \cref{thm:opf}, on the direct path by \cref{thm:qrom}, this happens with probability at most $12t^2\varepsilon^\ast+(48t^3+2q_V)/2^h+\delta_{\mathsf{col}}+2\Adv^{\mathsf{uf}}_{N_{\mathsf{tot}}}$, where all honest attestor keys generated in the experiment form the multi-user forgery game ($K=N_{\mathsf{tot}}$). Otherwise some honest attestor signed $M_e$, so by the honest-signing rule $M_e$ is final in its view, extends $M_{e-1}$ and satisfies (i)--(iv). Outside $\mathsf{Fin}$, $M_e$ is the honest epoch-$e$ checkpoint. Since the vault computed $M_e$ from the header it stores, its $\omega_e$, $\Ratt^{e+1}$, $W^{e+1}$ and $N^{e+1}$ are the honest ones, which gives $I_e$(a) and (b). For (c), $\mathsf{withdraw}$ pays only a leaf that opens under a finalized $\omega_{e'}$, and it marks the nullifier spent before transferring, so a leaf is paid at most once even across checkpoints and reentrant calls. A union bound over the $E$ epochs gives the first term.

Under $I_E$, every paid leaf is a leaf of an honest checkpoint. By rules~(ii) and~(iii), it stems from a withdrawal executed on the honest sidechain, which debited its amount from its owner's balance and carries the owner's ML-DSA signature on the full operation, including recipient, amount, nonce and vault. Producing such a signature without the owner is a forgery against one of the $U$ user keys. Deposits are credited once, and transfers and withdrawals conserve balances, so the total paid never exceeds the total credited deposits, which never exceed the escrow. The binding registry is not on the release path. Hence G1 can fail only through the stated events.
\end{proof}

\subsection{Remarks}
The event $\mathsf{Fin}$ is the safety of the sidechain's finality gadget, which we do not re-analyze; execution validity is not part of it, because it is enforced by the honest-signing rule.
\section{Evaluation: Additional Details}\label{S-eval}

This section complements \cref{sec:eval} of the main paper.

\subsection{Artifacts}\label{S-eval-artifacts}
(i)~\texttt{pqcert}, a Rust crate implementing \cref{alg:cert} with the QROM rule (\cref{eq:rule}), domain-separated Keccak-256 Merkle trees with compact multiproofs, and 128-bit coin reduction. It uses the \texttt{crystals-dilithium}~1.0.0 crate used by the prototype of~\cite{NIROPoK2025}, which implements round-3 Dilithium rather than FIPS~204, and \texttt{pqcrypto-falcon}. Round-3 Dilithium2 has the same key and signature sizes as ML-DSA-44, and Dilithium3/5 signatures are $16$ and $32$ bytes shorter than ML-DSA-65/87; the algorithms are not interoperable. Results produced with this crate are labeled Dilithium in the main paper. A FIPS~204 backend (RustCrypto \texttt{ml-dsa}~0.1.1) cross-checks the headline results. (ii)~A Python reference implementation (\texttt{pycert}) with the same logic over SHA3-256, using the FIPS~204 ML-DSA bindings of \texttt{pqcrypto}. (iii)~A multi-process network harness. (iv)~Monte Carlo simulators for QPoS-S and for grinding. (v)~The Solidity vault of \cref{sec:gas}, tested with Foundry against FIPS~204 signatures. Code, raw results, plotting scripts and reproduction scripts are available at \url{https://github.com/hassankhodaiemehr/PQ-Vault}~\cite{pqvault_artifact}; all experiments use fixed seeds.

\subsection{Machines}\label{S-eval-machines}
\emph{M1}: a cloud VM with 2 vCPUs (Intel Xeon, 2.8\,GHz) and 8\,GB RAM. \emph{M2}: a laptop VM with 2 vCPUs on an Intel Core i9-13900HX and 3\,GB RAM, used as a heterogeneous light client. \emph{M3}: 2 vCPUs (Intel Xeon, 2.1\,GHz), used for the FIPS~204 cross-check. All timings are single-threaded medians.

\subsection{Harness Notes}\label{S-eval-harness}
\noindent\textbf{Negative tests.} Every benchmark run includes negative tests: a tampered signature, a shifted interval, a dropped reveal, and a claimed signed weight above $W$. The verifier rejected all of them in every run. The checks $|P|\le\kr$ and ``no unused reveal'' were added after the size runs; unit tests and the FIPS~204 cross-check exercise them.

\noindent\textbf{Proven threshold in the size runs.} The size benchmarks used $\Wp=\lfloor 0.3333\,W\rfloor$, which yields the same $\kr$ as $\lfloor W/3\rfloor$ in all $40$ reported configurations; the contract and network experiments use $\lfloor W/3\rfloor$. The harness now uses $\lfloor W/3\rfloor$ too; this changes the coins, so a rerun gives slightly different reveal sets (e.g., $81$ instead of $94$ distinct reveals at $N=2^8$), and all reported sizes are from the $0.3333$ runs.

\subsection{Network Experiment Setup}\label{S-eval-network}
An aggregator process and two worker processes, each hosting about $N/2$ validators, communicate over real TCP sockets, with one connection per validator and real FIPS~204 ML-DSA-44 signatures. The aggregator verifies every signature, builds the certificate, and sends it to $8$ light clients, which verify it. WAN conditions are emulated per message. Validators are placed in five regions with one-way latency to the aggregator of $1$, $35$, $40$, $60$ and $95$\,ms (with weights $0.30$, $0.15$, $0.30$, $0.05$, $0.20$), with $10\%$ log-normal jitter, $50$\,Mbps validator uplinks, and a $1$\,Gbps aggregator link. $10\%$ of the validators are offline. Because all validators share two vCPUs, signing is charged as the measured single-core signing time of round-3 Dilithium2 in \texttt{pqcert} ($0.42$\,ms), while the signatures themselves are real FIPS~204 signatures; all aggregator and light-client work runs for real. The aggregator stops collecting at $80\%$ of the stake.

\subsection{Contract Implementation and Tests}\label{S-eval-contract}
Gas was measured with Foundry~\cite{foundry} (forge~1.5.1, solc~0.8.30, Osaka EVM rules). Each attestor registers its expanded ML-DSA-44 key once as an SSTORE2 contract, and its attestor leaf commits to the \texttt{EXTCODEHASH} of that contract, which the vault checks during a dispute. The test vectors are a real certificate over an attestor list of $N=2^{16}$ positions ($d=16$), bound in $M_e$ as $N^{e+1}$. Of these, $1024$ carry positive weight and the rest have weight zero, which $\mathsf{Local}$ never accepts because it requires $w_{i_j}>0$; the direct and local verifiers use the same $N$ and $d$. The vectors include $872$ FIPS~204 signatures produced with the ZKNOX reference signer (in deterministic mode, for reproducible vectors), $\kr=193$ coins and $130$ distinct reveals. Both verifiers accept these FIPS~204 signatures. Before signing $\Ratt^{e+1}$, honest validators check that each registered expanded key is the expansion of the committed $pk_i$, and the on-chain check therefore equals $\mldsa.\mathsf{Vf}(pk_i,\cdot)$. The $17$ tests per verifier also confirm the state machine: a record that fails $\mathsf{Local}$ rejects the proposal; a tampered record, a wrong path and a wrong key address supplied by a third party revert without rejecting; finalization waits for open challenges, and a late unanswered challenge rejects; an open-all request is either satisfied or rejects; a rejected epoch can be proposed again; a substituted $\omega_e$ is rejected; the record count must equal $\kr$ exactly; and a withdrawal is paid only once. The packed verifier of the ZKNOX experimental branch uses an external Keccak-f[1600] helper contract.

\subsection{Gas Accounting and Extrapolation}\label{S-eval-gas}
\Cref{tab:gas} reports transaction-level gas under EIP-7623~\cite{eip7623}, $21{,}000+\max(G_{\mathsf{exec}}+4T,\,10T)$, where $G_{\mathsf{exec}}$ is the execution gas measured by the test and $T=n_0+4n_1$ counts the zero and nonzero calldata bytes. Gas is measured inside Foundry tests, where storage written earlier in a test is warm. All operations fit in one transaction except direct verification, which exceeds the $2^{24}$-gas cap~\cite{eip7825}, and the open-all posting, which is split to respect the $128$\,KB mempool limit. The key-registration total assumes nonzero calldata bytes. The two verifiers are the NIST-conformant verifier of the ZKNOX main branch and the packed verifier of its experimental branch. The propose and finalize rows of \cref{tab:gas} sum to $4.2\times10^{5}$, and the single happy-path test measures $4.1\times10^{5}$ with either verifier. For direct verification, dividing the exact costs ($1.08\times10^{9}$ and $1.79\times10^{8}$) by the $2^{24}$ cap gives $64.2$ and $10.6$, hence at least $65$ and $11$ transactions. At $N=2^{16}$ with $189$ distinct reveals (\cref{tab:cert}), the measured cost per distinct reveal ($8.3\times10^{6}$ NIST, $1.4\times10^{6}$ packed) extrapolates to $1.6\times10^{9}$ and $2.6\times10^{8}$ gas. An open-all posting uses a different serialization from the compact certificate: each of the $193$ records carries its signature and two single-leaf paths ($4{,}160$ bytes in ABI encoding), for $803{,}580$ bytes in $7$ transactions under the $128$\,KB limit. Its $3.0\times10^{7}$ gas is set almost entirely by the calldata floor; execution adds only $1.6$--$2.0\times10^{6}$.
\section{Extended Discussion}\label{S-discussion}

This section gives the full version of the discussion in \cref{sec:discussion} of the main paper.

\subsection{Scope of the L1 Assumption}\label{S-disc-l1}
PQ-Vault protects assets against the failure of \emph{account-level} cryptography. A CRQC that breaks ECDSA also breaks the BLS signatures of Ethereum's consensus, and \cref{asm:l1} is therefore not implied by our threat model: it assumes that L1 consensus is not attacked, or is restored by a hard fork that preserves contract state~\cite{buterin2024emergency}. Such an attack is publicly visible and would meet a social response, and after such a fork vault balances remain spendable only through PQ predicates, whereas EOAs need special recovery. For the same reason, validators credit deposits from full L1 nodes, not from BLS-based light clients. \Cref{asm:l1} also bounds censorship: $T_{\mathsf{ch}}$ must exceed the time for which block producers can be bribed to exclude a challenge, which motivates our seven-day window. \Cref{asm:slow} is needed only for transactions that expose an ECDSA key: deposits, initial bindings, and the bonds of relayers and watchers. Private submission channels reduce the exposure. PQ smart accounts for relayers and watchers would remove it.

\subsection{Data Availability and Watcher Incentives}\label{S-disc-da}
\Cref{thm:opf} assumes that one honest watcher is online. Withholding blob data cannot cause a false acceptance (\cref{sec:opf}). Bonds follow from \cref{tab:gas} for a gas-price ceiling $p^{\max}$: $b_C\ge 8.2\times10^{6}p^{\max}$ covers a response, and $b_R\ge b_C$ plus a reward. Open-all requests carry no bond, so withholding cannot drain honest watchers: per proposal, a watcher spends at most one request fee and one bond $b_C$ (\cref{asm:watcher}), which it recovers with a reward when it exposes a failing record. The price is that anyone can force each proposal to post its records once, which costs the relayer $3.0\times10^{7}$ gas against about $2.4\times10^{4}$ for the request; relayers must recover this cost through fees, and G4 relies on their doing so. The prototype implements the vault and dispute logic but not the binding registry or the canary, and its reveal rule targets $\lambda=128$ per checkpoint rather than a lifetime budget (\cref{sec:rule}). Bonds deter griefing, not theft; safety rests on \cref{asm:watcher}, and any depositor can watch at about $20$\,ms per checkpoint ($18$\,ms with round-3 Dilithium2, $21$\,ms with FIPS~204). A full incentive analysis, including relayer compensation, is left open. Blob (KZG) commitments must never be on the acceptance path (\cref{prop:snark}).

\subsection{Public Leader Schedule}\label{S-disc-leader}
QPoS-S fixes and publishes each epoch's leaders when the seed is final. This makes the quantum proof of \cref{thm:fresh} possible: freshness rests on one hidden honest secret, and binding is statistical in the QROM. It also reveals each epoch's leaders as soon as the epoch begins, and an adversary with late reveals sees all its candidate schedules. QPoS-S limits both effects. With $L=3$ fallback proposers per slot, an adversary that suppresses $B$ honest proposers per slot after seeing the schedule removes the honest proposer only from slots that are not guaranteed: with probability $0.3$ per slot at $\alpha=0.3$ without denial of service and $0.514$ with $B=1$, so $n=31$ and $n=56$ slots suffice for inclusion except with probability $2^{-40}+2E(\epsilon_{\mathsf{LE}}+\delta_{\mathsf{LE}})$ (\cref{thm:schedule,cor:real}). Without fallback proposers, a single suppressed leader per slot would halt the sidechain. The stake-weighted random reveal order, keyed by the previous seed, makes late adversarial reveals rare ($\Pr[c\ge x]\le\mu\alpha^x$, \cref{prop:late}). Without denial of service, the target-slot probability at $\alpha=0.3$ is then at most $0.622$, against $0.997$ for an adversary that owns the last four reveals, but still above the $0.348$ of the single-slot runner-up attack on hash-chain PoS; for $\alpha\le0.2$ it is at most $0.265$. Inclusion of honest reveals is proved (\cref{lem:incl}), given that the finality gadget finalizes each block of the reveal window within its slot. Instant adaptive corruption, which could seize the last honest revealer once the order is known, is outside our model; corruptions take effect three epochs after they are requested. Hidden schedules need a PQ verifiable random function or secret leader election, whose quantum analysis we leave open. Because PQ-Vault attaches no exclusive role to a slot, a censored or attacked leader delays the sidechain but not certification.

\subsection{Limitations of the Evaluation}\label{S-disc-limits}
The sidechain is analyzed by simulation only. The network experiment emulates WAN latency on shared cores rather than running on geo-distributed hosts, with five rounds per configuration, and we did not inject relayer or data-withholding faults end to end. If the corrupt third abstains and $85\%$ of honest stake signs, $\Ws/\Wp=1.7$ and $\kr=340$, which enlarges certificates but leaves the optimistic L1 path unchanged. Gas was measured in a local Foundry EVM, not on a public network, and the hash-based baseline was composed from measured components rather than run end to end. Our sidechain relies on a standard BFT finality gadget that we did not re-analyze, and the leader-election bounds assume that the gadget finalizes each block of the reveal window within its slot; inclusion of honest reveals is then proved (\cref{lem:incl}). All our bounds hold in the QROM with ML-DSA's internal hash modeled as an independent oracle. A fully succinct PQ acceptance path would prove certificate verification inside a transparent, hash-based proof system; the prototype circuit of~\cite{NIROPoK2025} for the Dilithium key relation ($8{,}192$ multiplication gates, $140$\,KB proof, $0.7$\,s proving) indicates that this is feasible but costly today. These are natural targets for future work, together with a formally verified vault contract. Our analysis applies to any EUF-CMA signature scheme. Falcon-512 halves certificate size (\cref{sec:baselines}) but needs floating-point Gaussian sampling, which is hard to implement in constant time, and SLH-DSA signatures are much larger. Combining certificates with synchronized hash-based multi-signatures~\cite{drake2025hashsig} is a promising direction.

\ifCLASSOPTIONcompsoc
\section*{Acknowledgments}
\else
\section*{Acknowledgment}
\fi
We thank Saeid Yazdinejad for his simulation work on the design of~\cite{myeprint}.

\bibliographystyle{IEEEtran}
\bibliography{refs}

\end{document}